\documentclass[12pt,a4paper]{article}

\usepackage{amsmath,amssymb}
\usepackage{bm}
\usepackage{amsfonts}

\usepackage{adjustbox}
\usepackage{wrapfig}
\usepackage{subcaption}
\usepackage{float}
\usepackage{booktabs}
\usepackage{array}

\usepackage[utf8]{inputenc}
\usepackage[english]{babel}
\usepackage[margin=1in]{geometry}
\usepackage{graphicx}
\usepackage{float}
\usepackage{amsthm}
\usepackage{amsmath}
\usepackage{mathtools}
\usepackage{amsfonts}
\usepackage{amssymb}
\usepackage{bbm}
\usepackage{bm}
\usepackage{tikz}
\usepackage{algpseudocode}
\usepackage{breqn}
\usepackage[hyphens]{url}
\usepackage{nicefrac}
\usepackage[multiple]{footmisc}
 \usepackage{hyperref}
\usepackage{enumitem}
\usepackage{thmtools}
\usepackage{thm-restate}
\usepackage[nameinlink, capitalize]{cleveref}

\usepackage{cprotect}
\usepackage{etoc}
\usepackage{titlesec}
\titleformat{\section}
  {\normalfont\fontsize{16}{16}\bfseries}{\thesection}{1em}{}
\titleformat{\subsection}
  {\normalfont\fontsize{14}{16}\bfseries}{\thesubsection}{1em}{}
\usepackage{minitoc}
\usepackage{etaremune}

\usepackage[T1]{fontenc}    
\usepackage{booktabs}       
\usepackage{microtype}      
\floatstyle{ruled}
\usepackage{xcolor}

\usepackage{caption}
\usepackage{subcaption}

\usetikzlibrary{arrows,automata, calc, positioning}
\usepackage{adjustbox}
\usepackage{wrapfig}

\usepackage{multirow}
\usepackage{accents}

\usepackage{etoolbox}
\usepackage{xpatch}

\newcommand{\indep}{\raisebox{0.05em}{\rotatebox[origin=c]{90}{$\models$}}}

\DeclareRobustCommand{\VAN}[2]{#1}  

\newcommand{\PreserveBackslash}[1]{\let\temp=\\#1\let\\=\temp}
\newcolumntype{C}[1]{>{\PreserveBackslash\centering}p{#1}}
\newcolumntype{R}[1]{>{\PreserveBackslash\raggedleft}p{#1}}
\newcolumntype{L}[1]{>{\PreserveBackslash\raggedright}p{#1}}

\DeclareMathOperator{\1}{\mathbbm 1}

\DeclareMathOperator*{\argmin}{\arg\!\min}
\DeclareMathOperator*{\argmax}{\arg\!\max}

\theoremstyle{definition}

\newtheorem{theorem}{Theorem}[section]
\newtheorem{proposition}{Proposition}[section]
\newtheorem{definition}{Definition}[section]
\newtheorem{lemma}{Lemma}[section]

\newtheorem{algo}{Algorithm}[section]
\newtheorem{corollary}{Corollary}[section]
\newtheorem{assumption}{Assumption}[section]

\newtheorem{remark}{Remark}[section]

\let\originalleft\left
\let\originalright\right
\renewcommand{\left}{\mathopen{}\mathclose\bgroup\originalleft}
\renewcommand{\right}{\aftergroup\egroup\originalright}
\DeclareRobustCommand{\VAN}[2]{#1}  

\begin{document}

\def\spacingset#1{\renewcommand{\baselinestretch}%
{#1}\small\normalsize} \spacingset{1}

  \title{Sequential kernel embedding for mediated and time-varying dose response curves}

\author{
  Rahul Singh\thanks{Email: \url{rahul_singh@fas.harvard.edu}. Address: Littauer Center 123, 1805 Cambridge Street, Cambridge, MA 02138. The authors would like to thank Neil Shephard for helpful comments. William Liu, Miriam Nelson, and Ikenna Ogbogu provided excellent research assistance. The first author was supported by the Hausman Dissertation Fellowship and Simons-Berkeley Research Fellowship. 
The second and third authors were supported by the Gatsby
Charitable Foundation, GAT 3850 and GAT 3528.
  } \\
  Harvard University
  \and
  Liyuan Xu \\ University College London
  \and
  Arthur Gretton \\ University College London
}
\date{Original draft: October 2020. This draft: March 2025.}
\maketitle

\begin{abstract}
We propose simple nonparametric estimators for mediated and time-varying dose response curves based on kernel ridge regression. 
By embedding Pearl's mediation formula and Robins' g-formula with kernels, we allow treatments, mediators, and covariates to be continuous in general spaces, and also allow for nonlinear treatment-confounder feedback. 
Our key innovation is a reproducing kernel Hilbert space technique called sequential kernel embedding, which we use to construct simple estimators that account for complex feedback.
Our estimators preserve the generality of classic identification while also achieving nonasymptotic uniform rates.
In nonlinear simulations with many covariates, we demonstrate strong performance. 
We estimate mediated and time-varying dose response curves of the US Job Corps, and clean data that may serve as a benchmark in future work.
We extend our results to mediated and time-varying treatment effects and counterfactual distributions, verifying semiparametric efficiency and weak convergence.
\end{abstract}

\noindent%
\vfill


\newpage

\spacingset{1.65} 

\section{Introduction}\label{sec:intro}

We study mediation analysis and time-varying treatment effects with possibly continuous treatments. 
Mediation analysis asks, how much of the total effect of the treatment $D$ on the outcome $Y$ is mediated by a particular mechanism $M$ that takes place between the treatment and outcome? Time-varying analysis asks, what would be the effect of a sequence of treatments $D_{1:T}$ on the outcome $Y$, even when that sequence may not have been implemented? We consider nonparametric causal functions of continuous treatments. For example, the time-varying dose response curve of two continuous treatments is the function $\theta_0^{GF}(d_1,d_2):=E\{Y^{(d_1,d_2)}\}$, which may refer to medical dosages, lifestyle habits, occupational exposures, or training durations.

The time-varying dose response curve arises when evaluating social programs from several rounds of surveys. For example, the National Job Corps Study randomized access to a large scale job training program in the US and collected several rounds of surveys \cite{schochet2008does}. Individuals could decide whether to participate and for how many hours, possibly over multiple years. A natural question is: how much would the average individual benefit from a certain number of class hours in year one and a possibly different number of class hours in year two? This quantity is an example of a time-varying dose response curve, where the number of class hours is the time-varying dose, and the expected benefit is the response. 

The difficulty in estimating time-varying dose response curves is the complex feedback loop resulting from the initial dose. Formally, we model class hours in different years as a sequence of continuous treatments subject to treatment-confounder feedback. In other words, we consider the possibility that class hours in one year may affect health behaviors such as drug use in a subsequent year, which may then affect subsequent class hours. Though several rounds of Job Corps surveys were collected, economists typically study only the initial survey due to a concern for treatment-confounder feedback, and a lack of simple yet flexible estimators that can adjust for it, while allowing treatments to be continuous. We propose such an estimator using kernels, establish its properties, share a cleaned data set that includes the additional surveys, and carry out empirical analysis on the cleaned data as well as on simulated data. Previous methods using kernels for continuous treatments \cite{singh2020kernel} do not handle treatment-confounder feedback, and therefore cannot analyze the additional surveys. This paper's technical innovation is a sequential kernel embedding to do so.

While mediated and time-varying causal functions are identified in theory, they are challenging to estimate in practice. For example, under standard assumptions on time-varying covariates $X_{1:T}$, the time-varying dose response is identified as the g-formula, i.e. the sequential integral $
\theta_0^{GF}(d_1,d_2)=\int \gamma_0(d_1,d_2,x_1,x_2) \mathrm{d}P(x_2|d_1,x_1)\mathrm{d}P(x_1) 
$, where $\gamma_0(d_1,d_2,x_1,x_2)=E(Y|D_1=d_1,D_2=d_2,X_1=x_1,X_2=x_2)$ \cite{robins1986new}. When treatments are continuous, the functional 
$ \gamma $ $ \mapsto $ $ \int $ $ \gamma(d_1,d_2,x_1,x_2) $ $\mathrm{d}P(x_2|d_1,x_1)$ $\mathrm{d}P(x_1) 
$ is generally not bounded \cite{van1991differentiable,newey1994asymptotic} or pathwise differentiable \cite[ch. 3 and 5]{bickel1993efficient}. Popular estimators restrict attention to a binary treatment, parametric models, Markov simplifications, or constrained effect modification for tractability, and may even redefine the estimand. See \cite{vansteelandt2014structural} for a review. Each of these restrictions simplifies the sequential integral in order to simplify estimation. Our research question is: can we devise simple machine learning estimators for causal functions that preserve the richness of the sequential integral, and therefore the generality of treatment-confounder feedback in classic identification, while also achieving nonasymptotic uniform rates?

In this paper, we match the generality of mediated and time-varying identification with the flexibility and simplicity of kernel ridge regression estimation. We propose a new family of nonparametric estimators for causal inference over short horizons. Our algorithms combine kernel ridge regressions, so they inherit the practical and theoretical virtues that make kernel ridge regression widely used. 
Crucially, we preserve the nonlinearity, dependence, and effect modification of identification theory with time-varying confounders. Our contribution has three aspects.

First, we introduce an algorithmic technique that appears to be an innovation in the reproducing kernel Hilbert space (RKHS) literature: sequential kernel embeddings, i.e. RKHS representations of mediator and covariate conditional distributions given a hypothetical treatment sequence, which account for treatment-confounder feedback. For example, we introduce the sequential embedding $\mu_{x_1,x_2}(d_1)$ such that the inner product $\langle f,\mu_{x_1,x_2}(d_1) \rangle$ in an appropriately defined Hilbert space equals the g-formula's sequential integral $\int f(x_1,x_2) \mathrm{d}P(x_2|d_1,x_1)\mathrm{d}P(x_1)$. We prove that the sequential kernel embedding exists because the RKHS restores boundedness of the g-formula's functional $\gamma\mapsto \int \gamma(d_1,d_2,x_1,x_2)$ $ \mathrm{d}P(x_2|d_1,x_1) \mathrm{d}P(x_1) 
$, even when the treatments are continuous. 

Second, we use our new technique to derive estimators with simple closed forms that combine kernel ridge regressions, extending the regression product \cite{baron1986moderator} and recursive regression \cite{bang2005doubly} insights to machine learning. We use sequential embeddings to propose uniformly consistent machine learning estimators of time-varying dose response curves without restrictive linearity, Markov, or no-effect-modification assumptions, which to our knowledge is new. As extensions, we propose what may be the first unrestricted incremental response curves and counterfactual distributions for time-varying treatments, relaxing the restrictions of the structural nested distribution model \cite{robins1992estimation}. In Section~9 of \cite{rahul2024supplement}, for discrete treatments, we use sequential embeddings to propose simpler nuisance parameter estimators for known inferential procedures. In particular, we avoid multiple levels of sample splitting and iterative fitting.

Third, we prove that our simple estimators based on sequential embedding achieve nonasymptotic uniform rates for causal functions. Specifically, for the continuous treatment case, we prove uniform consistency with finite sample rates that combine minimax rates for kernel ridge regression \cite{caponnetto2007optimal,fischer2017sobolev,li2022optimal}. The rates do not directly depend on the data dimension, but rather smoothness and effective dimension, generalizing Sobolev rates. We extend these results to incremental response curves and counterfactual outcome distributions. We relate our results to semiparametric analysis in Section~9 of \cite{rahul2024supplement}. In particular, for the discrete treatment case, we verify $n^{-1/2}$ Gaussian approximation and semiparametric efficiency, articulating a double spectral robustness whereby some kernels may have higher effective dimensions as long as others have sufficiently low effective dimensions. 

We illustrate the practicality of our approach by conducting comparative simulations and estimating the mediated and time-varying response curves of the Jobs Corps. In nonlinear simulations over short horizons, the algorithms reliably outperform some state of the art alternatives. Under standard identifying assumptions, our direct and indirect dose response curve estimates suggest that job training reduces arrests via social mechanisms besides employment. By allowing for continuous treatments and treatment-confounder feedback, our time-varying dose response curve estimates suggest that relatively few class hours in the first and second year confer most of the benefit to counterfactual employment. Of independent interest, we clean and share a version of the Job Corps data that may serve as a benchmark for new approaches to time-varying estimation.

\section{Related work}\label{sec:related}

In seminal works, \cite{robins1992identifiability,pearl2001direct,imai2010identification} and \cite{robins1986new} rigorously develop identification theory for mediation analysis and time-varying treatment effects, respectively. A rich class of mediated and time-varying causal functions are estimable in principle if the analyst has a sufficiently rich set of covariates over time $X_{1:T}$. Treatments may be continuous; relationships among the outcome, treatments, and covariates may be nonlinear; and dependences may include treatment-confounder feedback and effect modification by time-varying confounders \cite{gill2001causal,vanderweele2009conceptual}.

For continuous treatments, nonparametric estimators for mediated response curves of \cite{huber2018direct,ghassami2021multiply} use density estimation, which can be challenging as dimension increases. Machine learning estimators for time-varying dose response curves of \cite{lewis2020double} rely on restrictive linearity, Markov, and no-effect-modification assumptions, which imply additive effects of time-varying treatments. 

For binary treatments, a rich literature provides abstract conditions for $n^{-1/2}$ semiparametric estimation, to which we relate our results in Section~9 of \cite{rahul2024supplement}; see e.g.
\cite{scharfstein1999adjusting,van2006targeted,zheng2011cross,van2012targeted,tchetgen2012semiparametric,petersen2014targeted,molina2017multiple,luedtke2017sequential,rotnitzky2017multiply,chernozhukov2018original,farbmacher2020causal,bodory2021evaluating,singh2021finite} and references therein. Still, estimators that preserve the full generality of identification for binary treatments are not widely used in empirical research \cite{vansteelandt2014structural}, perhaps due to the complexity of nuisance parameter estimation.

Unlike previous work that incorporates the RKHS into causal inference, we provide a framework for mediated and time-varying estimands. Previous work incorporates the RKHS into time-fixed causal inference.  \cite{nie2017quasi,foster2019orthogonal,kennedy2020optimal} propose methods based on orthogonal loss minimization for heterogeneous treatment effects, and \cite{wong2018kernel,zhao2019covariate,kallus2020generalized,hirshberg2019minimax,singh2021debiased} propose methods based on balancing weights for average treatment effects. \cite{muandet2020counterfactual} propose counterfactual distributions for a binary treatment, while \cite{singh2020kernel} propose dose responses and counterfactual distributions for a continuous treatment. Whereas previous work studies the time-fixed setting, we study sequential settings and prove strong results despite the additional challenges of treatment-confounder feedback and effect modification by time-varying confounders.

This paper subsumes our earlier draft \cite{singh2021kernel}, which subsumed
\cite[Sections B and C]{singh2020workshop}.
\section{RKHS assumptions}\label{sec:rkhs_main}



We summarize RKHS notation, interpretation, and assumptions that we use in this paper. Let $k:\mathcal{W}\times\mathcal{W}\rightarrow \mathbb{R}$ be a function that is continuous, symmetric, and positive definite. We call $k$ the kernel, and we call $\phi:w\mapsto k(w,\cdot)$ the feature map. The kernel is the inner product of features $k(w,w')=\langle \phi(w),\phi(w')\rangle_{\mathcal{H}}$, and we formally define the inner product below. The RKHS is the closure of the span of the features $\{\phi(w)\}_{w\in\mathcal{W}}$. As such, the features are interpretable as the dictionary of basis functions for the RKHS: for $f\in\mathcal{H}$, we have that $f(w)=\langle f,\phi(w) \rangle_{\mathcal{H}}$.



Kernel ridge regression uses the RKHS $\mathcal{H}$ as the hypothesis space in an infinite dimensional optimization problem with a ridge penalty, and it has a well known closed form solution:
\begin{equation}\label{eq:loss}
    \hat{f}:=\argmin_{f\in\mathcal{H}} \frac{1}{n}\sum_{i=1}^n\{Y_i-f(W_i)\}^2 +\lambda \|f\|^2_{\mathcal{H}};\quad \hat{f}(w)=Y^{\top}(K_{WW}+n\lambda I)^{-1}K_{Ww},
\end{equation}
where $K_{WW}\in\mathbb{R}^{n\times n}$ is the kernel matrix with $(i,j)$th entry $k(W_i,W_j)$ and $K_{Ww}\in\mathbb{R}^n$ is the kernel vector with $i$th entry $k(W_i,w)$. To tune the ridge penalty hyperparameter $\lambda$, both generalized cross validation and leave one out cross validation have closed form solutions, and the former is asymptotically optimal \cite{li1986asymptotic}. To analyze the bias and variance of kernel ridge regression, the statistical learning literature places assumptions on the smoothness of $f_0$ and effective dimension of $\mathcal{H}$. Both assumptions describe the kernel's spectrum, which we now define.

Denote by $\mathbb{L}^2_{\nu}(\mathcal{W})$ the space of square integrable functions with respect to the measure $\nu$. Consider the convolution operator 
$
L:\mathbb{L}_{\nu}^2(\mathcal{W})\rightarrow \mathbb{L}_{\nu}^2(\mathcal{W}),\; f\mapsto \int k(\cdot,w)f(w)\mathrm{d}\nu(w)
$. By the spectral theorem, we define its spectrum, 
where $(\eta_j)$ are weakly decreasing eigenvalues and $(\varphi_j)$ are orthonormal eigenfunctions that form a basis of $\mathbb{L}_{\nu}^2(\mathcal{W})$. As such,  $
Lf=\sum_{j=1}^{\infty} \eta_j\langle \varphi_j,f \rangle_{\mathbb{L}^2_{\nu}(\mathcal{W})}\cdot  \varphi_j
$.

\begin{remark}[RKHS versus $\mathbb{L}^2$ inner product]
    To interpret how the RKHS $\mathcal{H}$ compares to $\mathbb{L}^2_{\nu}(\mathcal{W})$, we express both function spaces in terms of the orthonormal basis $(\varphi_j)$. In other words, we present the spectral view of the RKHS. For any $f,g\in \mathbb{L}_{\nu}^2(\mathcal{W})$, write $
f=\sum_{j=1}^{\infty}f_j\varphi_j
$ and $
g=\sum_{j=1}^{\infty}g_j\varphi_j
$. 
Then
\begin{align*}
 \mathbb{L}^2_{\nu}(\mathcal{W})&=\left(f=\sum_{j=1}^{\infty}f_j\varphi_j:\; \sum_{j=1}^{\infty}f_j^2<\infty\right),\quad \langle f,g \rangle_{\mathbb{L}^2_{\nu}(\mathcal{W})}=\sum_{j=1}^{\infty} f_jg_j; \\
\mathcal{H}&=\left(f=\sum_{j=1}^{\infty}f_j\varphi_j:\;\sum_{j=1}^{\infty} \frac{f_j^2}{\eta_j}<\infty\right),\quad \langle f,g \rangle_{\mathcal{H}}=\sum_{j=1}^{\infty} \frac{f_jg_j}{\eta_j}.
\end{align*}
\end{remark}

The space $\mathcal{H}$ is the subset of $\mathbb{L}^2_{\nu}(\mathcal{W})$ for which higher order terms in $(\varphi_j)$ have a smaller contribution, subject to $\nu$ satisfying the conditions of Mercer's theorem \cite{steinwart2012mercer}. Under those conditions, $k(w,w')=\sum_{j=1}^{\infty}\eta_j \varphi_j(w)\varphi_j(w')$; $(\eta_j)$ and $(\varphi_j)$ describe the kernel's spectrum.

An analyst can place smoothness and effective dimension assumptions on the spectral properties of a statistical target $f_0$ estimated in the RKHS $\mathcal{H}$ \cite{caponnetto2007optimal,fischer2017sobolev}. These assumptions are formalized by parameters $(b,c)$:
\begin{equation}\label{eq:prior}
    f_0\in \mathcal{H}^c:=\left(f=\sum_{j=1}^{\infty}f_j\varphi_j:\;\sum_{j=1}^{\infty} \frac{f_j^2}{\eta^c_j}<\infty\right)\subset \mathcal{H},\quad c\in(1,2];\quad  \eta_j \leq C j^{-b},\quad b\geq 1.
\end{equation}
The value $c$ quantifies how well the leading terms in $(\varphi_j)$ approximate $f_0$; a larger value of $c$ corresponds to a smoother target $f_0$. A larger value of $b$ corresponds to a faster rate of spectral decay and therefore a lower effective dimension. Both $(b,c)$ are joint assumptions on the kernel and data distribution. Correct specification implies $c\geq 1$ and a bounded kernel implies $b\geq 1$ \cite[Lemma 10]{fischer2017sobolev}. Minimax optimal rates for regression are governed by $(b,c)$ \cite{caponnetto2007optimal,fischer2017sobolev}, with faster rates corresponding to higher values. In our analysis of causal estimands, we obtain nonparametric rates and semiparametric rate conditions that combine regression rates in terms of $(b,c)$.


Spectral assumptions are easy to interpret in Sobolev spaces \cite{fischer2017sobolev}. Let $\mathcal{W}\subset \mathbb{R}^p$. Denote by $\mathbb{H}_2^s$ the Sobolev space with $s>p/2$ square integrable derivatives, which can be generated by the Mat\`ern kernel. Suppose $\mathcal{H}=\mathbb{H}_2^s$ is chosen as the RKHS for estimation. If $f_0\in \mathbb{H}_2^{s_0}$, then $c=s_0/s$; $c$ quantifies the additional smoothness of $f_0$ relative to $\mathcal{H}$. In this Sobolev space, $b=2s/p>1$. The effective dimension is increasing in the original dimension $p$ and decreasing in the degree of smoothness $s$. The minimax optimal regression rates are
\begin{equation}\label{eq:optimal}
    n^{-\frac{1}{2}\frac{c}{c+1/b}}=n^{-\frac{s_0}{2s_0+p}} \text{ in $\mathbb{L}^2$ norm}, \quad n^{-\frac{1}{2}\frac{c-1}{c+1/b}}=n^{-\frac{s_0-s}{2s_0+p}}\text{ in Sobolev norm},
\end{equation}
and both are achieved by kernel ridge regression with 
$\lambda=n^{-1/(c+1/b)}=n^{-2s/(2s_0+p)}$.


We place five types of assumptions in this paper, generalizing the standard RKHS learning theory assumptions from kernel ridge regression to mediated and time-varying causal inference: identification, RKHS regularity, original space regularity, smoothness, and effective dimension.
%
%
We formally instantiate these assumptions for mediated responses in Section~\ref{sec:mediation}, time-varying responses in Section~\ref{sec:dynamic}, and counterfactual distributions in Section~10 of \cite{rahul2024supplement}. We uncover a double spectral robustness in related semiparametric inferential theory in Section~9 of \cite{rahul2024supplement}: some kernels may have higher effective dimensions, as long as other kernels have lower effective dimensions.
\section{Mediated response curves}\label{sec:mediation}

\subsection{Pearl's mediation formula}

Mediation analysis decomposes the total effect of a treatment $D$ on an outcome $Y$ into the direct effect versus the indirect effect mediated via the mechanism $M$. The problem is sequential since $D$ causes $M$ and $Y$, then $M$ also causes $Y$. We denote the counterfactual mediator $M^{(d)}$ given a hypothetical intervention on the treatment $D=d$. We denote the counterfactual outcome $Y^{(d,m)}$ given a hypothetical intervention on the treatment $D=d$ and the mediator $M=m$. 

\begin{definition}[Pure mediated response curves \cite{robins1992identifiability}]\label{def:mediation}
 
Suppose the treatment $D$ is continuous.
\begin{enumerate}
    \item 
The total response $\theta_0^{TE}(d,d')=E[Y^{\{d',M^{(d')}\}}-Y^{\{d,M^{(d)}\}}]$ is the total effect of a new treatment value $d'$ compared to an old value $d$.
\item 
The indirect response $\theta_0^{IE}(d,d')=E[Y^{\{d',M^{(d')}\}}-Y^{\{d',M^{(d)}\}}]$ is the component of the total response mediated by $M$.
 \item 
The direct response $\theta_0^{DE}(d,d')=E[Y^{\{d',M^{(d)}\}}-Y^{\{d,M^{(d)}\}}]$ is the component of the total response that is not mediated by $M$.
 \item 
The mediated response $\theta_0^{ME}(d,d')=E[Y^{\{d',M^{(d)}\}}]
$ is the counterfactual mean outcome in the thought experiment that the treatment is set at a new value $D=d'$ but the mediator $M$ follows the distribution it would have followed if the treatment were set at its old value $D=d$.
\end{enumerate}
Likewise we define incremental response curves, e.g. $\theta_0^{ME,\nabla}(d,d')=E[\nabla_{d'} Y^{\{d',M^{(d)}\}}]$. 
\end{definition}

\begin{remark}[Interventional mediated response curves]\label{remark:controversy}
    Definition~\ref{def:mediation} considers cross world counterfactuals that involve different treatment values for the potential outcomes and potential mediators \cite{robins1992identifiability,pearl2001direct,imai2010identification}. 
 An alternative paradigm instead considers interventional counterfactuals; see e.g. \cite{robins2010alternative,richardson2013single,robins2022interventionist} and references therein. The alternative view avoids defining potential mediators and instead supposes that the treatment can be decomposed into multiple separable components. Though these two paradigms define different mediated response curves, when they are identified, their identifying formulae coincide as Pearl's mediation formula, which we quote in Lemma~\ref{lemma:id_mediation} as the starting point for our analysis. As such, our estimation results apply to both important paradigms for mediation analysis. 
\end{remark}

$\theta_0^{TE}(d,d')$ generalizes average total effects. The average total effect of a binary treatment is $E[Y^{\{1,M^{(1)}\}}-Y^{\{0,M^{(0)}\}}]$. For a continuous treatment, the function $d\mapsto E[Y^{\{d,M^{(d)}\}}]$ may be infinite dimensional, which makes this problem fully nonparametric.


An analyst may wish to measure how much of the total effect is indirect: how much of the total effect would be achieved by simply intervening on the distribution of the mediator $M$? For example, in Section~\ref{sec:experiments}, we investigate the extent to which employment mediates the effect of job training on arrests.  With a binary treatment, the indirect effect is $E[Y^{\{1,M^{(1)}\}}-Y^{\{1,M^{(0)}\}}]$. In the former term, the mediator follows the counterfactual distribution under the intervention $D=1$, and in the latter, it follows the counterfactual distribution under the intervention $D=0$.

The remaining component of the total effect is the direct effect: if the mediator were held at the original distribution corresponding to $D=d$, what would be the impact of the treatment $D=d'$? For example, in Section~\ref{sec:experiments}, we investigate the effect of job training on arrests holding employment at the original distribution. With a binary treatment, the direct effect is $E[Y^{\{1,M^{(0)}\}}-Y^{\{0,M^{(0)}\}}]$.

The final target parameter is $\theta_0^{ME}(d,d')$. It is useful because $\theta_0^{TE}(d,d')$, $\theta_0^{IE}(d,d')$, and  $\theta_0^{DE}(d,d')$ can be expressed in terms of $\theta_0^{ME}(d,d')$. With a binary treatment, this quantity is a matrix in $\mathbb{R}^{2\times 2}$. With a continuous treatment, it is a surface over $\mathcal{D}\times\mathcal{D}$. 

\begin{proposition}[Convenient expressions]\label{prop:mediation}
Mediated response curves can be expressed in terms of $\theta_0^{ME}(d,d')$:
\begin{enumerate}
    \item $\theta_0^{TE}(d,d')=\theta_0^{DE}(d,d')+\theta_0^{IE}(d,d')$;
    \item $\theta_0^{IE}(d,d')=\theta_0^{ME}(d',d')-\theta_0^{ME}(d,d')$; 
    \item $\theta_0^{DE}(d,d')=\theta_0^{ME}(d,d')-\theta_0^{ME}(d,d)$.
\end{enumerate}

\end{proposition}


In seminal works, \cite{robins1992identifiability,pearl2001direct,imai2010identification} state sufficient conditions under which the mediated response curves can be measured from the outcome $Y$, treatment $D$, mediator $M$, and covariates $X$, which we call selection on observables for mediation. We modestly extend the classic identification result from $\theta_0^{ME}(d,d')$ to its  incremental version $\theta_0^{ME,\nabla}(d,d')$. Let
$
\gamma_0(d,m,x)=E(Y|D=d,M=m,X=x)
$.

\begin{lemma}[Pearl's mediation formula]\label{lemma:id_mediation}
Under selection on observables for mediation, 
\begin{enumerate}
    \item  $
\theta_0^{ME}(d,d')=\int \gamma_0(d',m,x) \mathrm{d}P(m|d,x)\mathrm{d}P(x)
$ \cite{robins1992identifiability,pearl2001direct,imai2010identification} and 
    \item $\theta_0^{ME,\nabla}(d,d')=\int \nabla_{d'} \gamma_0(d',m,x) \mathrm{d}P(m|d,x)\mathrm{d}P(x)$.
\end{enumerate}
\end{lemma}
See Section~12 of \cite{rahul2024supplement} for the identifying assumptions and proof of our extension. Proposition~\ref{prop:mediation} identifies the other quantities in Definition~\ref{def:mediation}. For subsequent analysis, it helps to define $\omega_0(d,d';x)=\int \gamma_0(d',m,x) \mathrm{d}P(m|d,x)$, so that $\theta_0^{ME}(d,d')=\int \omega_0(d,d';x) \mathrm{d}P(x)$.

\begin{remark}[Pearl's mediation formula is unbounded over $\mathbb{L}^2$ when the treatment is continuous]\label{remark:unbounded}
Define the functional $F:\mathbb{L}^2\rightarrow \mathbb{R}$, $\gamma\mapsto \int \gamma(d',m,x) \mathrm{d}P(m|d,x)\mathrm{d}P(x)$. When the treatment is continuous, $F$ is generally unbounded, i.e. there does not exists some $C<\infty$ such that $F(\gamma)\leq C \|\gamma\|_{\mathbb{L}^2}$ for all $\gamma\in\mathbb{L}^2$ \cite{van1991differentiable,newey1994asymptotic}. This technical challenge is well documented in the causal inference literature; see \cite{van2018cv} for references.
\end{remark}

\begin{remark}[Mediational g-formula]
    A rich literature defines mediated response curves in the context of time-varying treatments; see e.g. \cite{vanderweele2017mediation,malinsky2019potential} and references therein. The mediational g-formula synthesizes Pearl's mediation formula in Lemma~\ref{lemma:id_mediation} and Robins' g-formula in Lemma~\ref{lemma:id_planning}. Our framework generalizes to these more complex causal functions, using the techniques in Section~11 of \cite{rahul2024supplement}.
\end{remark}

\subsection{Sequential kernel embedding}

Lemma~\ref{lemma:id_mediation} makes precise how each mediated response curve is identified as a sequential integral of the form $\int \gamma_0(d',m,x)\mathrm{d}Q$ for the distribution $Q=P(m|d,x)P(x)$. Since $x$ appears in $\gamma_0(d',m,x)$, $P(m|d,x)$, and $P(x)$, the sequential integral is coupled and therefore challenging to estimate. We prove that, with the appropriate RKHS construction, the components $\gamma_0(d',m,x)$, $P(m|d,x)$, and $P(x)$ can be decoupled. Moreover, the sequential distribution $Q$ can be encoded by a sequential kernel embedding, which is our key innovation. We use these techniques to reduce sequential causal inference into the combination of kernel ridge regressions, which then allows us to propose simple estimators with closed form solutions.

To begin, we construct the appropriate RKHS for $\gamma_0$. In our construction, we define RKHSs for the treatment $D$, mediator $M$, and covariates $X$, then assume that the regression is an element of a certain composite space. To lighten notation, we will suppress subscripts when arguments are provided. We assume the regression $\gamma_0$ is an element of the RKHS $\mathcal{H}$ with the kernel $k(d,m,x;d',m',x') = k_{\mathcal{D}}(d,d') k_{\mathcal{M}}(m,m') k_{\mathcal{X}}(x,x')$. Formally, this choice of kernel corresponds to the tensor product: $\gamma_0\in\mathcal{H}=\mathcal{H}_{\mathcal{D}}\otimes \mathcal{H}_{\mathcal{M}} \otimes \mathcal{H}_{\mathcal{X}}$, with the tensor product dictionary of basis functions $\phi(d)\otimes \phi(m)\otimes \phi(x)$. As such, $\gamma_0(d,m,x)=\langle \gamma_0, \phi(d)\otimes \phi(m)\otimes \phi(x)\rangle_{\mathcal{H}} $ and $\|\phi(d)\otimes \phi(m)\otimes \phi(x)\|_{\mathcal{H}}=\|\phi(d)\|_{\mathcal{H}_{\mathcal{D}}}  \|\phi(m)\|_{\mathcal{H}_{\mathcal{M}}}  \|\phi(x)\|_{\mathcal{H}_{\mathcal{X}}}$. We place regularity conditions on this RKHS construction in order to prove our decoupling result. Anticipating Section~10 of \cite{rahul2024supplement}, we include conditions for an outcome RKHS in parentheses.

\begin{assumption}[RKHS regularity conditions]\label{assumption:RKHS}
Assume 
\begin{enumerate}
    \item $k_{\mathcal{D}}$, $k_{\mathcal{M}}$, $k_{\mathcal{X}}$ (and $k_{\mathcal{Y}}$) are continuous and bounded, i.e.
    $
    \sup_{d\in\mathcal{D}}\|\phi(d)\|_{\mathcal{H}_{\mathcal{D}}}\leq \kappa_d$, $ \sup_{m\in\mathcal{M}}$ $\|\phi(m)\|_{\mathcal{H}_{\mathcal{M}}}\leq$ $ \kappa_m$, $ \sup_{x\in\mathcal{X}}\|\phi(x)\|_{\mathcal{H}_{\mathcal{X}}}\leq \kappa_x
    $ \{and $ \sup_{y\in\mathcal{Y}}\|\phi(y)\|_{\mathcal{H}_{\mathcal{Y}}}\leq \kappa_y$\}; 
    \item $\phi(d)$, $\phi(m)$, $\phi(x)$ \{and $\phi(y)$\} are measurable;
    \item $k_{\mathcal{M}}$, $k_{\mathcal{X}}$ (and $k_{\mathcal{Y}}$) are characteristic \cite{sriperumbudur2010relation}.
\end{enumerate}

For incremental responses, further assume $\mathcal{D}\subset \mathbb{R}$ is an open set and $\nabla_d  \nabla_{d'} k_{\mathcal{D}}(d,d')$ exists and is continuous, hence $\sup_{d\in\mathcal{D}}\|\nabla_d\phi(d)\|_{\mathcal{H}}\leq \kappa_d'$. 
\end{assumption}
Commonly used kernels are continuous and bounded. Measurability is a similarly weak condition. The characteristic property means that different distributions will have different embeddings in the RKHS. For example, the indicator kernel is characteristic over a discrete domain, while the exponentiated quadratic kernel is characteristic over a continuous domain.

\begin{theorem}[Decoupling via sequential kernel embeddings]\label{theorem:representation_mediation}
Suppose the conditions of Lemma~\ref{lemma:id_mediation} hold. Further suppose Assumption~\ref{assumption:RKHS} holds and $\gamma_0\in\mathcal{H}$. Then
\begin{enumerate}
    \item $\omega_0(d,d';x)=\langle \gamma_0,\phi(d')\otimes \mu_m(d,x) \otimes \phi(x)  \rangle_{\mathcal{H}} $,
    \item $
\theta_0^{ME}(d,d')=\langle \gamma_0, \phi(d')\otimes \mu_{m,x}(d) \rangle_{\mathcal{H}}
$, and 
    \item $
\theta_0^{ME,\nabla}(d,d')=\langle \gamma_0, \nabla_{d'}\phi(d')\otimes \mu_{m,x}(d) \rangle_{\mathcal{H}},
$ 
\end{enumerate}
where $\mu_{m}(d,x)=\int \phi(m)\mathrm{d}P(m|d,x)$ and $\mu_{m,x}(d)=\int \{\mu_m(d,x) \otimes \phi(x)\} \mathrm{d}P(x)$. 
\end{theorem}

\begin{proof}[Proof sketch]
Consider $\omega_0(d,d';x)=\int \gamma_0(d',m,x) \mathrm{d}P(m|d,x)$. We show
$$
\omega_0(d,d';x)
=\int \langle \gamma_0, \phi(d')\otimes \phi(m) \otimes \phi(x) \rangle_{\mathcal{H}}  \mathrm{d}P(m|d,x) 
= \langle \gamma_0, \phi(d')\otimes \mu_m(d,x) \otimes \phi(x) \rangle_{\mathcal{H}}.
$$
Sequentially repeating this technique for $\theta_0^{ME}(d,d')=\int \omega_0(d,d';x) \mathrm{d}P(x)$,
$$
\theta_0^{ME}(d,d')
=\int \langle \gamma_0, \phi(d')\otimes \mu_m(d,x) \otimes \phi(x) \rangle_{\mathcal{H}} \mathrm{d}P(x)
=\langle \gamma_0, \phi(d')\otimes\mu_{m,x}(d) \rangle_{\mathcal{H}}. \qedhere
$$
\end{proof}

See Section~13 of \cite{rahul2024supplement} for the full proof. The quantity $\mu_{m}(d,x)=\int \phi(m)\mathrm{d}P(m|d,x)$ embeds the conditional distribution $P(m|d,x)$ as an element of the RKHS $\mathcal{H}_{\mathcal{M}}$, which is a popular technique in the RKHS literature. It satisfies, for $f\in\mathcal{H}_{\mathcal{M}}$, $\langle f, \mu_{m}(d,x) \rangle_{\mathcal{H}_{\mathcal{M}}}=\int f(m) \mathrm{d}P(m|d,x)$. 

\begin{remark}[Key innovation]
    The quantity $\mu_{m,x}(d)$ is a sequential kernel embedding that encodes the counterfactual distribution of the mediator $M$ and covariates $X$ when the counterfactual treatment value is $D=d$. It is our key innovation. It has the property that, for $f\in\mathcal{H}_{\mathcal{M}} \otimes \mathcal{H}_{\mathcal{X}}$, $\langle f, \mu_{m,x}(d) \rangle_{\mathcal{H}_{\mathcal{M}}\otimes \mathcal{H}_{\mathcal{X}}}=\int f(m,x) \mathrm{d}P(m|d,x)\mathrm{d}P(x)$, implementing Pearl's mediation formula.
\end{remark}


\begin{remark}[Pearl's mediation formula is bounded over $\mathcal{H}$ when the treatment is continuous]\label{remark:bounded}
Define the functional $F:\mathcal{H} \rightarrow \mathbb{R}$, $\gamma\mapsto \int \gamma(d',m,x) \mathrm{d}P(m|d,x)\mathrm{d}P(x)$. Under Assumption~\ref{assumption:RKHS}, we show that $F$ is bounded even when the treatment is continuous, i.e. there exists some $C<\infty$ such that $F(\gamma)\leq C \|\gamma\|_{\mathcal{H}}$ for all $\gamma\in\mathcal{H}$. This observation generalizes the observation of \cite{singh2020kernel} for time-fixed dose response curves. It follows immediately from the definition of the RKHS as the space of functions for which the functional $\gamma\mapsto \gamma(d',m,x)$ is bounded \cite{berlinet2004reproducing}. Boundedness of the functional is what guarantees the existence of the sequential kernel embeddings in
 Theorem~\ref{theorem:representation_mediation}.
\end{remark}

Theorem~\ref{theorem:representation_mediation} decouples $\gamma_0(d',m,x)$, $P(m|d,x)$, and $P(x)$, providing a blueprint for estimation that avoids density estimation and sequential integration. Our estimator will be $\hat{\theta}^{ME}(d,d')=\langle \hat{\gamma}, \phi(d')\otimes \hat{\mu}_{m,x}(d)  \rangle_{\mathcal{H}}$ where $
\hat{\mu}_{m,x}(d)  =n^{-1}\sum_{i=1}^n\{\hat{\mu}_m(d,X_i)\otimes \phi(X_i)\}$. The estimator $\hat{\gamma}$ is a standard kernel ridge regression. The estimator $\hat{\mu}_m(d,x)$ is an appropriately generalized kernel ridge regression. We combine them by averaging and taking the product.

\begin{algo}[Nonparametric estimation of mediated response curves]\label{algorithm:mediation}
Denote the kernel matrices by
$
K_{DD},$ $K_{MM},$ $K_{XX}\in\mathbb{R}^{n\times n}
$. Let $\odot$ be the elementwise product. Mediated response curves have closed form solutions:
\begin{enumerate}
    \item $\hat{\omega}(d,d';x)= Y^{\top}(K_{DD}\odot K_{MM}\odot K_{XX}+n\lambda I)^{-1}
    [K_{Dd'}\odot  \{K_{MM}(K_{DD}\odot K_{XX}+n\lambda_1 I)^{-1}(K_{Dd}\odot K_{Xx})\}\odot K_{Xx}]$ and
    \item  $\hat{\theta}^{ME}(d,d') 
    =n^{-1}\sum_{i=1}^n \hat{\omega}(d,d';X_i)$,
\end{enumerate}
where $(\lambda,\lambda_1)$ are ridge regression penalty parameters.
For mediated incremental response curve estimators, we replace $K_{Dd'}$ with $\nabla_{d'}K_{Dd'}$ where $(\nabla_{d'}  K_{D{d'}})_i=\nabla_{d'} k(D_i,d')$. 
\end{algo}
\begin{proof}[Derivation sketch]
Consider $\omega_0(d,d';x)$. Analogously to~\eqref{eq:loss}, the kernel ridge regression estimators of the regression $\gamma_0$ and the conditional kernel embedding $\mu_m(d,x)$ are 
\begin{align*}
    \hat{\gamma}&=\argmin_{\gamma \in\mathcal{H}} \frac{1}{n}\sum_{i=1}^n [Y_i-\langle\gamma, \phi(D_i)\otimes \phi(M_i) \otimes\phi (X_i)\rangle_{\mathcal{H}}]^2 + \lambda \|\gamma\|^2_{\mathcal{H}}, \\
    \hat{E}&=\argmin_{E\in\mathcal{L}_2(\mathcal{H}_{\mathcal{M}},\mathcal{H}_{\mathcal{D}}\otimes \mathcal{H}_{\mathcal{X}})} \frac{1}{n}\sum_{i=1}^n [\phi(M_i)-E^*\{\phi(D_i)\otimes \phi(X_i)\}]^2 + \lambda_1 \|E\|^2_{\mathcal{L}_2(\mathcal{H}_{\mathcal{M}},\mathcal{H}_{\mathcal{D}}\otimes \mathcal{H}_{\mathcal{X}})},
\end{align*}
where $\hat{\mu}_m(d,x)=\hat{E}^*\{\phi(d)\otimes \phi(x)\}$ and $E^*$ is the adjoint of $E$. The closed forms are
\begin{align*}
    \hat{\gamma}(d',\cdot,x)&=Y^{\top}(K_{DD}\odot K_{MM}\odot K_{XX}+n\lambda I)^{-1}(K_{Dd'}\odot K_{M\cdot }\odot K_{Xx}),\\
    [\hat{\mu}_m(d,x)](\cdot)&=K_{\cdot M}(K_{DD}\odot K_{XX}+n\lambda_1 I)^{-1}(K_{Dd}\odot K_{Xx}).
\end{align*}
To arrive at the main result, match the empty arguments $(\cdot)$ of the kernel ridge regressions.
\end{proof}
We derive this algorithm in Section~13 of \cite{rahul2024supplement}. 
We give theoretical values for $(\lambda,\lambda_1)$ that optimally balance bias and variance in Theorem~\ref{theorem:consistency_mediation} below. Section~16 gives practical tuning procedures with closed form solutions to empirically balance bias and variance, one of which is asymptotically optimal. We formally define the operator space $\mathcal{L}_2(\mathcal{H}_{\mathcal{M}},\mathcal{H}_{\mathcal{D}}\otimes \mathcal{H}_{\mathcal{X}})$ below. 

\subsection{Uniform consistency with finite sample rates}

Towards a guarantee of uniform consistency, we place regularity conditions on the original spaces. Anticipating Section~10 of \cite{rahul2024supplement}, we include conditions for the outcome in parentheses.
\begin{assumption}[Original space regularity conditions]\label{assumption:original}
Assume (i) $\mathcal{D}$, $\mathcal{M}$, $\mathcal{X}$ (and $\mathcal{Y}$) are Polish spaces; 
    (ii) $\mathcal{Y}\subset \mathbb{R}$ and $|Y|\leq C$ almost surely.
\end{assumption}
A Polish space is a separable and completely metrizable topological space. Random variables supported in a Polish space may be discrete or continuous and even texts, graphs, or images. Boundedness of the outcome $Y$ can be relaxed.
Next, we place assumptions on the smoothness of the regression $\gamma_0$ and the effective dimension of $\mathcal{H}$ the sense of~\eqref{eq:prior}.
\begin{assumption}[Smoothness and effective dimension of the regression]\label{assumption:smooth_gamma}
Assume $\gamma_0\in\mathcal{H}^c$ with $c\in(1,2]$, and $\eta_j(\mathcal{H}) \leq C j^{-b}$ with $b\geq 1$.
\end{assumption}
See Section~14 of \cite{rahul2024supplement} for alternative ways of writing and interpreting Assumption~\ref{assumption:smooth_gamma} in the tensor product space $\mathcal{H}$.  We place similar conditions on the conditional kernel embedding $\mu_m(d,x)$,  which is a generalized regression. We articulate this assumption abstractly for the conditional kernel embedding $\mu_{a}(b)=\int \phi(a)\mathrm{d}P(a|b)$ where $a\in\mathcal{A}_{\ell}$ and $b\in\mathcal{B}_{\ell}$. As such, all one has to do is specify $\mathcal{A}_{\ell}$ and $\mathcal{B}_{\ell}$ to specialize the assumption. For $\mu_m(d,x)$, $\mathcal{A}_1=\mathcal{M}$ and $\mathcal{B}_1=\mathcal{D}\times \mathcal{X}$. We parametrize the effective dimension and smoothness of $\mu_{a}(b)$ by $(b_{\ell},c_{\ell})$.

Formally, define the conditional expectation operator $E_{\ell}:\mathcal{H}_{\mathcal{A}_{\ell}}\rightarrow\mathcal{H}_{\mathcal{B}_{\ell}}$, $f(\cdot)\mapsto E\{f(A_{\ell})|B_{\ell}=\cdot\}$. By construction, $E_{\ell}$ encodes the same information as $\mu_{a}(b)$ since
$$
\mu_{a}(b)=\int \phi(a)\mathrm{d}P(a|b) =[E_{\ell}\{\phi(\cdot)\}](b) =[E_{\ell}^* \{\phi(b)\}](\cdot),\quad a\in\mathcal{A}_{\ell},\quad  b\in\mathcal{B}_{\ell},
$$
where $E_{\ell}^*$ is the adjoint of $E_{\ell}$. We denote the space of Hilbert--Schmidt operators between $\mathcal{H}_{\mathcal{A}_{\ell}}$ and $\mathcal{H}_{\mathcal{B}_{\ell}}$ by $\mathcal{L}_2(\mathcal{H}_{\mathcal{A}_{\ell}},\mathcal{H}_{\mathcal{B}_{\ell}})$. Here, $\mathcal{L}_2(\mathcal{H}_{\mathcal{A}_{\ell}},\mathcal{H}_{\mathcal{B}_{\ell}})$ is an RKHS in its own right, for which we place smoothness and effective dimension assumptions in the sense of~\eqref{eq:prior}. See \cite[Appendix B]{grunewalder2013smooth} and \cite[Appendix A.3]{singh2019kernel} for details.

\begin{assumption}[Smoothness and effective dimension of a conditional kernel embedding]\label{assumption:smooth_op}
Assume $E_{\ell}\in \mathcal{L}_2(\mathcal{H}_{\mathcal{A}_{\ell}},\mathcal{H}^{c_{\ell}}_{\mathcal{B}_{\ell}})$ with $c_{\ell}\in(1,2]$, and $\eta_j(\mathcal{H}_{\mathcal{B}_{\ell}})\leq C j^{-b_{\ell}}$ with $b_{\ell}\geq 1$.
\end{assumption}
Just as we place approximation assumptions for $\gamma_0$ in terms of $\mathcal{H}$, which provides the features onto which we project $Y$, we place approximation assumptions for $E_{\ell}$ in terms of $\mathcal{H}_{\mathcal{B}_{\ell}}$, which provides the features $\phi(B_{\ell})$ onto which we project $\phi(A_{\ell})$. Under these conditions, we arrive at our first main theoretical guarantee.

\begin{theorem}[Uniform consistency of mediated response curves]\label{theorem:consistency_mediation}
Suppose the conditions of Theorem~\ref{theorem:representation_mediation} hold, as well as Assumptions~\ref{assumption:original},~\ref{assumption:smooth_gamma}, and~\ref{assumption:smooth_op} with $\mathcal{A}_1=\mathcal{M}$ and $\mathcal{B}_1=\mathcal{D}\times \mathcal{X}$. Set $(\lambda,\lambda_1)=\{n^{-1/(c+1/b)},n^{-1/(c_1+1/b_1)}\}$, which is rate optimal regularization. Then
\begin{enumerate}
    \item $
\|\hat{\theta}^{ME}-\theta_0^{ME}\|_{\infty}=O_p\left[n^{-(c-1)/\{2(c+1/b)\}}+n^{-(c_1-1)/\{2(c_1+1/b_1)\}}\right]
$ and
    \item  $
\|\hat{\theta}^{ME,\nabla}-\theta_0^{ME,\nabla}\|_{\infty}=O_p\big[n^{-(c-1)/\{2(c+1/b)\}}+n^{-(c_1-1)/\{2(c_1+1/b_1)\}}\big].
$
\end{enumerate}
\end{theorem}
See Section~14 of \cite{rahul2024supplement} for the proof. By Proposition~\ref{prop:mediation}, the quantities in Definition~\ref{def:mediation} are uniformly consistent with the same rate, which combines optimal rates for standard \cite[Theorem 2]{fischer2017sobolev} and generalized \cite[Theorem 3]{li2022optimal} kernel ridge regression in RKHS norm. Section~14 gives the exact finite sample rate and the explicit specialization of Assumption~\ref{assumption:smooth_op}. The rate is at best $n^{-1/4}$ when $(c,c_1)=2$ and $(b,b_1)\rightarrow \infty$, i.e. when $(\gamma_0,\mu_m)$ are very smooth with finite effective dimensions. The rates reflect the challenge of a $\sup$ norm guarantee, which is much stronger than an $\mathbb{L}^2$ norm guarantee and is useful for policymakers who may be concerned about each treatment value. See~\eqref{eq:optimal} to specialize these rates for Sobolev spaces.

\begin{remark}[Technical contribution]
The technical contribution underlying our theoretical guarantee is an RKHS norm rate for the sequential kernel embedding. 
In particular, Proposition~14.1 in Section~14 of \cite{rahul2024supplement} derives a nonasymptotic bound on $\sup_{d \in\mathcal{D}}\|\hat{\mu}_{m,x}(d)-\mu_{m,x}(d)\|_{\mathcal{H}_{\mathcal{M}}\otimes \mathcal{H}_{\mathcal{X}}}$, as a stepping stone to Theorem~\ref{theorem:consistency_mediation}.  
This intermediate result is not contained in \cite{singh2020kernel} nor in other previous works on kernel methods for causal inference.
\end{remark}

\begin{remark}[Rate improvements]\label{remark:oracle}
    For the time-fixed dose response curve, \cite{kennedy2017nonparametric} prove pointwise rates, assuming smoothness of the dose response as well as finite uniform entropy integrals for the regression estimator and for the density estimator of treatment given covariates. Later works use sample splitting to relax the entropy conditions to rate conditions, e.g. \cite{semenova2017estimation,colangelo2020double}.

    If the regression and density have sufficiently fast rates, then pointwise rate improvements are possible, reflecting the smoothness and lower dimension of the dose response. These are called oracle rates with second order dependence on the regression and density. See e.g. \cite{nie2017quasi,foster2019orthogonal,kennedy2020optimal} for time-fixed heterogeneous treatment effects.

 For our RKHS estimator of the mediated response curve, we prove uniform rates, under smoothness assumptions on the regression function and a conditional expectation operator.
    To achieve $\sup$ norm rate improvements for our RKHS estimator, future work may
place an additional smoothness assumption on the mediated response curve and rate conditions on the regression, conditional expectation operator, and appropriate conditional densities.
\end{remark}
\section{Time-varying response curves}\label{sec:dynamic}

\subsection{Robins' g-formula}

So far we have considered the effect of a single treatment $D$. Next, we consider the effect of a sequence of time-varying treatments $D_{1:T}=d_{1:T}$ on the counterfactual outcome $Y^{(d_{1:T})}$. If the sequence of treatment values $d_{1:T}$ is observed in the data, this problem may be called on-policy planning; if not, it may be called off-policy planning. 

\begin{definition}[Time-varying dose response curves \cite{robins1986new}]
Suppose the treatments $D_{1:T}$ are continuous. 
\begin{enumerate}
    \item The time-varying response $\theta_0^{GF}(d_{1:T})=E\{Y^{(d_{1:T})}\}$ is the counterfactual mean outcome given the interventions $D_{1:T}=d_{1:T}$ for the entire population.
    \item With distribution shift, $ \theta_0^{DS}(d_{1:T},\tilde{P})=E_{\tilde{P}}\{Y^{(d_{1:T})}\}$ is the counterfactual mean outcome given the interventions $D_{1:T}=d_{1:T}$ for an alternative population with the data distribution $\tilde{P}$.
\end{enumerate}
Likewise we define incremental response curves, e.g. $\theta_0^{GF,\nabla}(d_{1:T})=E\{\nabla_{d_T} Y^{(d_{1:T})}\}$. 
\end{definition}

\begin{remark}[Randomized dynamic policies]\label{remark:dynamic_policy_informal}
For clarity, we focus on the deterministic, static counterfactual policy $d_{1:T}$. It is deterministic in that it is nonrandom. It is static in that it does not depend on the observed sequence of covariates $X_{1:T}$. The time-varying treatment effect literature extends to policies that may be randomized and dynamic \cite{robins1986new}. Our approach extends to randomized and dynamic policies with additional notation; see Remark~\ref{remark:dynamic_policy_formal}.
\end{remark}

Whereas much of the semiparametric literature restricts $d_t$ to be discrete, we allow $d_t$ to be continuous and consider a nonparametric approach to time-varying response curves \cite{gill2001causal}. For example, in Section~\ref{sec:experiments}, we estimate the effect of $d_1$ class hours in year one, and $d_2$ class hours in year two, on counterfactual employment, with treatment-confounder feedback. In the spirit of off-policy planning, we consider a distribution shift from $P$ to $\tilde{P}$. 

In seminal work, \cite{robins1986new} states sufficient conditions under which time-varying responses can be measured from the outcome $Y$, treatments $D_{1:T}$, and covariates $X_{1:T}$. We refer to this collection of conditions as sequential selection on observables. We modestly extend the classic identification result by considering incremental responses.

\begin{lemma}[Robins' g-formula]\label{lemma:id_planning}
Under sequential selection on observables and a distribution shift condition, 
\begin{enumerate}
    \item $\theta_0^{GF}(d_{1:T})=\int \gamma_0(d_{1:T},x_{1:T}) \mathrm{d}P(x_1) \prod_{t=2}^{T} \mathrm{d}P\{x_t|d_{1:(t-1)},x_{1:(t-1)}\}$ \cite{robins1986new} and  
    \item 
 $
\theta_0^{DS}(d_{1:T},\tilde{P})=\int \gamma_0(d_{1:T},x_{1:T}) \mathrm{d}\tilde{P}(x_1)$ $\prod_{t=2}^{T} \mathrm{d}\tilde{P}\{x_t|d_{1:(t-1)},x_{1:(t-1)}\}
$ \cite{pearl2014external}. 
    \item For incremental response curves, we replace $\gamma_0(d_{1:T},x_{1:T})$ with $\nabla_{d_T} \gamma_0(d_{1:T},x_{1:T})$.
\end{enumerate} 
\end{lemma}
See Section~12 of \cite{rahul2024supplement} for the identifying assumptions and proof of our extension. 
We consider a fully nonparametric g-formula with possibly continuous treatments that allows for distribution shift. Lemma~\ref{lemma:id_planning} handles auxiliary Markov restrictions as special cases, e.g. if covariates follow a Markov process, then $\theta_0^{GF}$ simplifies by setting $P\{x_t|d_{1:(t-1)},x_{1:(t-1)}\}=P(x_t|d_{t-1},x_{t-1})$. 

\begin{remark}[Robins' g-formula is unbounded over $\mathbb{L}^2$ when the treatments are continuous]\label{remark:unbounded2}
Define the functional $F:\mathbb{L}^2\rightarrow \mathbb{R}$, $\gamma\mapsto \int \gamma(d_{1:T},x_{1:T}) \mathrm{d}P(x_1) \prod_{t=2}^{T} \mathrm{d}P\{x_t|d_{1:(t-1)},x_{1:(t-1)}\}$. When the treatments are continuous, $F$ is generally unbounded in the sense of Remark~\ref{remark:unbounded} \cite{van1991differentiable,newey1994asymptotic,van2018cv}.
\end{remark}

\begin{remark}[Randomized dynamic policies]\label{remark:dynamic_policy_formal}
    To accommodate randomized dynamic policies, the product in Lemma~\ref{lemma:id_planning} will include factors for the conditional distributions of treatments. Specifically, the integral will replace $d_t$ with $g_t:\{d_{1:(t-1)}, x_{1:t}\}\mapsto G\{d_t|d_{1:(t-1)}, x_{1:t}\}$ where $G$ is the distribution induced by the randomized dynamic policy $(g_t)$.
\end{remark}

\subsection{Sequential kernel embedding}

Similar to Lemma~\ref{lemma:id_mediation}, Lemma~\ref{lemma:id_planning} identifies each time-varying response as a sequential integral of the form $\int \gamma_0(d_{1:T},x_{1:T})\mathrm{d}Q$ for the distribution $Q=P(x_1) \prod_{t=2}^{T} P\{x_t|d_{1:(t-1)},x_{1:(t-1)}\}$ or $Q=\tilde{P}(x_1) \prod_{t=2}^{T} \tilde{P}\{x_t|d_{1:(t-1)},x_{1:(t-1)}\}$. As before, components of the sequential integral are coupled and therefore challenging to estimate, since $x_1$ appears in $\gamma_0(d_{1:T},x_{1:T})$, $P\{x_t|d_{1:(t-1)},x_{1:(t-1)}\}$, and $P(x_1)$. As in Section~\ref{sec:mediation}, we construct an appropriate RKHS to decouple these components, then to encode $Q$ by a sequential kernel embedding. With these techniques, we again reduce sequential causal inference into the combination of kernel ridge regressions. For clarity, we present the algorithm with $T=2$, and we define $\omega_0(d_1,d_2;x_1)=\int \gamma_0(d_1,d_2,x_1,x_2)\mathrm{d}P(x_2|d_1,x_1)$ so that $\theta_0^{GF}(d_1,d_2)=\int \omega_0(d_1,d_2;x_1)\mathrm{d}P(x_1)$. We consider $T>2$ in Section~11 of \cite{rahul2024supplement}, which also showcases the role of Markov assumptions.

To construct the RKHS for $\gamma_0$, we define RKHSs for each treatment $D_t$ and each covariate $X_t$. Using identical notation as Section~\ref{sec:mediation}, we assume the regression $\gamma_0$ is an element of the RKHS $\mathcal{H}$ with the kernel $k(d_1, d_2, x_1, x_2; d_1', d_2', x_1', x_2') = k_{\mathcal{D}}(d_1, d_1') k_{\mathcal{D}}(d_2, d_2') k_{\mathcal{X}}(x_1, x_1') k_{\mathcal{X}}(x_2, x_2')$, i.e. $\gamma_0\in\mathcal{H}=\mathcal{H}_{\mathcal{D}}\otimes\mathcal{H}_{\mathcal{D}}\otimes  \mathcal{H}_{\mathcal{X}}  \otimes \mathcal{H}_{\mathcal{X}}$. As such, $\gamma_0(d_1,d_2,x_1,x_2)=\langle \gamma_0, \phi(d_1)\otimes\phi(d_2) \otimes \phi(x_1)\otimes \phi(x_2)\rangle_{\mathcal{H}} $ and $\|\phi(d_1)\otimes \phi(d_2)\otimes \phi(x_1)\otimes \phi(x_2)\|_{\mathcal{H}}=\|\phi(d_1)\|_{\mathcal{H}_{\mathcal{D}}}  \|\phi(d_2)\|_{\mathcal{H}_{\mathcal{D}}}  \|\phi(x_1)\|_{\mathcal{H}_{\mathcal{X}}}  \|\phi(x_2)\|_{\mathcal{H}_{\mathcal{X}}}$. Under regularity conditions on this RKHS construction, we prove an analogous decoupling result.
\begin{theorem}[Decoupling via sequential kernel embeddings]\label{theorem:representation_planning}
Suppose the conditions of Lemma~\ref{lemma:id_planning} hold. Further suppose Assumption~\ref{assumption:RKHS} holds and $\gamma_0\in\mathcal{H}$. 
Then
\begin{enumerate}
    \item $\omega_0(d_1,d_2;x_1)=\langle \gamma_0,\phi(d_1)\otimes\phi(d_2) \otimes \phi(x_1)\otimes \mu_{x_2}(d_1,x_1) \rangle_{\mathcal{H}}$;
    \item $\theta_0^{GF}(d_1,d_2)=\langle \gamma_0,  \phi(d_1)\otimes\phi(d_2) \otimes 
    \mu_{x_1,x_2}(d_1)\rangle_{\mathcal{H}}$ where $\mu_{x_2}(d_1,x_1)=\int \phi(x_2)\mathrm{d}P(x_2|d_1,x_1)$ and $\mu_{x_1,x_2}(d_1)=\int \{\phi(x_1)\otimes \mu_{x_2}(d_1,x_1)\}\mathrm{d} P(x_1)$; 
    \item $\theta_0^{DS}(d_1,d_2;\tilde{P})=\langle \gamma_0,  \phi(d_1)\otimes\phi(d_2) \otimes \nu_{x_1,x_2}(d_1) \rangle_{\mathcal{H}} $ where $\nu_{x_2}(d_1,x_1)=\int \phi(x_2)\mathrm{d}\tilde{P}(x_2|d_1,x_1)$ and $\nu_{x_1,x_2}(d_1)=\int \{\phi(x_1)\otimes \nu_{x_2}(d_1,x_1)\}\mathrm{d} \tilde{P}(x_1)$. 
\end{enumerate}
For incremental responses, we replace $\phi(d_2)$ with $\nabla_{d_2} \phi(d_2)$. 
\end{theorem}
See Section~13 of \cite{rahul2024supplement} for the proof. In $\theta_0^{GF}$, the conditional kernel embedding of $P(x_2|d_1,x_1)$ is $\mu_{x_2}(d_1,x_1)=\int \phi(x_2)\mathrm{d}P(x_2|d_1,x_1)$, and it satisfies $\langle f,\mu_{x_2}(d_1,x_1) \rangle_{\mathcal{H}_{\mathcal{X}}}= \int f(x_2)\mathrm{d}P(x_2|d_1,x_1)$. 

\begin{remark}[Key innovation]
    Here, $\mu_{x_1,x_2}(d_1)$ is a sequential kernel embedding that encodes the counterfactual distribution of the covariates $(X_1,X_2)$ when the initial, counterfactual treatment value is $D_1=d_1$. It satisfies $\langle f,\mu_{x_1,x_2}(d_1) \rangle_{\mathcal{H}_{\mathcal{X}}\otimes \mathcal{H}_{\mathcal{X}}}=\int f(x_1,x_2) \mathrm{d}P(x_2|d_1,x_1) \mathrm{d}P(x_1)$, implementing Robins' g-formula. 
    It is our key innovation, and it accounts for treatment-confounder feedback in this setting.
\end{remark}

\begin{remark}[Robins' g-formula is bounded over $\mathcal{H}$ when the treatments are continuous]\label{remark:bounded2}
Define the functional $F:\mathcal{H} \rightarrow \mathbb{R}$, $\gamma\mapsto \int \gamma(d_{1:T},x_{1:T}) \mathrm{d}P(x_1) \prod_{t=2}^{T} \mathrm{d}P\{x_t|d_{1:(t-1)},x_{1:(t-1)}\}$. Similar to Remark~\ref{remark:bounded}, when the treatments are continuous and when Assumption~\ref{assumption:RKHS} holds, we show that $F$ is bounded by appealing to the definition of the RKHS. Boundedness of $F$ guarantees the existence of the sequential kernel embeddings in
 Theorem~\ref{theorem:representation_planning}.
\end{remark}

As before, this decoupling is a blueprint for estimation. For example, our estimator will be $\hat{\theta}^{GF}(d_1,d_2)=\langle \hat{\gamma}, \phi(d_1)\otimes\phi(d_2)\otimes \hat{\mu}_{x_1,x_2}(d_1) \rangle_{\mathcal{H}}$ where $
\hat{\mu}_{x_1,x_2}(d_1)=n^{-1}\sum_{i=1}^n\{\phi(X_{1i})\otimes \hat{\mu}_{x_2}(d_1,X_{1i})
$. Here, $\hat{\gamma}$ is a kernel ridge regression,  $\hat{\mu}_{x_2}(d_1,x_1)$ is a generalized kernel ridge regression, and we combine them by averaging and taking the product.

\begin{algo}[Nonparametric estimation of time-varying response curves]\label{algorithm:planning}
Denote the kernel matrices
$
K_{D_1D_1}$, $K_{D_2D_2}$, $K_{X_1X_1}$, $K_{X_2X_2}\in\mathbb{R}^{n\times n}
$
calculated from the population $P$.  Denote the kernel matrices
$
K_{\tilde{D}_1\tilde{D}_1}$, $K_{\tilde{D}_2\tilde{D}_2}$, $K_{\tilde{X}_1\tilde{X}_1}$, $K_{\tilde{X}_2\tilde{X}_2}\in\mathbb{R}^{\tilde{n}\times \tilde{n}}
$
calculated from the population $\tilde{P}$. Let $\odot$ be the elementwise product. Time-varying dose response curves have closed form solutions:
\begin{enumerate}
\item $\hat{\omega}(d_1,d_2;x_1)=Y^{\top}(K_{D_1D_1}\odot K_{D_2D_2}\odot K_{X_1X_1}\odot K_{X_2X_2}+n\lambda I)^{-1}$ \\
        $[K_{D_1d_1}\odot K_{D_2d_2}\odot  K_{X_1x_{1}}\odot \{K_{X_2 X_2}(K_{D_1D_1}\odot K_{X_1X_1}+n\lambda_4 I)^{-1}(K_{D_1d_1}\odot K_{X_1x_{1}})\}]$;
    \item $\hat{\theta}^{GF}(d_1,d_2)=n^{-1}\sum_{i=1}^n\hat{\omega}(d_1,d_2;X_{1i})$;
    \item $\hat{\theta}^{DS}(d_1,d_2;\tilde{P})=\tilde{n}^{-1}\sum_{i=1}^{\tilde{n}} Y^{\top}(K_{D_1D_1}\odot K_{D_2D_2}\odot K_{X_1X_1}\odot K_{X_2X_2}+n\lambda I)^{-1}$ \\
        $[K_{D_1d_1}\odot K_{D_2d_2}\odot  K_{X_1\tilde{x}_{1i}}\odot \{K_{X_2 \tilde{X}_2}(K_{\tilde{D}_1\tilde{D}_1}\odot K_{\tilde{X}_1\tilde{X}_1}+\tilde{n}\lambda_5 I)^{-1}(K_{\tilde{D}_1d_1}\odot K_{\tilde{X}_1\tilde{x}_{1i}})\}]  $,
\end{enumerate}
where $(\lambda,\lambda_4,\lambda_5)$ are ridge regression penalty parameters. For incremental responses, we replace $K_{D_2d_2}$ with $\nabla_{d_2}K_{D_2d_2}$ where $(\nabla_{d_2}  K_{D_2{d_2}})_i=\nabla_{d_2} k(D_{2i},d_2)$. 
\end{algo}

We derive these algorithms in Section~13 of \cite{rahul2024supplement}. We give theoretical values for $(\lambda,\lambda_4,\lambda_5)$ that optimally balance bias and variance in Theorem~\ref{theorem:consistency_planning} below. Section~16 gives practical tuning procedures with closed form solutions to empirically balance bias and variance, one of which is asymptotically optimal. Note that $\hat{\theta}^{DS}$ requires observations of the treatments and covariates from the alternative population $\tilde{P}$.

\subsection{Uniform consistency with finite sample rates}

Towards a guarantee of uniform consistency, we place regularity conditions on the RKHSs and original spaces via Assumptions~\ref{assumption:RKHS} and~\ref{assumption:original}. We also assume the regression $\gamma_0$ is smooth and quantify the effective dimension of $\mathcal{H}$ via Assumption~\ref{assumption:smooth_gamma}. For the conditional kernel embeddings $\mu_{x_2}(d_1,x_1)$ and $\nu_{x_2}(d_1,x_1)$, we place further smoothness and effective dimension conditions via Assumption~\ref{assumption:smooth_op}. With these assumptions, we arrive at our next main result.
\begin{theorem}[Uniform consistency of time-varying response curves]\label{theorem:consistency_planning}
Suppose the conditions of Theorem~\ref{theorem:representation_planning} hold, as well as Assumptions~\ref{assumption:original} and~\ref{assumption:smooth_gamma}. Set $(\lambda,\lambda_4,\lambda_5)=\{n^{-1/(c+1/b)},n^{-1/(c_4+1/b_4)},$ $\tilde{n}^{-1/(c_5+1/b_5)}\}$, which is rate optimal regularization.
\begin{enumerate}
    \item If in addition Assumption~\ref{assumption:smooth_op} holds with $\mathcal{A}_4=\mathcal{X}$ and $\mathcal{B}_4=\mathcal{D}\times \mathcal{X}$, then
$
\|\hat{\theta}^{GF}-\theta_0^{GF}\|_{\infty}=O_p\left[n^{-(c-1)/\{2(c+1/b)\}}+n^{-(c_4-1)/\{2(c_4+1/b_4)\}}\right]
$.
    \item If in addition Assumption~\ref{assumption:smooth_op} holds with $\mathcal{A}_5=\mathcal{X}$ and $\mathcal{B}_5=\mathcal{D}\times \mathcal{X}$, then
    $
\|\hat{\theta}^{DS}-\theta_0^{DS}\|_{\infty}=O_p\left[n^{-(c-1)/\{2(c+1/b)\}}+\tilde{n}^{-(c_5-1)/\{2(c_5+1/b_5)\}}\right]
$.
\end{enumerate}
Likewise for the incremental responses. For example,  $
    \|\hat{\theta}^{GF,\nabla}-\theta_0^{GF,\nabla}\|_{\infty}=O_p\big[n^{-(c-1)/\{2(c+1/b)\}} $ $ + $ $ n^{-(c_4-1)/\{2(c_4+1/b_4)\}}\big]
    $. 
\end{theorem}
See Section~14 of \cite{rahul2024supplement} for the proof, exact finite sample rates, and explicit specializations of Assumption~\ref{assumption:smooth_op}. As before, these rates are at best $n^{-1/4}$ when $(c,c_4,c_5)=2$ and $(b,b_4,b_5)\rightarrow \infty$. See~\eqref{eq:optimal} for the Sobolev special case.

\begin{remark}[Technical contribution]
As before, the technical contribution underlying this theoretical guarantee is an RKHS norm rate for the sequential kernel embedding. 
In particular, Proposition~14.2 in Section~14 of \cite{rahul2024supplement} derives nonasymptotic bounds on $\sup_{d_1 \in\mathcal{D}}\|\hat{\mu}_{x_1,x_2}(d_1)-\mu_{x_1,x_2}(d_1)\|_{\mathcal{H}_{\mathcal{X}}\otimes \mathcal{H}_{\mathcal{X}}}$ and $\sup_{d_1 \in\mathcal{D}}\|\hat{\nu}_{x_1,x_2}(d_1)-\nu_{x_1,x_2}(d_1)\|_{\mathcal{H}_{\mathcal{X}}\otimes \mathcal{H}_{\mathcal{X}}}$, as a stepping stone to Theorem~\ref{theorem:consistency_planning}.  
These intermediate results appear to be new.
\end{remark}

\begin{remark}[No effect]\label{remark:no}
Consider the scenario when there is no effect of the second dose, i.e. $E\{Y^{(d_1,d_2)}\}=E\{Y^{(d_1)}\}$. 
   As argued in Section~12 of \cite{rahul2024supplement}, under this additional restriction, $\gamma_0(d_1,d_2,x_1,x_2)=E(Y|D=d_1,X_1=x_1,X_2=x_2)$. In the proof technique of Section~14 of \cite{rahul2024supplement}, if the kernel ridge regression estimator $\hat{\gamma}$ is consistent for a function $\gamma_0$ that is constant in $d_2$, our rates remain valid.

In particular, our RKHS estimator for the time-varying dose response remains uniformly consistent as long as $\mathcal{H}_{\mathcal{D}}$ contains constant functions. While the RKHS with exponentiated quadratic kernel does not satisfy this property \cite{steinwart2006explicit}, other RKHSs do. Another option is to augment an RKHS that does not contain constant functions with constant functions.

Interestingly, in this scenario Robins' g-formula simplifies to $\theta_0^{GF}(d_1)=\int E(Y|D=d_1,X=x_1)\mathrm{d}P(x_1)$ by the law of iterated expectations. When neither dose has any effect, a similar argument yields $\theta_0^{GF}=E(Y)$. Simplifying the g-formula from a surface to a curve to a scalar suggests that rate improvements may be possible. Remark~\ref{remark:oracle} discusses possible directions for future work.
\end{remark}
\section{Simulations and application}\label{sec:experiments}

\subsection{Simulations}

We evaluate 
our estimators on various nonparametric designs. 
For each nonparametric design and sample size, we implement 100 simulations and calculate mean square error (MSE) with respect to the true causal function. Figure~\ref{fig:me} visualizes results, where a lower MSE is desirable. See Section~17 of \cite{rahul2024supplement} for the data generating processes and implementation details.

\begin{figure}[ht]
\begin{centering}
     \begin{subfigure}[b]{0.49\textwidth}
         \centering
         \includegraphics[width=0.7\textwidth]{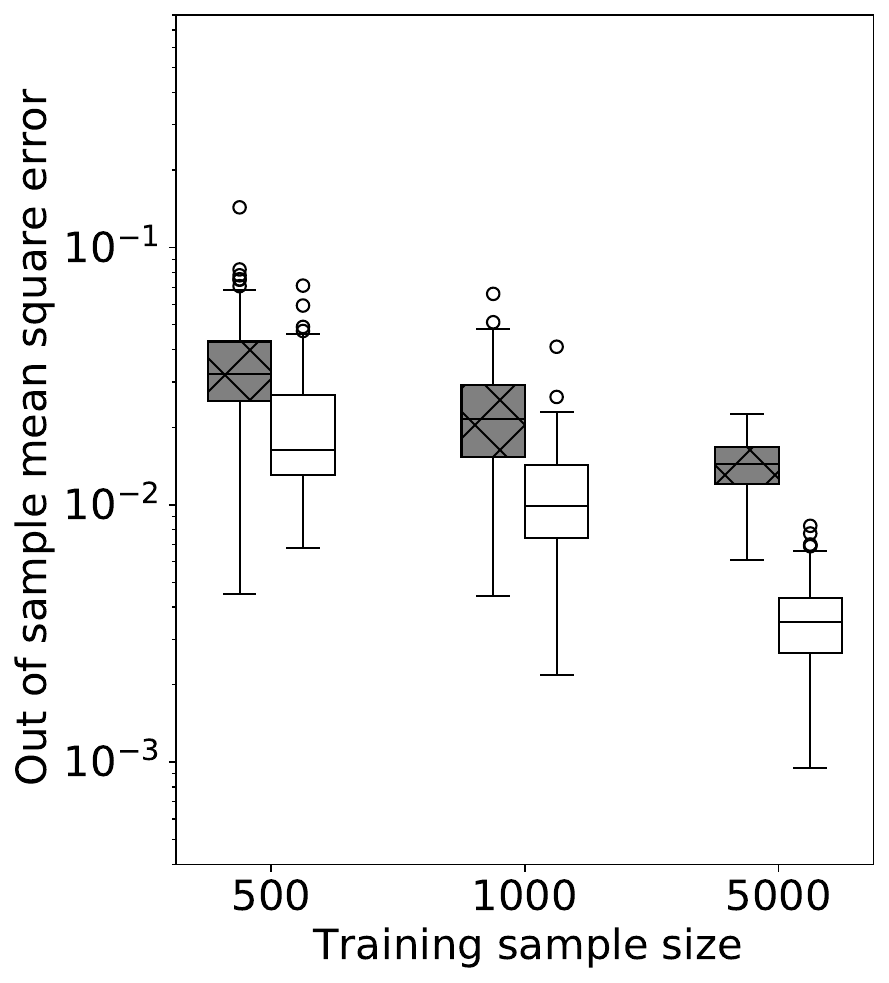}
         \caption{Mediated response.}
     \end{subfigure}
     \hfill
     \begin{subfigure}[b]{0.49\textwidth}
         \centering
         \includegraphics[width=0.7\textwidth]{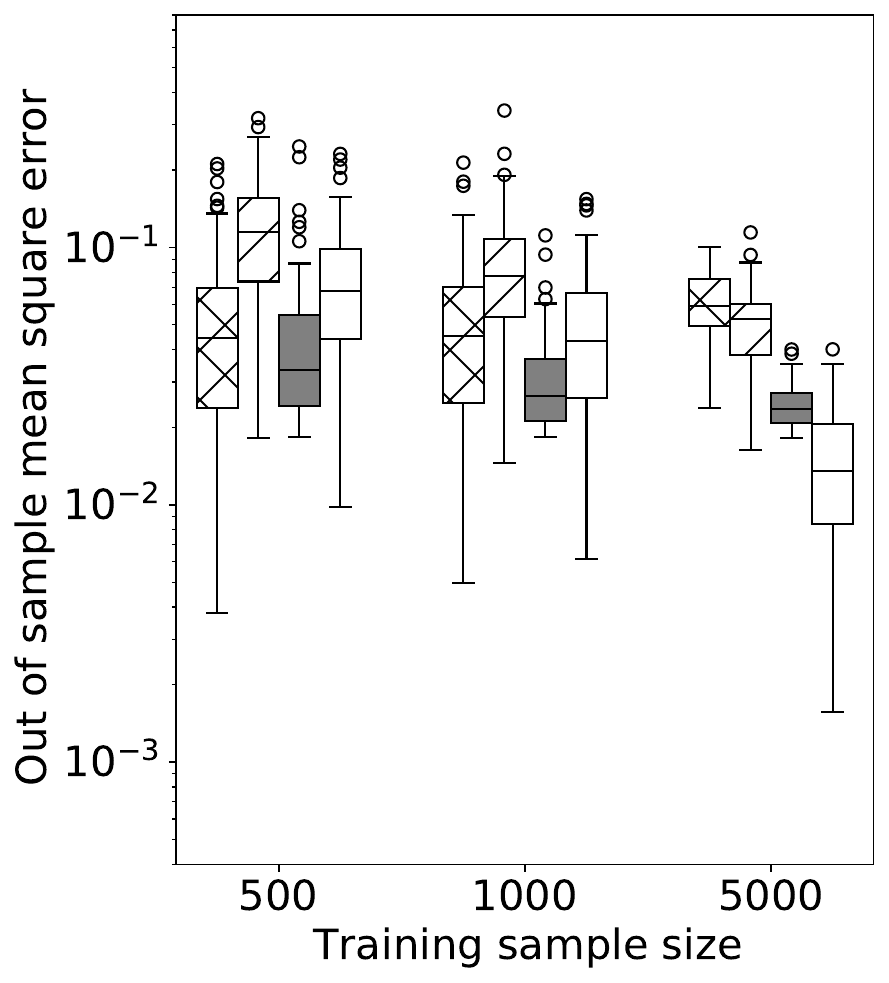}
         \caption{Time-varying dose response.}
     \end{subfigure}
\par
\caption{\label{fig:me}
Nonparametric response simulations. For the mediated response, we implement two estimators:
         \cite{huber2018direct} (\texttt{IPW}, checkered gray) 
         and our own (\texttt{RKHS}, white). For the time-varying dose response, we implement four estimators. 
         From left to right, these are 
         \cite{singh2020kernel} \{\texttt{RKHS(ATE)}\!, checkered white\}, 
         \cite{singh2020kernel} \{\texttt{RKHS(CATE)}\!, lined white\},
         \cite{lewis2020double} (\texttt{SNMM}, gray), 
         and our own \{\texttt{RKHS(GF)}\!, white\}.}
\end{centering}
\end{figure}

The mediated response design \cite{huber2018direct} involves learning the nonlinear causal function $\theta_0^{ME}(d,d')=0.3d'+0.09d+0.15 dd'+0.25  (d')^3$.  A single observation is a tuple $(Y,D,M,X)$ for the outcome, treatment, mediator, and covariates where $Y,D,M,X\in\mathbb{R}$. In addition to our estimator (\texttt{RKHS}, white), we implement the estimator of \cite{huber2018direct} (\texttt{IPW}, checkered gray), which involves Nadaraya--Watson density estimation en route to generalized inverse propensity weighting. By the Wilcoxon rank sum test, \texttt{RKHS} significantly outperforms \texttt{IPW} at all sample sizes, with p values below $10^{-3}$.

Next, we consider a time-varying dose response design, extending a time-fixed design \cite{colangelo2020double}. The nonlinear causal function is $\theta_0^{GF}(d_1,d_2)=0.6d_1+0.5d_1^2+1.2d_2+d_2^2$. A single observation is a tuple $(Y,D_{1:2},X_{1:2})$ for the outcome, treatments, and covariates where $Y,D_t\in\mathbb{R}$ and $X_t\in\mathbb{R}^{100}$. See Section~17 of \cite{rahul2024supplement} for low and moderate dimensional settings. Our machine learning approach for time-varying response curves is uniformly consistent and allows for nonlinearity, dependence over time, and effect modification, which appears to be new. 

To illustrate why treatment-confounder feedback and effect modification matter, we compare $\hat{\theta}^{GF}(d_1,d_2)$ \{\texttt{RKHS(GF)}\!, white\} with estimators that ignore these complexities to various degrees. Using the dose response estimator of \cite{singh2020kernel} \{\texttt{RKHS(ATE)}\!, checkered white\}, we take $D_2$ to be the treatment and misclassify $D_1$ as a covariate. Using the heterogeneous response estimator of \cite{singh2020kernel} \{\texttt{RKHS(CATE)}\!, lined white\}, we take $D_2$ to be the treatment and misclassify $D_1$ as the subcovariate with meaningful heterogeneity. 
We also implement the estimator of \cite{lewis2020double} (\texttt{SNMM}, gray), which is a machine learning approach with linearity, Markov, and no-effect-modification assumptions that do not hold in this setting. By the the Wilcoxon rank sum test, \texttt{RKHS(GF)} significantly outperforms the alternatives at $n=5000$ with p value below $10^{-3}$. The ability of \texttt{RKHS(GF)} to capture treatment-confounder feedback, and effect modification by time-varying confounders, helps when the sample size is large enough.

\subsection{Application: US Job Corps}

We estimate the mediated and time-varying responses of the Job Corps, the largest job training program for disadvantaged youth in the US. The Job Corps serves about 50,000 participants annually, and it is free for individuals who meet low income requirements. Access to the program was randomized from November 1994 to February 1996; see \cite{schochet2008does} for details. 
Though access to the program was randomized, individuals could decide whether to participate and for how many hours over multiple years. We assume that, conditional on the observed covariates, those decisions were as good as random in the sense formalized in Section~12 of \cite{rahul2024supplement}.

First, we consider employment to be a possible mechanism through which class hours affect arrests, under the identifying assumptions of  \cite{flores2012estimating,huber2018direct}. The covariates $X\in\mathbb{R}^{40}$ are measured at baseline; the treatment $D\in\mathbb{R}$ is total hours spent in academic or vocational classes in the first year after randomization; the mediator $M\in\mathbb{R}$ is the proportion of weeks employed in the second year after randomization; and the outcome $Y\in\mathbb{R}$ is the number of times an individual is arrested by police in the fourth year after randomization. We use the same covariates $X\in\mathbb{R}^{40}$ and sample as \cite{colangelo2020double}, with $n=2,913$ observations. 
In Figure~\ref{fig:JC2}, $\theta_0^{TE}(d,d')$, $\theta_0^{DE}(d,d')$, and $\theta_0^{IE}(d,d')$ are the total, direct, and indirect responses, respectively, of $d'$ class hours relative to $d$ class hours on arrests. 
In particular, $\theta_0^{IE}(d,d')$ estimates the extent to which the response is mediated by the mechanism of employment. 

\begin{figure}[ht]
\begin{centering}
     \begin{subfigure}[b]{0.3\textwidth}
         \centering
         \includegraphics[width=\textwidth]{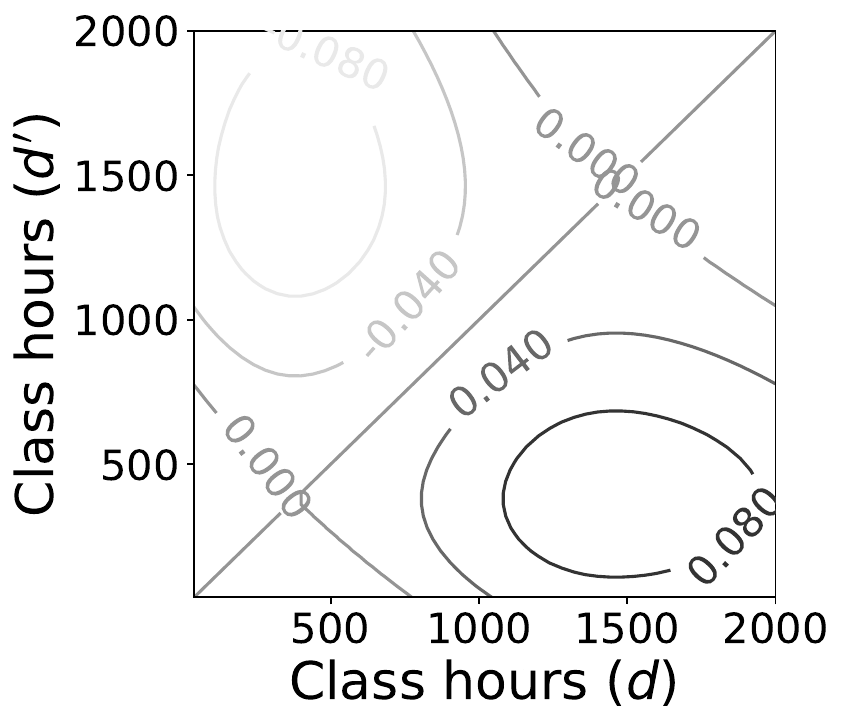}
         \caption{Total response.}
     \end{subfigure}
     \hfill
     \begin{subfigure}[b]{0.3\textwidth}
         \centering
         \includegraphics[width=\textwidth]{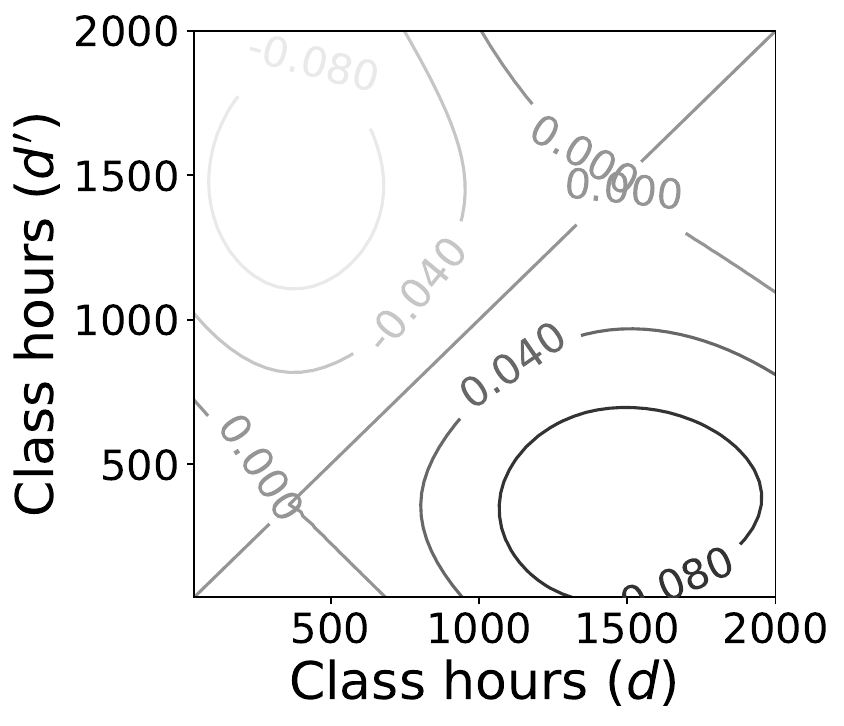}
         \caption{Direct response.}
     \end{subfigure}
     \hfill 
     \begin{subfigure}[b]{0.3\textwidth}
         \centering
         \includegraphics[width=\textwidth]{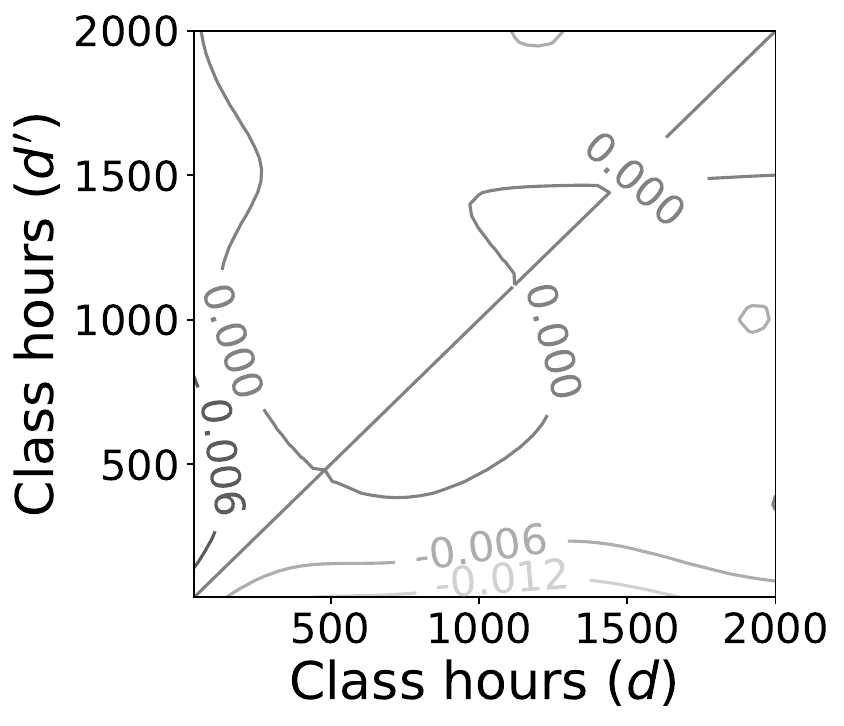}
         \caption{Indirect response.}
     \end{subfigure}
\par
\caption{\label{fig:JC2}
Effect of job training on arrests. We implement our estimators for total, direct, and indirect response curves (\texttt{RKHS}, solid).
%
}
\end{centering}
\end{figure}

At best, the total response of receiving 1,600 class hours (40 weeks) versus 480 class hours (12 weeks) may be a reduction of about 0.1 arrests. The direct response estimate, of class hours on arrests, mirrors the total response estimate. Our indirect response estimate of class hours on arrests, as mediated through employment, is essentially zero. Our results extend the findings of \cite{huber2018direct}, allowing both $(d,d')$ to vary. It appears that the effect of class hours on arrests is direct; there may be benefits of the training program that are not explained by employment alone. These benefits, however, may require many class hours. 

Next, we evaluate the time-varying response of job training on employment. Here, $X_1\in\mathbb{R}^{65}$ are covariates at baseline; $D_1\in\mathbb{R}$ is the total class hours in the first year; $X_2\in\mathbb{R}^{30}$ are covariates observed at the end of the first year; $D_2\in \mathbb{R}$ is the total class hours in the second year; and $Y\in\mathbb{R}$ is the proportion of weeks employed in the fourth year.
The covariates and the sample of $n=3,141$  observations we use are similar to \cite{colangelo2020double}. The time-varying response $\theta^{GF}_0(d_1,d_2)$ is the counterfactual mean employment given $d_1$ class hours in year one and $d_2$ class hours in year two; $\theta^{GF,\nabla}_0(d_1,d_2)$ is the increment of counterfactual mean employment given $d_1$ class hours in year one and incrementally more than $d_2$ class hours in year two. Figure~\ref{fig:JC3} visualizes the time-varying response estimate and its derivative with respect to the second dose.

\begin{figure}[ht]
\begin{centering}
     \begin{subfigure}[b]{0.49\textwidth}
         \centering
         \includegraphics[width=0.9\textwidth]{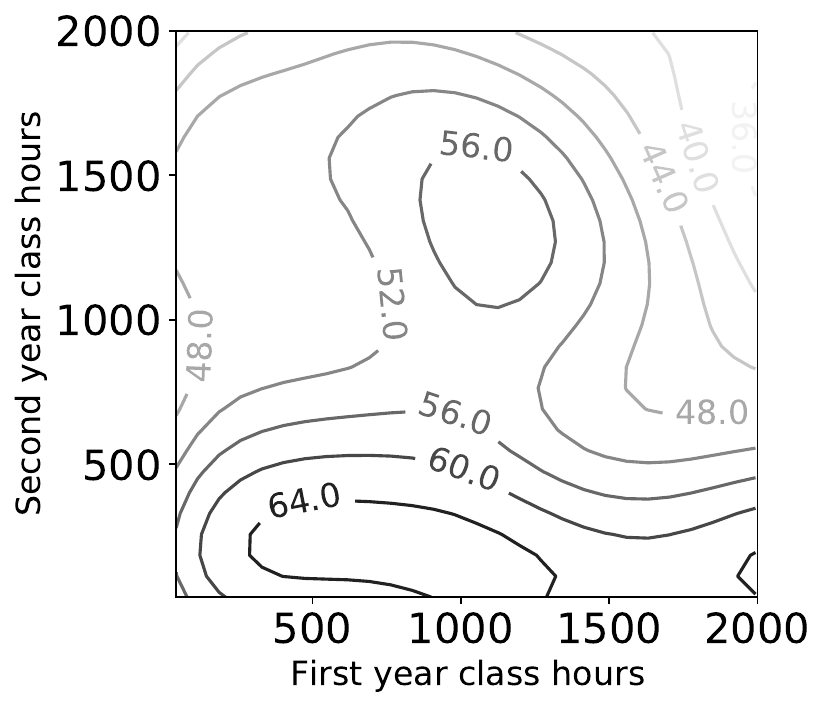}
         \caption{Time-varying dose response.}
     \end{subfigure}
     \hfill
     \begin{subfigure}[b]{0.49\textwidth}
         \centering
         \includegraphics[width=0.9\textwidth]{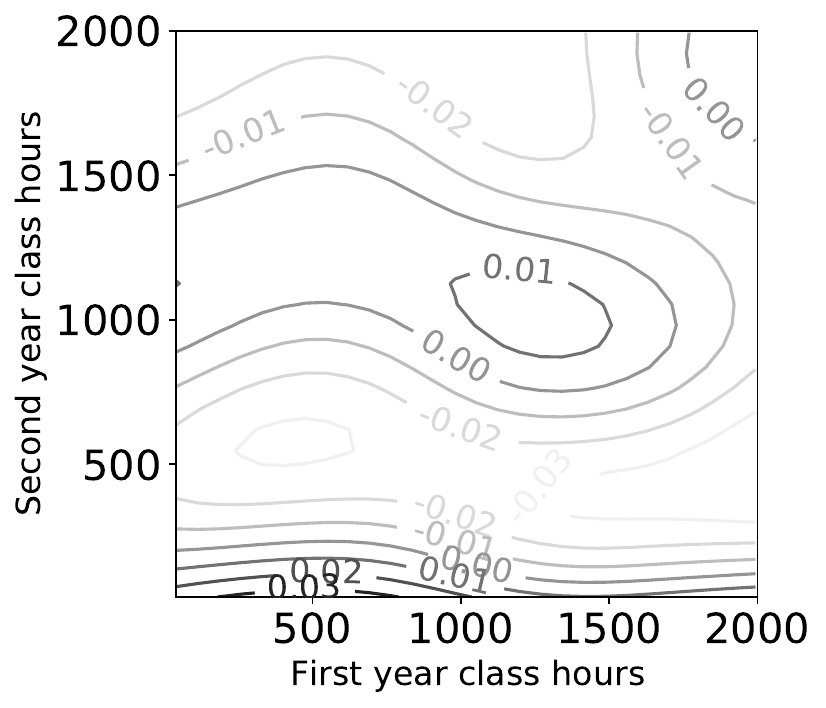}
         \caption{Time-varying incremental response.}
     \end{subfigure}
\par
\caption{
Effect of job training on employment. We implement our estimators for time-varying dose and incremental response curves \{\texttt{RKHS(GF)}\!, solid\}.\label{fig:JC3}
%
}

\end{centering}
\end{figure}

The effect of training on employment appears to be positive when the duration of training is relatively brief. At best, the response to receiving job training appears to be 64\% employment, compared to receiving no class hours at all which gives 56\% employment. The maximum response is achieved by 480--1280 class hours (12--32 weeks) in year one and 0--480 (0--12 weeks) in year two. There is another local maximum of counterfactual employment achieved by 1200 class hours (30 weeks) in both years. Class hours in year one and year two may be complementary at low levels, as visualized by the incremental response. The large plateau in counterfactual employment suggests that a successful yet cost effective policy may be 480 class hours (12 weeks) in the first year and an optional, brief follow up in the second year.  

In summary, under standard identifying assumptions, we find that the US Job Corps may provide two distinct benefits: reducing arrests and increasing employment, under different durations of class hours. Many class hours in the first year may directly decrease arrests in the fourth year, while
few class hours in the first and second years may significantly increase employment in the fourth year.
Section~18 of \cite{rahul2024supplement} provides implementation details and verifies that our results are robust to the sample choice.

\section{Discussion}\label{section:conclusion}

Previous methods using kernels for continuous treatments \cite{singh2020kernel} do not handle treatment-confounder feedback, and therefore cannot analyze the later rounds of Job Corps surveys. 
We propose the sequential kernel embedding to do so.
Whereas survey data for mediation analysis were previously available \cite{huber2018direct}, we clean additional survey data for time-varying analysis from raw files \cite[Section III.A]{schochet2008does}. 
By providing clean data, we enable empirical analysis of the later rounds of Job Corps surveys, where class hours in different years may be viewed as a sequence of time-varying continuous treatments. 
%
%
Future work may apply the sequential embedding in dynamic programming for optimal policy estimation.

\spacingset{1}

\bibliographystyle{apalike}
\DeclareRobustCommand{\VAN}[2]{#2}  
\DeclareRobustCommand{\VAN}[2]{#1}  


\newpage 

\spacingset{1.65}



\appendix
\section{Glossary of results}\label{sec:overview}
\vspace{-25pt}
\begin{table}[H]
\center
\setlength{\tabcolsep}{4pt}
\renewcommand{\arraystretch}{0}
\begin{tabular}{L{6.4cm}cccc}
\hline
Parameter & Symbol & Guarantee &  Best Rate & Section \\
\hline
\\[-3pt]  
 Total response & TE & uniform consistency & $n^{-\frac{1}{4}}$ & \ref{sec:mediation}\\
    Direct response & DE & uniform consistency & $n^{-\frac{1}{4}}$ & \ref{sec:mediation}\\
        Indirect response & IE & uniform consistency & $n^{-\frac{1}{4}}$ & \ref{sec:mediation}\\
     Total incremental response & TE,$\nabla$ & uniform consistency & $n^{-\frac{1}{4}}$ & \ref{sec:mediation}\\
    Direct incremental response & DE,$\nabla$ & uniform consistency & $n^{-\frac{1}{4}}$ & \ref{sec:mediation}\\
        Indirect incremental response & IE,$\nabla$ & uniform consistency & $n^{-\frac{1}{4}}$ & \ref{sec:mediation}\\
     Total effect & TE & Gaussian approx. & $n^{-\frac{1}{2}}$ & \ref{sec:semi}\\
    Direct effect & DE & Gaussian approx. & $n^{-\frac{1}{2}}$ & \ref{sec:semi}\\
        Indirect effect & IE & Gaussian approx. & $n^{-\frac{1}{2}}$ & \ref{sec:semi}\\
    Total counterfactual dist. & DTE & convergence in dist. & $n^{-\frac{1}{4}}$ & \ref{sec:dist}\\
    Direct counterfactual dist. & DDE & convergence in dist. & $n^{-\frac{1}{4}}$ & \ref{sec:dist}\\
        Indirect counterfactual dist. & DIE & convergence in dist. & $n^{-\frac{1}{4}}$ & \ref{sec:dist}\\
\hline
\end{tabular}
\caption{New estimators and guarantees for mediation analysis
}

\begin{tabular}{L{6.4cm}cccc}
\hline
Parameter & Symbol & Guarantee &  Best Rate & Section \\
\hline
\\[-3pt]  
  Time-varying dose response & GF & uniform consistency & $n^{-\frac{1}{4}}$ & \ref{sec:dynamic}\\
    Time-varying dose response with dist. shift & DS & uniform consistency & $n^{-\frac{1}{4}}$ & \ref{sec:dynamic}\\
     Time-varying incremental response & GF,$\nabla$ & uniform consistency & $n^{-\frac{1}{4}}$ & \ref{sec:dynamic}\\
    Time-varying incremental response with dist. shift & DS,$\nabla$ & uniform consistency & $n^{-\frac{1}{4}}$ & \ref{sec:dynamic}\\
     Time-varying treatment effect & GF & Gaussian approx. & $n^{-\frac{1}{2}}$ & \ref{sec:semi}\\
    Time-varying treatment effect with dist. shift & DS & Gaussian approx. & $n^{-\frac{1}{2}}$ & \ref{sec:semi}\\
    Time-varying counterfactual dist. & DGF & convergence in dist. & $n^{-\frac{1}{4}}$ & \ref{sec:dist}\\
    Time-varying counterfactual dist. with dist. shift & DDS & convergence in dist. & $n^{-\frac{1}{4}}$ & \ref{sec:dist}\\
\hline
\end{tabular}
\caption{New estimators and guarantees for time-varying analysis}
\end{table}

\section{Semiparametric inference}\label{sec:semi}

In this supplement, we turn to semiparametric estimation and inference. We use sequential embeddings to propose simpler nuisance parameter estimators for known inferential procedures. In particular, we avoid multiple levels of sample splitting and iterative fitting.

\subsection{Mediated treatment effects}

We quote the doubly robust moment function for mediation analysis parametrized to avoid density estimation. 
\begin{lemma}[Doubly robust moment \cite{tchetgen2012semiparametric}]\label{theorem:dr_mediation} Suppose the treatment $D$ is binary, so that $d,d'\in\{0,1\}$. If Assumption~\ref{assumption:selection_mediation} holds then $\theta^{ME}_0(d,d')=$ $E\{$ $ \psi^{ME}(d,d';\gamma_0,\omega_0,\pi_0,\rho_0;W)\}$ where $W=(Y,D,M,X)$ and 
\begin{align*}
&\psi^{ME}(d,d';\gamma,\omega,\pi,\rho;W)
=\omega(d,d';X) 
 + \frac{\1_{D=d'}}{\rho(d';M,X)}\frac{\rho(d;M,X)}{\pi(d;X)}\{Y-\gamma(d',M,X)\}\\
  &\quad + \frac{\1_{D=d}}{\pi(d;X)}\left\{\gamma(d',M,X)-\omega(d,d';X)\right\}.
\end{align*}
The function $\gamma_0$ is a regression and $
 \omega_0(d,d';X)=\int \gamma_0(d',m,X) \mathrm{d}P(m|d,X)$ is its partial mean, while $\pi_0(d;X)=\text{\normalfont pr}(d|X)$ and $\rho_0(d;M,X)=\text{\normalfont pr}(d|M,X)
$ are propensity scores.
\end{lemma}
The equation $\theta^{ME}_0(d,d')=E\{\psi^{ME}(d,d';\gamma_0,\omega_0,\pi_0,\rho_0;W)\}$ is doubly robust with respect to the nonparametric quantities $(\gamma_0,\omega_0,\pi_0,\rho_0)$ in the sense that it continues to hold if $(\gamma_0,\omega_0)$ or $(\pi_0,\rho_0)$ are misspecified.
As a consequence, consistency continues to hold if the former pair or latter pair are not actually elements of an RKHS. This parametrization, widely used in the targeted machine learning literature, suits our approach since the technique of sequential kernel embeddings allows us to estimate $\omega_0$ without explicitly estimating the conditional density $f(m|d,x)$.

The semiparametric procedure uses our new nonparametric estimator from Algorithm~\ref{algorithm:mediation} as a nuisance estimator within a meta procedure that combines the doubly robust moment function from Lemma~\ref{theorem:dr_mediation} and sample splitting. The meta algorithm of using the doubly robust moment function \cite{robins1995semiparametric} and sample splitting \cite{klaassen1987consistent}, without specifying the nonparametric nuisance estimators, is a variant of targeted \cite{zheng2011cross} and debiased \cite{chernozhukov2018original} machine learning for mediation analysis.

\begin{algo}[Semiparametric inference \cite{zheng2011cross,chernozhukov2018original}]\label{algorithm:mediation_dml}
Partition the sample into the folds $(I_{\ell})$ for $(\ell=1,...,L)$.
\begin{enumerate}
 \item For each fold $\ell$, estimate $(\hat{\gamma}_{\ell},\hat{\omega}_{\ell},\hat{\pi}_{\ell},\hat{\rho}_{\ell})$ from observations in $I_{\ell}^c$, i.e. the complement of $I_{\ell}$.
 \item Estimate $\hat{\theta}^{ME}(d,d')=n^{-1}\sum_{\ell=1}^L\sum_{i\in I_{\ell}} \psi^{ME}(d,d';\hat{\gamma}_{\ell},\hat{\omega}_{\ell},\hat{\pi}_{\ell},\hat{\rho}_{\ell};W_i)$.
 \item Estimate its $(1-a) 100$\% confidence interval as $\hat{\theta}^{ME}(d,d')\pm \varsigma_a\hat{\sigma}(d,d')n^{-1/2}$ where $\hat{\sigma}^2(d,d')
  =n^{-1}\sum_{\ell=1}^L\sum_{i\in I_{\ell}} \{\psi^{ME}(d,d';\hat{\gamma}_{\ell},\hat{\omega}_{\ell},\hat{\pi}_{\ell},\hat{\rho}_{\ell};W_i)-\hat{\theta}^{ME}(d,d')\}^2$ and $\varsigma_a$ is the $1-a/2$ quantile of the standard Gaussian. 
\end{enumerate}
\end{algo}

Unlike \cite[Algorithm 2]{farbmacher2020causal}, our procedure does not require multiple levels of sample splitting. Our algorithmic innovation is to propose and analyze a new nuisance estimator $\hat{\omega}$ in Algorithm~\ref{algorithm:mediation}, based on the idea of sequential kernel embedding, that does not require sample splitting of its own. Note that $\hat{\omega}$ is computed according to Algorithm~\ref{algorithm:mediation} with regularization parameters $(\lambda,\lambda_1)$. Here, $(\hat{\gamma},\hat{\pi},\hat{\rho})$ are standard kernel ridge regressions with regularization parameters $(\lambda,\lambda_2,\lambda_3)$. See Supplement~\ref{sec:alg_deriv} for explicit computations. 

We give theoretical values for regularization parameters that optimally balance bias and variance of the nuisances in Theorem~\ref{theorem:inference_mediation} below. Supplement~\ref{sec:tuning} gives practical tuning procedures with closed form solutions to empirically balance bias and variance, one of which is asymptotically optimal for the nuisances. Optimal regularization for the nuisances may not be optimal for the mediated treatment effect, yet it delivers the rate  $n^{-1/2}$ for the mediated treatment effect in our analysis below.

Proposition~\ref{prop:mediation} and the delta method imply confidence intervals for the other quantities in Definition~\ref{def:mediation} when the treatment is binary. 


When $D$ is binary, $\theta_0^{ME}(d,d')$ is a matrix in $\mathbb{R}^{2\times 2}$. Theorem~\ref{theorem:consistency_mediation} simplifies to a guarantee on the maximum element of the matrix of differences $|\hat{\theta}^{ME}(d,d')-\theta_0^{ME}(d,d')|$. For this case, we improve the rate from $n^{-1/4}$ to $n^{-1/2}$ by using Algorithm~\ref{algorithm:mediation_dml} and imposing additional assumptions. In particular, we place additional assumptions on the propensity scores. 

\begin{assumption}[Bounded propensities]\label{assumption:bounded_propensity}
Propensity scores $\pi_0(d;X)$ and $\rho_0(d;M,X)$, and estimators $\hat{\pi}(d;X)$ and $\hat{\rho}(d;M,X)$, are bounded away from zero and one almost surely.
\end{assumption}
The first part of Assumption~\ref{assumption:bounded_propensity} is a mild strengthening of the overlap condition in Assumption~\ref{assumption:selection_mediation}. The second part can be imposed by trimming the estimators $(\hat{\pi},\hat{\rho})$, which only improves prediction quality since $(\pi_0,\rho_0)$ are bounded. Assumption~\ref{assumption:bounded_propensity} is not specific to the RKHS setting. We now place smoothness and spectral decay assumptions that are. We assume the propensity scores are smooth and belong to RKHSs with quantifiable effective dimension in the sense of~\eqref{eq:prior}.

\begin{assumption}[Smoothness and effective dimension of propensities]\label{assumption:smooth_propensity}
Assume 
\begin{enumerate}
 \item $\pi_0\in \mathcal{H}_{\mathcal{X}}^{c_2}$ with $c_2\in(1,2]$, and $\eta_j(\mathcal{H}_{\mathcal{X}})\leq C j^{-b_2}$ with $b_2\geq 1$;
 \item $\rho_0\in (\mathcal{H}_{\mathcal{M}}\otimes \mathcal{H}_{\mathcal{X}})^{c_3}$  with $c_3\in(1,2]$, and $\eta_j(\mathcal{H}_{\mathcal{M}}\otimes \mathcal{H}_{\mathcal{X}})\leq C j^{-b_3}$  with $b_3\geq 1$.
\end{enumerate}
\end{assumption}



Finally, we prove our main result for semiparametric mediated effects.

\begin{theorem}[Semiparametric efficiency of mediated effects]\label{theorem:inference_mediation}
Suppose the conditions of Theorem~\ref{theorem:consistency_mediation} hold, as well as Assumptions~\ref{assumption:bounded_propensity} and~\ref{assumption:smooth_propensity}. Set $(\lambda,\lambda_1,\lambda_2,\lambda_3)=\big\{n^{-1/(c+1/b)},n^{-1/(c_1+1/b_1)},$ $n^{-1/(c_2+1/b_2)},n^{-1/(c_3+1/b_3)}\big\}$, which is rate optimal regularization for the nuisances. Then for any $c,c_1,c_2\in(1,2]$ and $b,b_1,b_2\geq 1$ satisfying
\begin{align*}
 &\min \left(\frac{c-1}{c+1/b},\frac{c_1-1}{c_1+1/b_1}\right)+\frac{c_2}{c_2+1/b_2}>1,
\end{align*}
we have that
$
\hat{\theta}^{ME}(d,d')=\theta^{ME}_0(d,d')+o_p(1)$,
$n^{1/2} \sigma(d,d')^{-1}\{\hat{\theta}^{ME}(d,d')-\theta^{ME}_0(d,d')\}\leadsto $ $\mathcal{N} $ $(0,1)$, and 
$
\text{\normalfont pr} \left[\theta^{ME}_0(d,d') \in  \left\{\hat{\theta}^{ME}(d,d')\pm \varsigma_a\hat{\sigma}(d,d') n^{-1/2} \right\}\right]\rightarrow 1-a.
$
Moreover, the estimator is semiparametrically efficient. 
\end{theorem}
See Supplement~\ref{sec:inference_proof} for the proof and a stronger finite sample guarantee. By Proposition~\ref{prop:mediation}, the remaining quantities in Definition~\ref{def:mediation} are $n^{-1/2}$ consistent. In the favorable case that $(c,c_1,c_2)=2$ i.e. $(\gamma_0,\mu_m,\pi_0)$ are very smooth, the conditions simplify: either $(b,b_1)>1$ or $b_2>1$. In other words, we require either the kernels of  $(\gamma_0,\mu_m)$ or the kernel of $\pi_0$ to have a low effective dimension. In this sense, we characterize a double spectral robustness that is similar in spirit to the double sparsity robustness that is familiar for lasso estimation of treatment effects. More generally, there is a trade-off among the smoothness and effective dimension assumptions across nonparametric objects. The spectral decay quantified by $(b,b_1,b_2)$ must be sufficiently fast relative to the smoothness of the various nonparametric objects $(c,c_1,c_2)$. See~\eqref{eq:optimal} to specialize these conditions for Sobolev spaces.

\begin{remark}[Nonparametric inference]\label{remark:nonparametric_inference}
    Nonparametric inference for estimators based on sequential kernel embedding, with a continuous treatment, is an important direction for future work. We outline a few possibilities and their challenges.

    One possible way forward maintains the mediated response curve as the estimand of interest. Combining \cite[Corollary 6.2]{singh2021finite} with our nonasymptotic intermediate results in Supplement~\ref{sec:inference_proof} may translate the semiparametric Gaussian approximation of mediated treatment effects into nonparametric inference on mediated response curves. The exact details are an interesting challenge.
    
    Another possibility is to change the causal estimand. An appropriately defined average derivative of potential outcomes, integrating over dosages, would reduce the inferential problem to a modified mediated treatment effect. Identification becomes an interesting challenge, for which the techniques of Supplement~\ref{sec:id} may be useful.
\end{remark}

\subsection{Time-varying treatment effects}

For time-varying treatment effects, we quote the doubly robust moment function parametrized to avoid density estimation, which is widely used in the targeted machine learning literature.
\begin{lemma}[Doubly robust moment \cite{scharfstein1999adjusting}]\label{theorem:dr_planning}
Suppose the treatment $D_t$ is binary, so that $d_1,d_2\in\{0,1\}$. If Assumption~\ref{assumption:selection_planning} holds then $\theta^{GF}_0(d_1,d_2)=E\{\psi^{GF}$ $(d_1,d_2;\gamma_0,\omega_0,\pi_0,\rho_0;W)\}$ where $W=(Y,D_1,D_2,X_1,X_2)$ and
\begin{align*}
 &\psi^{GF}(d_1,d_2;\gamma,\omega,\pi,\rho;W)
 =
 \omega(d_1,d_2;X_1) 
+ \frac{\1_{D_1=d_1}\1_{D_2=d_2}}{\pi(d_1;X_1)\rho(d_2;d_1,X_1,X_2)}\{Y-\gamma(d_1,d_2,X_1,X_2)\}\\
 &\quad + \frac{\1_{D_1=d_1}}{\pi(d_1;X_1)}\left\{\gamma(d_1,d_2,X_1,X_2)-\omega(d_1,d_2;X_1)\right\}.
\end{align*}
The function $\gamma_0$ is a regression and $\omega_0(d_1,d_2;X_1)=\int \gamma_0(d_1,d_2,X_1,x_2) \mathrm{d}P(x_2|d_1,X_1)$ is its partial mean, while $ \pi_0(d_1;X_1)=\text{\normalfont pr}(d_1|X_1)$ and $\rho_0(d_2;d_1,X_1,X_2)=\text{\normalfont pr}(d_2|d_1,X_1,X_2)$ are propensity scores.
\end{lemma}
Similar to Lemma~\ref{theorem:dr_mediation}, 
the equation $\theta^{GF}_0(d_1,d_2)=E\{\psi^{GF}(d_1,d_2;\gamma_0,\omega_0,\pi_0,\rho_0;W)\}$
is doubly robust with respect to the nonparametric quantities $(\gamma_0,\omega_0,\pi_0,\rho_0)$ in the sense that it continues to hold if $(\gamma_0,\omega_0)$ or $(\pi_0,\rho_0)$ are misspecified. 
Once again, 
the technique of sequential kernel embeddings allows us to estimate $\omega_0$ without explicitly estimating the conditional density $f(x_2|d_1,x_1)$. Altogether, the semiparametric procedure uses our new nonparametric estimator from Algorithm~\ref{algorithm:planning} as a nuisance estimator within the framework of targeted \cite{zheng2011cross} and debiased \cite{chernozhukov2018original} machine learning for time-varying treatment effects.

\begin{algo}[Semiparametric inference \cite{zheng2011cross,chernozhukov2018original}]\label{algorithm:planning_dml}
Partition the sample into the folds $(I_{\ell})$ for $(\ell=1,...,L)$.
\begin{enumerate}
 \item For each fold $\ell$, estimate $(\hat{\gamma}_{\ell},\hat{\omega}_{\ell},\hat{\pi}_{\ell},\hat{\rho}_{\ell})$ from observations in $I_{\ell}^c$, i.e. the complement of $I_{\ell}$.
 \item Estimate $\hat{\theta}^{GF}(d_1,d_2)=n^{-1}\sum_{\ell=1}^L\sum_{i\in I_{\ell}} \psi^{GF}(d_1,d_2;\hat{\gamma}_{\ell},\hat{\omega}_{\ell},\hat{\pi}_{\ell},\hat{\rho}_{\ell};W_i)$.
 \item Estimate its $(1-a) 100$\% confidence interval as $\hat{\theta}^{GF}(d_1,d_2)\pm \varsigma_a\hat{\sigma}(d_1,d_2)n^{-1/2}$ where
 $\hat{\sigma}^2$ $(d_1,d_2)
  $ $=n^{-1}\sum_{\ell=1}^L\sum_{i\in I_{\ell}} \{\psi^{GF}(d_1,d_2;\hat{\gamma}_{\ell},\hat{\omega}_{\ell},\hat{\pi}_{\ell},\hat{\rho}_{\ell};W_i)-\hat{\theta}^{GF}(d_1,d_2)\}^2$ and $\varsigma_a$ is the $1-a/2$ quantile of the standard Gaussian. 
\end{enumerate}
\end{algo}

Unlike \cite[Algorithm 1]{bodory2021evaluating}, our procedure does not require multiple levels of sample splitting. Our algorithmic innovation is the nuisance estimator $\hat{\omega}$ in Algorithm~\ref{algorithm:planning}, based on the idea of sequential kernel embedding. The regularization parameters for $\hat{\omega}$ are $(\lambda,\lambda_4)$, while the regularization parameters for $(\hat{\gamma},\hat{\pi},\hat{\rho})$ are $(\lambda,\lambda_6,\lambda_7)$. 

 Theorem~\ref{theorem:inference_planning} gives theoretical values for regularization parameters that optimally balance bias and variance of the nuisances, while Supplement~\ref{sec:tuning} gives practical tuning procedures with closed form solutions to empirically balance bias and variance of the nuisances. 
As before, optimal regularization for the nuisances may not be optimal for the time-varying treatment effect, yet it delivers the rate  $n^{-1/2}$ for the time-varying treatment effect in our analysis below.



As before, when $(D_1,D_2)$ are binary, we are able to improve the rate from $n^{-1/4}$ to $n^{-1/2}$ by using Algorithm~\ref{algorithm:planning_dml} and imposing additional assumptions. In particular, we place additional assumptions on the propensity scores. 

\begin{assumption}[Bounded propensities and density ratio]\label{assumption:bounded_propensity_planning}
Assume that, almost surely, (i)
propensity scores $\pi_0(d_1;X_1)$ and $\rho_0(d_2;D_1,X_1,X_2)$ are bounded away from zero and one; (ii) their estimators $\hat{\pi}(d_1;X_1)$ and $\hat{\rho}(d_2;D_1,X_1,X_2)$ are bounded away from zero and one;
(iii) the density ratio $f(x_2|d_1,X_1)/f(x_2|d'_1X_1)$ is bounded for any values $d_1,d_1'\in\mathcal{D}$.
\end{assumption}
The first part of Assumption~\ref{assumption:bounded_propensity_planning} is a mild strengthening of the overlap condition in Assumption~\ref{assumption:selection_planning}. The second part can be imposed by trimming the estimators $(\hat{\pi},\hat{\rho})$. The third part simplifies the proof technique and can be relaxed. Finally, as before, we assume the propensity scores are smooth and belong to RKHSs with quantifiable effective dimension in the sense  of~\eqref{eq:prior}.

\begin{assumption}[Smoothness and effective dimension of propensities]\label{assumption:smooth_propensity_planning}
Assume 
\begin{enumerate}
 \item $\pi_0\in \mathcal{H}_{\mathcal{X}}^{c_6}$ with $c_6\in(1,2]$, and $\eta_j(\mathcal{H}_{\mathcal{X}})\leq C j^{-b_6}$ with $b_6\geq 1$; 
 \item $\rho_0\in (\mathcal{H}_{\mathcal{D}}\otimes \mathcal{H}_{\mathcal{X}}\otimes \mathcal{H}_{\mathcal{X}})^{c_7}$ with $c_7\in(1,2]$, and $\eta_j(\mathcal{H}_{\mathcal{D}}\otimes \mathcal{H}_{\mathcal{X}}\otimes \mathcal{H}_{\mathcal{X}})\leq C j^{-b_7}$ with $b_7\geq 1$.
\end{enumerate}
\end{assumption}

Our main result for semiparametric time-varying treatment effects is as follows.

\begin{theorem}[Semiparametric efficiency of time-varying treatment effects]\label{theorem:inference_planning}
Suppose the conditions of Theorem~\ref{theorem:consistency_planning} hold, as well as Assumptions~\ref{assumption:bounded_propensity_planning} and~\ref{assumption:smooth_propensity_planning}. Set $(\lambda,\lambda_4,\lambda_6,\lambda_7)=\big\{n^{-1/(c+1/b)},$ $n^{-1/(c_4+1/b_4)},n^{-1/(c_6+1/b_6)},n^{-1/(c_7+1/b_7)}\big\}$, which is rate optimal regularization for nuisances. Then for any $c,c_4,c_6$ $\in(1,2]$ and $b,b_4,b_6\geq 1$ satisfying
\begin{align*}
 &\min \left(\frac{c-1}{c+1/b},\frac{c_1-1}{c_4+1/b_4}\right)+\frac{c_6}{c_6+1/b_6}>1,
\end{align*}
we have that
$
\hat{\theta}^{GF}(d_1,d_2)=\theta^{GF}_0(d_1,d_2)+o_p(1)$,
$n^{1/2} \sigma(d_1,d_2)^{-1}\{\hat{\theta}^{GF}(d_1,d_2)-\theta^{GF}_0(d_1,d_2)\}\leadsto\mathcal{N}(0,1)$,
and
$
\text{\normalfont pr} \left[\theta^{GF}_0(d_1,d_2) \in  \left\{\hat{\theta}^{GF}(d_1,d_2)\pm \varsigma_a\hat{\sigma}(d_1,d_2) n^{-1/2} \right\}\right]\rightarrow 1-a.
$
Moreover, the estimator is semiparametrically efficient.
\end{theorem}
See Supplement~\ref{sec:inference_proof} for the proof and a stronger finite sample guarantee. Once again, we obtain double spectral robustness. In the favorable case that $(c,c_4,c_6)=2$ i.e. $(\gamma_0,\mu_{x_2},\pi_0)$ are very smooth, either $(b,b_4)>1$ or $b_6>1$; we require either the kernels of  $(\gamma_0,\mu_{x_2})$ or the kernel of $\pi_0$ to have low effective dimension. More generally, the spectral decay quantified by $(b,b_4,b_6)$ must be sufficiently fast relative to the smoothness of the various nonparametric objects $(c,c_4,c_6)$. See~\eqref{eq:optimal} for the Sobolev special case.

\begin{remark}[Nonparametric inference]
    Nonparametric inference for estimators based on sequential kernel embedding
    is an important direction for future work. See Remark~\ref{remark:nonparametric_inference} for possible ways forward.
\end{remark}

\subsection{Simulations and application}


We implement semiparametric simulation designs that are simplifications of the nonparametric simulation designs. The goal of the semiparametric mediated effect design is to learn $\theta_0^{ME}(d,d')=0.09d+0.55d'+0.15dd'$. A single observation consists of the tuple $(Y,D,M,X)$ where now $D\in\{0,1\}$. The goal of the semiparametric time-varying treatment effect design is to learn $\theta_0^{GF}(d_1,d_2)=1.1d_1+2.2d_2+0.5d_1d_2$. A single observation consists of the tuple $(Y,D_{1:2},X_{1:2})$ where now $D_t\in\{0,1\}$.  

Over 100 simulations, we calculate estimates with 95\% confidence intervals. Tables~\ref{tab:coverage} and~\ref{tab:coverage_planning} present (i) the average estimate, (ii) the standard error of estimates, and (iii) the percentage of confidence intervals containing the true value. Across designs and sample sizes, our confidence intervals achieve nominal coverage, and the confidence intervals shrink as the sample increases. 



\begin{table}[H]
\center

\setlength{\tabcolsep}{5pt}
\begin{tabular}{lcccccccccccc}
 \hline
 & \multicolumn{3}{c}{$\theta_0^{ME}(0,0)=0$} 
 & \multicolumn{3}{c}{$\theta_0^{ME}(1,0)=0.09$}
 & \multicolumn{3}{c}{$\theta_0^{ME}(0,1)=0.55$}
 & \multicolumn{3}{c}{$\theta_0^{ME}(1,1)=0.79$} \\
\cmidrule(lr){2-4}\cmidrule(lr){5-7}\cmidrule(lr){8-10}\cmidrule(lr){11-13}
Sample & Mean & S.E. &  Cov. &  Mean & S.E. &Cov. &  Mean & S.E. &  Cov. &  Mean & S.E. &Cov. \\
\hline
500 & 0.004  &0.008 &94\% & 0.077  & 0.011 & 98\% & 0.575&0.013&  93\% & 0.792  & 0.010  &94\% \\
1000 & 0.000 &0.005 &97\% &  0.085 & 0.007 &  98\% & 0.583  & 0.009&  93\% &0.796&0.007  &96\% \\
5000 &  0.000 &0.002  &99\% & 0.083 & 0.002 & 98\%  & 0.571&0.003 & 87\% & 0.793&0.003 & 99\% \\
\hline
\end{tabular}
\caption{Semiparametric mediated effect coverage simulations}\label{tab:coverage}
\end{table}

\begin{table}[H]
\center

\setlength{\tabcolsep}{5pt}
\begin{tabular}{lcccccccccccc}
 \hline
 & \multicolumn{3}{c}{$\theta_0^{GF}(0,0)=0$} 
 & \multicolumn{3}{c}{$\theta_0^{GF}(1,0)=1.1$}
 & \multicolumn{3}{c}{$\theta_0^{GF}(0,1)=2.2$}
 & \multicolumn{3}{c}{$\theta_0^{GF}(1,1)=3.8$} \\
 \cmidrule(lr){2-4}\cmidrule(lr){5-7}\cmidrule(lr){8-10}\cmidrule(lr){11-13}
Sample & Mean & S.E. &  Cov. &  Mean & S.E. &Cov. &  Mean & S.E. &  Cov. &  Mean & S.E. &Cov. \\
\hline
500 &0.079 & 0.048 &95\% & 1.272  &0.057 &99\% & 2.848  &0.097& 99\% & 4.476 & 0.079 &95\% \\
1000 & 0.034 & 0.029& 96\%&  1.212 & 0.028  &  99\% & 2.480  &0.049& 99\% &4.145 & 0.040 &91\% \\
5000 & 0.013  &0.007  &  98\% & 1.114&  0.007  &  96\%  & 2.249  &0.012& 97\% & 3.853 & 0.008 &93\%\\
\hline
\end{tabular}
\caption{Semiparametric time-varying treatment effect coverage simulations}\label{tab:coverage_planning}
\end{table}


We implement the semiparametric time-varying treatment effect estimator with confidence intervals for the Job Corps. Though class hours  $(D_1,D_2)$ are observed as continuous variables, we discretize them into the bins of less versus more than $10^3$ class hours to transform the nonparametric problem into a semiparametric one. We choose a coarse grid because the simulations show that our procedure requires a sufficiently large sample size for good performance. Table~\ref{tab:JC_semi} summarizes results, arrayed to mirror the quadrants of Figure~\ref{fig:JC3}. The semiparametric results corroborate our nonparametric results: under standard identifying assumptions, most of the gain in counterfactual employment is achieved with few class hours in both years. The results also support our finding that many class hours in the second year are unproductive for employment. Our results appear to be statistically significant.

\begin{table}[H]
\center

\begin{tabular}{lccccc}
\hline
  &   & \multicolumn{4}{c}{Year two class hours} \\[5pt]
  &   &  \multicolumn{2}{c}{$\leq 10^3$}   &  \multicolumn{2}{c}{$>10^3$}   \\
\cmidrule(lr){3-4}\cmidrule(lr){5-6}
  &   &  Est. & S.E.  &  Est. & S.E.  \\
\hline
 Year one & $>10^3$  & 41.5  & 0.75 &44.5  & 0.70 \\
  class hours & $\leq 10^3$  &  53.9 & 1.00 & 56.6  & 0.95\\
\hline
\end{tabular}
\caption{Semiparametric time-varying effect of job training on employment}\label{tab:JC_semi}
\end{table}




\section{Counterfactual distributions}\label{sec:dist}

In this supplement, we extend the algorithms and analyses presented in the main text to counterfactual distributions.

\subsection{Definition and identification}

In the main text, we study target parameters defined as the means or increments of potential outcomes. In fact, our framework for nonparametric estimation extends to target parameters defined as distributions of potential outcomes. Counterfactual distributions can be embedded by a kernel using a new feature map $\phi(y)$ for a new scalar-valued RKHS $\mathcal{H}_{\mathcal{Y}}$. We now allow $\mathcal{Y}$ to be a Polish space, modifying Assumption~\ref{assumption:original}.
\begin{definition}[Counterfactual distributions and embeddings] 
For mediation analysis, $\theta_0^{DME}$ $(d,d')=$ $ P[Y^{\{d',M^{(d)}\}}]
$ is the counterfactual distribution of outcomes in the thought experiment that the treatment is set at a new value $D=d'$ but the mediator $M$ follows the distribution it would have followed if the treatment were set at its old value $D=d$.
For time-varying analysis, (i) $\theta_0^{DGF}(d_{1:T})=P\{Y^{(d_{1:T})}\}$ is the counterfactual distribution of outcomes given the interventions $D_{1:T}=d_{1:T}$ for the entire population; (ii) $ \theta_0^{DDS}(d_{1:T};\tilde{P})=\tilde{P}\{Y^{(d_{1:T})}\}$ is the counterfactual distribution of outcomes given the interventions $D_{1:T}=d_{1:T}$ for an alternative population with the data distribution $\tilde{P}$, elaborated in Assumption~\ref{assumption:covariate}. 
Likewise we define embeddings of the counterfactual distributions, e.g. $
\check{\theta}_0^{DME}(d,d')=E(\phi[Y^{\{d',M^{(d)}\}}])
$.
\end{definition}

The general strategy will be to estimate the embedding of a  counterfactual distribution. At that point, the analyst may use the embedding to (i) estimate moments of the counterfactual distribution \cite{kanagawa2014recovering} or (ii) sample from the counterfactual distribution \cite{welling2009herding}. We focus on the latter in this supplement. The same identification results apply to counterfactual distributions.

\begin{theorem}[Identification of counterfactual distributions and embeddings]
We show the following.
\begin{enumerate}
    \item If Assumption~\ref{assumption:selection_mediation} holds then
 $\{\theta_0^{DME}(d,d')\}(y)=\int P(y|d',m,x) \mathrm{d}P(m|d,x)\mathrm{d}P(x)$.
    \item If Assumption~\ref{assumption:selection_planning} holds then
 $
\{\theta_0^{DGF}(d_{1:T})\}(y)=\int P(y|d_{1:T},x_{1:T}) \mathrm{d}P(x_1) $ \\ $ \prod_{t=2}^{T} \mathrm{d}P\{x_t|d_{1:(t-1)},x_{1:(t-1)}\}
$.
    \item If in addition Assumption~\ref{assumption:covariate} holds then  $
\{\theta_0^{DDS}(d_{1:T};\tilde{P})\}(y)=\int P(y|d_{1:T},x_{1:T}) \mathrm{d}\tilde{P}(x_1)$ \\$ \prod_{t=2}^{T} \mathrm{d}\tilde{P}\{x_t|d_{1:(t-1)},x_{1:(t-1)}\}
$.
\end{enumerate}
Likewise for embeddings of counterfactual distributions. For example, if in addition Assumption~\ref{assumption:RKHS} holds, then $\check{\theta}_0^{DME}(d,d')=\int E\{\phi(Y)|D=d',M=m,X=x\}\mathrm{d}P(m|d,x)\mathrm{d}P(x)
$.
\end{theorem}
The identification results for embeddings of counterfactual distributions resemble those presented in the main text. Define the generalized regressions
$
\gamma_0(d,m,x)=E\{\phi(Y)|D=d,M=m,X=x\}
$ and $\gamma_0(d_{1:T},x_{1:T})=E\{\phi(Y)|D_{1:T}=d_{1:T},X_{1:T}=x_{1:T}\}$. Then we can express these results in the familiar form, e.g. $\check{\theta}_0^{DME}(d,d')=\int \gamma_0(d',m,x)\mathrm{d}P(m|d,x)\mathrm{d}P(x)$.

\subsection{Sequential kernel embedding}

To estimate counterfactual distributions, we extend the RKHS constructions in Sections~4 and~5. Define an additional scalar-valued RKHS for the outcome $Y$. Because the regression $\gamma_0$ is now a conditional kernel embedding, we present a construction involving a conditional expectation operator. For mediation analysis, define the conditional expectation operator
$
E_8:\mathcal{H}_{\mathcal{Y}}\rightarrow \mathcal{H}_{\mathcal{D}}\otimes \mathcal{H}_{\mathcal{M}}\otimes  \mathcal{H}_{\mathcal{X}},\; f(\cdot)\mapsto E\{f(Y)|D=\cdot,M=\cdot,X=\cdot \}
$. By construction,
$
\gamma_0(d,m,x)=E_8^*\{\phi(d)\otimes \phi(m)\otimes \phi(x)\}
$.  Likewise, for time-varying analysis, define the conditional expectation operator
$
E_9:\mathcal{H}_{\mathcal{Y}}\rightarrow \mathcal{H}_{\mathcal{D}}\otimes \mathcal{H}_{\mathcal{D}}\otimes  \mathcal{H}_{\mathcal{X}}\otimes  \mathcal{H}_{\mathcal{X}},\; f(\cdot)\mapsto E\{f(Y)|D_1=\cdot,D_2=\cdot,X_1=\cdot,X_2=\cdot \}
$.  We place regularity conditions on this RKHS construction in order to prove a generalized decoupling result. 

\begin{theorem}[Decoupling via sequential kernel embeddings]\label{theorem:representation_dist}
We show the following.
\begin{enumerate}
    \item Suppose the conditions of Lemma~\ref{lemma:id_mediation} hold. Further suppose Assumption~\ref{assumption:RKHS} holds and $E_8\in\mathcal{L}_2(\mathcal{H}_{\mathcal{Y}},\mathcal{H}_{\mathcal{D}}\otimes\mathcal{H}_{\mathcal{M}}\otimes \mathcal{H}_{\mathcal{X}})$. Then
$\check{\theta}_0^{DME}(d,d')=E_8^*\left[\phi(d')\otimes \mu_{m,x}(d)\right] $. 
    \item Suppose the conditions of Lemma~\ref{lemma:id_planning} hold. Further suppose Assumption~\ref{assumption:RKHS} holds and $E_9\in\mathcal{L}_2(\mathcal{H}_{\mathcal{Y}},\mathcal{H}_{\mathcal{D}}\otimes \mathcal{H}_{\mathcal{D}}\otimes \mathcal{H}_{\mathcal{X}}\otimes \mathcal{H}_{\mathcal{X}})$. Then
 $\check{\theta}_0^{DGF}(d_1,d_2)=E_9^*\left[\phi(d_1)\otimes\phi(d_2) \otimes \mu_{x_1,x_2}(d_1) \right]$; 
  and $\check{\theta}_0^{DDS}(d_1,d_2;\tilde{P})=E_9^*\left[\phi(d_1)\otimes\phi(d_2) \otimes \nu_{x_1,x_2}(d_1) \right]$. 
\end{enumerate}
\end{theorem}

See Supplement~\ref{sec:alg_deriv} for the proof. The kernel embeddings are the same as in the main text. They encode the reweighting distributions as elements in the RKHS such that the counterfactual distribution embeddings can be expressed as evaluations of $(E_8^*,E_9^*)$. As in the main text, these decoupled representations help to define estimators with closed form solutions that can be easily computed. For example, for $\check{\theta}_0^{GF}(d_1,d_2)$, our estimator will be $\hat{\theta}^{DGF}(d_1,d_2)$, which equals
$$\hat{E}_9^*\left[\phi(d_1)\otimes\phi(d_2) \otimes \hat{\mu}_{x_1,x_2}(d_1) \right]=\hat{E}_9^*\left[\phi(d_1)\otimes\phi(d_2) \otimes \frac{1}{n}\sum_{i=1}^n \{\phi(X_{1i})\otimes \hat{\mu}_{x_2}(d_1,X_{1i})\} \right]\!.$$ The estimators $\hat{E}_9$ and $\hat{\mu}_{x_2}(d_1,x_1)$ are generalized kernel ridge regressions.

\begin{algo}[Nonparametric estimation of counterfactual distribution embeddings]\label{algorithm:dist}
Denote by $\odot$ the elementwise product. For mediation analysis denote the kernel matrices
$
K_{DD},$ $K_{MM}$, $K_{XX}$, $K_{YY}\in\mathbb{R}^{n\times n}
$. Then
\begin{enumerate}
    \item $\{\hat{\theta}^{DME}(d,d')\}(y)=\frac{1}{n}\sum_{i=1}^n K_{yY}(K_{DD}\odot K_{MM}\odot K_{XX}+n\lambda_8 I)^{-1}$ \\ $
    \quad [K_{Dd'}\odot  \{K_{MM}(K_{DD}\odot K_{XX}+n\lambda_1 I)^{-1}(K_{Dd}\odot K_{Xx_i})\}\odot K_{Xx_i}]$.
\end{enumerate}
For time-varying analysis, denote the kernel matrices
$
K_{D_1D_1}$, $K_{D_2D_2}$, $K_{X_1X_1}$, $K_{X_2X_2}\in\mathbb{R}^{n\times n}
$
calculated from the observations drawn from $P$.  Denote the kernel matrices
$
K_{\tilde{D}_1\tilde{D}_1}$, $K_{\tilde{D}_2\tilde{D}_2}$, $K_{\tilde{X}_1\tilde{X}_1}$, $K_{\tilde{X}_2\tilde{X}_2}\in\mathbb{R}^{\tilde{n}\times \tilde{n}}
$
calculated from the observations drawn from $\tilde{P}$. Then 
\begin{enumerate}
    \item[2.]  $\{\hat{\theta}^{DGF}(d_1,d_2)\}(y)=\frac{1}{n}\sum_{i=1}^n K_{yY}(K_{D_1D_1}\odot K_{D_2D_2}\odot K_{X_1X_1}\odot K_{X_2X_2}+n\lambda_9 I)^{-1}$ \\$
        [K_{D_1d_1}\odot K_{D_2d_2}\odot  K_{X_1x_{1i}}\odot \{K_{X_2 X_2}(K_{D_1D_1}\odot K_{X_1X_1}+n\lambda_4 I)^{-1}(K_{D_1d_1}\odot K_{X_1x_{1i}})\}]  $;
    \item[3.]  $\{\hat{\theta}^{DDS}(d_2,d_2;\tilde{P})\}(y)=\frac{1}{\tilde{n}}\sum_{i=1}^{\tilde{n}} K_{yY}(K_{D_1D_1}\odot K_{D_2D_2}\odot K_{X_1X_1}\odot K_{X_2X_2}+n\lambda_9 I)^{-1}$ \\$
        [K_{D_1d_1}\odot K_{D_2d_2}\odot  K_{X_1\tilde{x}_{1i}}\odot \{K_{X_2 \tilde{X}_2}(K_{\tilde{D}_1\tilde{D}_1}\odot K_{\tilde{X}_1\tilde{X}_1}+\tilde{n}\lambda_5 I)^{-1}(K_{\tilde{D}_1d_1}\odot K_{\tilde{X}_1\tilde{x}_{1i}})\}]  $,
\end{enumerate}
where $(\lambda_1,\lambda_4,\lambda_5,\lambda_8,\lambda_9)$ are ridge regression penalty hyperparameters.
\end{algo}
We derive these algorithms in Supplement~\ref{sec:alg_deriv}. We give theoretical values for $(\lambda_1,\lambda_4,\lambda_5,\lambda_8,\lambda_9)$ that balance bias and variance in Theorem~\ref{theorem:consistency_dist} below. Supplement~\ref{sec:tuning} gives practical tuning procedures to empirically balance bias and variance. Note that $\hat{\theta}^{DDS}$ requires observations of the treatments and covariates from the alternative population $\tilde{P}$.

Algorithm~\ref{algorithm:dist} estimates counterfactual distribution embeddings. The ultimate parameters of interest are counterfactual distributions. We present a deterministic procedure that uses the distribution embedding to provide samples $(\tilde{Y_j})$ from the distribution. The procedure is a variant of kernel herding  \cite{welling2009herding,muandet2020counterfactual}.

\begin{algo}[Nonparametric estimation of counterfactual distributions]\label{algorithm:herding}
Recall that $\hat{\theta}^{DME}$ $(d,d')$ is a mapping from $\mathcal{Y}$ to $\mathbb{R}$. 
Given $\hat{\theta}^{DME}(d,d')$, calculate
\begin{enumerate}
    \item $\tilde{Y}_1=\argmax_{y\in\mathcal{Y}} \left[\{\hat{\theta}^{DME}(d,d')\}(y)\right]$,
    \item $\tilde{Y}_{j}=\argmax_{y\in\mathcal{Y}} \left[\{ \hat{\theta}^{DME}(d,d')\}(y)-\frac{1}{j+1}\sum_{\ell=1}^{j-1}k_{\mathcal{Y}}(\tilde{Y}_{\ell},y)\right]$ for $j>1$.
\end{enumerate}
Likewise for the other counterfactual distributions, replacing $\hat{\theta}^{DME}(d,d')$ with the other quantities in Algorithm~\ref{algorithm:dist}.
\end{algo}
By this procedure, samples from counterfactual distributions are straightforward to compute. With such samples, one may visualize a histogram as an estimator of the counterfactual density of potential outcomes. Alternatively, one may test statistical hypotheses using the samples $(\tilde{Y}_j)$.

\subsection{Convergence in distribution}

Towards a guarantee of uniform consistency, we place regularity conditions on the original spaces as in Assumption~\ref{assumption:original}. We relax the condition that $Y\in\mathbb{R}$ and that it is bounded; instead, we assume $Y\in\mathcal{Y}$, which is a Polish space.
We place assumptions on the smoothness of the generalized regression $\gamma_0$ and the effective dimension of its RKHS, expressed in terms of the conditional expectation operators $(E_8,E_9)$. Likewise, we place assumptions on the smoothness of the conditional kernel embeddings and the effective dimensions of their RKHSs. With these assumptions, we arrive at our next theoretical result.
\begin{theorem}[Consistency of counterfactual distribution embeddings]\label{theorem:consistency_dist}
Set $(\lambda_1,\lambda_4,\lambda_5,\lambda_8,\lambda_9)=$ $\left\{n^{-1/(c_1+1/b_1)},n^{-1/(c_4+1/b_4)},\tilde{n}^{-1/(c_5+1/b_5)},n^{-1/(c_8+1/b_8)},n^{-1/(c_9+1/b_9)}\right\}$, which is rate optimal regularization.
For mediation analysis, suppose the conditions of Theorem~\ref{theorem:consistency_mediation} hold as well as Assumption~\ref{assumption:smooth_op} with $\mathcal{A}_8=\mathcal{Y}$ and $\mathcal{B}_8=\mathcal{D}\times \mathcal{M}\times \mathcal{X}$. Then
 $$
    \sup_{d,d'\in\mathcal{D}}\|\hat{\theta}^{DME}(d,d')-\check{\theta}_0^{DME}(d,d')\|_{\mathcal{H}_{\mathcal{Y}}}=O_p\left(n^{-\frac{1}{2}\frac{c_8-1}{c_8+1/b_8}}+n^{-\frac{1}{2}\frac{c_1-1}{c_1+1/c_1}}\right)\!.
$$
For time-varying analysis, suppose the conditions of Theorem~\ref{theorem:consistency_planning} hold as well as Assumption~\ref{assumption:smooth_op} with $\mathcal{A}_9=\mathcal{Y}$ and $\mathcal{B}_9=\mathcal{D}\times \mathcal{D}\times \mathcal{X}\times \mathcal{X}$. 
 Then
    $$
    \sup_{d_1,d_2\in\mathcal{D}}\|\hat{\theta}^{DGF}(d_1,d_2)-\check{\theta}_0^{DGF}(d_1,d_2)\|_{\mathcal{H}_{\mathcal{Y}}}=O_p\left(n^{-\frac{1}{2}\frac{c_9-1}{c_9+1/b_9}}+n^{-\frac{1}{2}\frac{c_4-1}{c_4+1/c_4}}\right)\!.
$$
 If in addition Assumption~\ref{assumption:covariate} holds, then 
     $$
     \sup_{d_1,d_2 \in\mathcal{D}}\|\hat{\theta}^{DDS}(d_1,d_2;\tilde{P})-\check{\theta}_0^{DDS}(d_1,d_2;\tilde{P})\|_{\mathcal{H}_{\mathcal{Y}}}=O_p\left(n^{-\frac{1}{2}\frac{c_9-1}{c_9+1/b_9}}+\tilde{n}^{-\frac{1}{2}\frac{c_5-1}{c_5+1/b_5}}\right)\!.
    $$
\end{theorem}

Exact finite sample rates are given in Supplement~\ref{sec:consistency_proof}, as well as the explicit specializations of Assumption~\ref{assumption:smooth_op}. Again, these rates are at best $n^{-\frac{1}{4}}$ when $(c_1,c_4,c_5,c_8,c_9)=2$ and $(b_1,b_4,b_5,b_8,b_9)\rightarrow \infty$. Finally, we state an additional regularity condition under which we can prove that the samples $(\tilde{Y}_j)$ calculated from the distribution embeddings weakly converge to the desired distribution.
\begin{assumption}[Additional regularity]\label{assumption:regularity}
Assume
(i) $\mathcal{Y}$ is locally compact;
(ii) $\mathcal{H}_{\mathcal{Y}}\subset\mathcal{C}$, where $\mathcal{C}$ is the space of bounded, continuous, real-valued functions that vanish at infinity.
\end{assumption}
As discussed by \cite{simon2020metrizing}, the assumptions that $\mathcal{Y}$ is Polish and locally compact impose weak restrictions. In particular, if $\mathcal{Y}$ is a Banach space, then to satisfy both conditions it must be finite dimensional. We arrive at our final result of this section.
\begin{corollary}[Weak convergence; Theorem S3 of \cite{singh2020kernel}]\label{theorem:conv_dist}
Suppose the conditions of Theorem~\ref{theorem:consistency_dist} hold, as well as Assumption~\ref{assumption:regularity}. Suppose the samples $(\tilde{Y}_j)$ are calculated for $\theta_0^{DME}(d,d')$ as described in Algorithm~\ref{algorithm:herding}. Then $(\tilde{Y}_j)\leadsto  \theta_0^{DME}(d,d')$. Likewise for the other counterfactual distributions, replacing $\hat{\theta}^{DME}(d,d')$ with the other quantities in Algorithm~\ref{algorithm:dist}.
\end{corollary}
\section{Time-varying dose response with longer horizons}\label{sec:higher_T}

In the main text, we focus on the time-varying dose response curve with $T=2$ time periods. In this supplement, we generalize our approach to $T>2$ time periods. In doing so, we also clarify the role of auxiliary Markov assumptions that are common in the literature. In particular, we focus on the time-varying dose response $\theta_0^{GF}(d_{1:T})$ with and without Markov assumptions.

\subsection{A generalized decoupling}

Previously, Theorem~\ref{theorem:representation_planning} gave the RKHS representation for $T=2$. In particular, if $\gamma_0\in \mathcal{H}=\mathcal{H}_{\mathcal{D}}\otimes \mathcal{H}_{\mathcal{D}} \otimes  \mathcal{H}_{\mathcal{X}}\otimes \mathcal{H}_{\mathcal{X}} $ and if both Assumptions~\ref{assumption:RKHS} and~\ref{assumption:selection_planning} hold then
\begin{align*}
    \theta_0^{GF}(d_1,d_2)&=\left\langle \gamma_0,  \phi(d_1)\otimes\phi(d_2) \otimes \int \{\phi(x_1)\otimes \mu_{x_2}(d_1,x_1)\} \mathrm{d}P(x_1) \right\rangle_{\mathcal{H}}\!,\\ \mu_{x_2}(d_1,x_1)&=\int \phi(x_2)\mathrm{d}P(x_2|d_1,x_1).
\end{align*}
Towards a more general result, define the notation
$
\mathcal{H}_T=\{\otimes_{t=1}^T(\mathcal{H}_{\mathcal{D}})\}\otimes \{\otimes_{t=1}^T (\mathcal{H}_{\mathcal{X}})\}
$
where $\otimes_{t=1}^T (\mathcal{H}_{\mathcal{D}})$ means the $T$-times tensor product of $\mathcal{H}_{\mathcal{D}}$. Also define the condensed notation $\phi(d_{1:T})=\phi(d_1)\otimes ...\otimes \phi(d_T)$, so that the feature map of $\mathcal{H}_T$ is $\phi(d_{1:T})\otimes \phi(x_{1:T})$. 

\begin{theorem}[Decoupling via sequential kernel embeddings: $T>2$]\label{theorem:representation_planning_T}
Suppose the conditions of Lemma~\ref{lemma:id_planning} hold. Further suppose Assumption~\ref{assumption:RKHS} holds and $\gamma_0\in\mathcal{H}_T$. Then
\begin{align*}
    \theta_0^{GF}(d_{1:T})&=\left\langle \gamma_0,  \phi(d_{1:T})\otimes \mu_{1:T}\{d_{1:(T-1)}\} \right\rangle_{\mathcal{H}_T}\!,\\
    \mu_{1:T}\{d_{1:(T-1)}\}&=\int \phi(x_{1:T}) \mathrm{d}P(x_1) \prod_{t=2}^{T} \mathrm{d}P\{x_t|d_{1:(t-1)},x_{1:(t-1)}\}.
\end{align*}
\end{theorem}
See Supplement~\ref{sec:alg_deriv} for the proof. Here, $\mu_{1:T}\{d_{1:(T-1)}\}$ is the sequential kernel embedding in $T$ time periods, generalizing $\mu_{1:2}(d_1)=\int \{ \phi(x_1)\otimes \mu_{x_2}(d_1,x_1)\} \mathrm{d}P(x_1)$. What remains is an account of the structure of $\mu_{1:T}\{d_{1:(T-1)}\}$, which then implies an estimation procedure for $\hat{\mu}_{1:T}\{d_{1:(T-1)}\}$ and hence
$
 \hat{\theta}^{GF}(d_{1:T})=\langle \hat{\gamma},  \phi(d_{1:T})\otimes \hat{\mu}_{1:T}\{d_{1:(T-1)}\} \rangle_{\mathcal{H}_T}
$,
naturally generalizing Algorithm~\ref{algorithm:planning}. 

\subsection{A generalized sequential kernel embedding}

Towards this end, we provide a recursive representation of the sequential kernel embedding $\mu_{1:T}$ $\{d_{1:(T-1)}\}$ that suggests a recursive estimator $\hat{\mu}_{1:T}\{d_{1:(T-1)}\}$. Write
\begin{align*}
    &\mu_{1:T}\{d_{1:(T-1)}\} \\
    &=\int \{\phi(x_1)\otimes ...\otimes \phi(x_{T-1}) \otimes \phi(x_{T})\} \mathrm{d}P(x_1) \prod_{t=2}^{T} \mathrm{d}P\{x_t|d_{1:(t-1)},x_{1:(t-1)}\} \\
    &=\int [\phi(x_1)\otimes ...\otimes \phi(x_{T-1}) \otimes \mu_{T}\{d_{1:(T-1)},x_{1:(T-1)}\}] \mathrm{d}P(x_1) \prod_{t=2}^{T-1} \mathrm{d}P\{x_t|d_{1:(t-1)},x_{1:(t-1)}\} \\
    &=\int [\phi(x_1)\otimes ...\otimes \mu_{(T-1):T}\{d_{1:(T-1)},x_{1:(T-2)}\}] \mathrm{d}P(x_1) \prod_{t=2}^{T-2} \mathrm{d}P\{x_t|d_{1:(t-1)},x_{1:(t-1)}\},
\end{align*}
where 
\begin{align*}
    \mu_{T}\{d_{1:(T-1)},x_{1:(T-1)}\}&=\int \phi(x_T)  \mathrm{d}P\{x_T|d_{1:(T-1)},x_{1:(T-1)}\}, \\
    \mu_{(T-1):T}\{d_{1:(T-1)},x_{1:(T-2)}\} &= \int [\phi(x_{T-1}) \otimes \mu_{T}\{d_{1:(T-1)},x_{1:(T-1)}\}]\mathrm{d}P\{x_{T-1}|d_{1:(T-2)},x_{1:(T-2)}\}.
\end{align*}
Note that $\mu_{T}\{d_{1:(T-1)},x_{1:(T-1)}\}$ is a conditional kernel embedding as before, obtained by projecting $\phi(x_T)$ onto $\phi\{d_{1:(T-1)}\}\otimes \phi\{x_{1:(T-1)}\}$. Meanwhile, $\mu_{(T-1):T}\{d_{1:(T-1)},x_{1:(T-2)}\}$ is a sequential kernel embedding obtained by projecting $[\phi(x_{T-1})\otimes \mu_{T}\{d_{1:(T-1)},x_{1:(T-1)}\}]$ onto $\phi \{d_{1:(T-2)}\}\otimes \phi \{x_{1:(T-2)}\}$.

One can continue recursively defining $\mu_{t:T}\{d_{1:(T-1)},x_{1:(t-1)}\}$ all the way up to $\mu_{1:T}\{d_{1:(T-1)}\}$. This recursive construction implies a recursive estimation procedure in which $\hat{\mu}_{t:T}\{d_{1:(T-1)},x_{1:(t-1)}\}$ is estimated all the way up to $\hat{\mu}_{1:T}\{d_{1:(T-1)}\}$. Altogether, estimating the sequential kernel embedding in this way requires $T-1$ kernel ridge regressions.

\begin{remark}[Recursive regressions]
    Our nonparametric estimators of time-varying dose response curves bear some resemblance to the parametric estimators of time-varying treatment effects in \cite{bang2005doubly}. The procedure of \cite{bang2005doubly} involves recursive regressions augmented with ``clever'' covariates. Our procedure involves recursive regressions of kernel mean embeddings. Future work may formalize the connection.
\end{remark}

\subsection{Simplification via Markov assumptions}

Next, we consider the setting where $T>2$ and auxiliary Markov assumptions hold. We focus on the assumption that the covariates $X_t$ have the Markov property.
\begin{assumption}[Markov property]\label{assumption:markov}
Suppose $P\{x_t|d_{1:(t-1)},x_{1:(t-1)}\}=P(x_t|d_{t-1},x_{t-1})$ for $t\in\{2,...,T\}$.
\end{assumption}
In words, the distribution of $X_t$ only depends on the immediate history $(D_{t-1},X_{t-1})$; previous treatments and covariates are conditionally independent. Assumption~\ref{assumption:markov} simplifies the representation of $\theta_0^{GF}$.

\begin{corollary}[Decoupling via sequential kernel embeddings: Markov]\label{cor:representation_planning_Markov}
Suppose the conditions of Lemma~\ref{lemma:id_planning} hold. Further suppose Assumptions~\ref{assumption:RKHS} and~\ref{assumption:markov} hold, and $\gamma_0\in\mathcal{H}_T$. Then
\begin{align*}
    \theta_0^{GF}(d_{1:T})&=\left\langle \gamma_0,  \phi(d_{1:T})\otimes \mu_{1:T}\{d_{1:(T-1)}\} \right\rangle_{\mathcal{H}_T}\!,\\
    \mu_{1:T}\{d_{1:(T-1)}\} &=\int \phi(x_{1:T}) \mathrm{d}P(x_1) \prod_{t=2}^{T} \mathrm{d}P(x_t|d_{t-1},x_{t-1}).
\end{align*}
\end{corollary}
See Supplement~\ref{sec:alg_deriv} for the proof. Here, $\mu_{1:T}\{d_{1:(T-1)}\}$ is the sequential kernel embedding in $T$ time periods, with limited dependence across time periods. Finally, we revisit the structure of $\mu_{1:T}\{d_{1:(T-1)}\}$. Write
\begin{align*}
    \mu_{1:T}\{d_{1:(T-1)}\} &=\int \{\phi(x_1)\otimes ...\otimes \phi(x_{T-1}) \otimes \phi(x_{T})\} \mathrm{d}P(x_1) \prod_{t=2}^{T} \mathrm{d}P(x_t|d_{t-1},x_{t-1}) \\
    &=\int \{\phi(x_1)\otimes ...\otimes \phi(x_{T-1}) \otimes \mu_{T}(d_{T-1},x_{T-1})\} \mathrm{d}P(x_1) \prod_{t=2}^{T-1} \mathrm{d}P(x_t|d_{t-1},x_{t-1}) \\
    &=\int [\phi(x_1)\otimes ...\otimes \mu_{(T-1):T}\{d_{(T-2):(T-1)},x_{T-2}\} ]\mathrm{d}P(x_1) \prod_{t=2}^{T-2} \mathrm{d}P(x_t|d_{t-1},x_{t-1}),
\end{align*}
where 
\begin{align*}
    \mu_{T}(d_{T-1},x_{T-1})&=\int \phi(x_T)  \mathrm{d}P(x_T|d_{T-1},x_{T-1}), \\
    \mu_{(T-1):T}\{d_{(T-2):(T-1)},x_{T-2}\} &= \int \{\phi(x_{T-1}) \otimes \mu_{T}(d_{T-1},x_{T-1})\}\mathrm{d}P(x_{T-1}|d_{T-2},x_{T-2}).
\end{align*}
Note that $\mu_{T}(d_{T-1},x_{T-1})$ is a conditional kernel embedding as before, obtained by projecting $\phi(x_T)$ onto $\phi(d_{T-1})\otimes \phi(x_{T-1})$. Meanwhile, $\mu_{(T-1):T}\{d_{(T-2):(T-1)},x_{T-2}\}$ is a sequential kernel embedding obtained by projecting $[\phi(x_{T-1})\otimes \mu_{T}(d_{T-1},x_{T-1})]$ onto $\phi (d_{T-2})\otimes \phi (x_{T-2})$.

One can continue recursively defining $\mu_{t:T}\{d_{(t-1):(T-1)},x_{t-1}\}$ all the way up to $\mu_{1:T}\{d_{1:(T-1)}\}$. This recursive construction implies a recursive estimation procedure as before. Altogether, estimating the sequential kernel embedding in this way requires $T-1$ kernel ridge regressions. Compared to the setting without Markov assumptions, each kernel ridge regression involves fewer regressors.

\section{Identification of incremental responses}\label{sec:id}

In this supplement, we prove identification of incremental response curves under standard identifying assumptions. We also discuss the no-effect scenario.

\subsection{Mediated incremental responses}

\begin{assumption}[Selection on observables for mediation]\label{assumption:selection_mediation}
Assume
\begin{enumerate}
    \item no interference: if $D=d$ then $M=M^{(d)}$; if $D=d$ and $M=m$ then $Y=Y^{(d,m)}$;
    \item conditional exchangeability: $\{Y^{(d,m)}\}\indep D |X$, $\{M^{(d)}\}\indep D |X$, and $\{Y^{(d,m)}\}\indep M|D,X$;
    \item cross world exchangeability: $\{Y^{(d',m)}\} \indep \{M^{(d)}\} |X$;
    \item positivity: if $f(d,x)>0$ then $f(m|d,x)>0$; if $f(x)>0$ then $f(d|x)>0$, where $f(d,x)$, $f(m|d,x)$ $f(x)$, and $f(d|x)$ are densities.
\end{enumerate}
\end{assumption}

No interference is also called the stable unit treatment value assumption when there are no hidden versions of treatments. It rules out network effects, also called spillovers. Conditional exchangeability states that conditional on the covariates, the treatment assignment is as good as random. Moreover, conditional on the treatment and covariates, the mediator assignment is as good as random. Cross world exchangeability states that, conditional on covariates, the potential outcomes under one counterfactual treatment value are independent of the potential mediators under a different counterfactual treatment value. Positivity ensures that there is no covariate stratum $X=x$ such that the treatment has a restricted support, and there is no treatment-covariate stratum $(D,X)=(d,x)$ such that the mediator has a restricted support.

\begin{remark}[Cross world exchangeability]
   The cross world exchangeability assumption is the subject of ongoing debate for reasons discussed by e.g.
 \cite{robins2010alternative,richardson2013single,robins2022interventionist}. The alternative paradigm of interventional counterfactuals described in Remark~\ref{remark:controversy} avoids this assumption. Future work may identify incremental response curves in the interventional paradigm.
\end{remark}

\begin{proof}[Proof of Lemma~\ref{lemma:id_mediation}]
We extend the standard argument from $\theta_0^{ME}$ to $\theta_0^{ME,\nabla}$. For the proof, we adopt the nonseparable model notation
$
Y^{(d',m)}=Y(d',m,\eta)$ and $M^{(d)}=M(d,\eta)
$,
where $\eta$ is unobserved heterogeneity. Our innovation is to restructure the argument so that taking the derivative with respect to $d'$ does not lead to additional factors, as it otherwise would by the chain rule. We proceed in steps.

\begin{enumerate}
    \item Regression. By definition,
    $$
    \gamma_0(d',m,x)=E(Y|D=d',M=m,X=x)=\int Y(d',m,\eta)\mathrm{d}P(\eta|d',m,x).
    $$
    By the assumed independences, 
    $
    P(\eta|d',m,x)=P(\eta|d',x)=P(\eta|x).
    $
    In summary,
    $$
      \gamma_0(d',m,x)=\int Y(d',m,\eta)\mathrm{d}P(\eta|x),\quad 
    \nabla_{d'}\gamma_0(d',m,x)=\int \nabla_{d'}Y(d',m,\eta)\mathrm{d}P(\eta|x).
    $$
    \item Target expression. Beginning with the desired expression,
   \begin{align*}
        RHS&=\int \nabla_{d'}\gamma_0(d',m,x)\mathrm{d}P(m|d,x)\mathrm{d}P(x) \\
        &=\int \nabla_{d'}Y(d',m,\eta)\mathrm{d}P(\eta|x)\mathrm{d}P(m|d,x)\mathrm{d}P(x).
   \end{align*}
    Note that
    $
    P(m|d,x)=P\{M(d,\eta)|d,x\}=P\{M(d,\eta)|x\}.
    $
    Since $\{M^{(d)},Y^{(d',m)}\}\indep D|X$ implies $(Y^{d',m})\indep D|\{M^{(d)}\},X$
    $$
    P(\eta|x)=P(\eta|d,x)=P\{\eta|d,M(d,\eta),x\}=P\{\eta|M(d,\eta),x\}.
    $$
    In summary,
    $$
    RHS=\int \nabla_{d'}Y\{d',M(d,\eta),\eta\}\mathrm{d}P\{\eta|M(d,\eta),x\}\mathrm{d}P\{M(d,\eta)|x\}\mathrm{d}P(x).
    $$
    Conveniently,
    \begin{align*}
         &\int P\{\eta|M(d,\eta),x\}\mathrm{d}P\{M(d,\eta)|x\}=\int  P\{\eta, M(d,\eta)|x\}\mathrm{d}\{M(d,\eta)\}= P(\eta|x);\\
    &P(\eta|x)P(x)=P(\eta,x).
    \end{align*}
    Therefore
    $$
    RHS=\int \nabla_{d'}Y\{d',M(d,\eta),\eta\}\mathrm{d}P(\eta,x)=\int \nabla_{d'}Y\{d',M(d,\eta),\eta\}\mathrm{d}P(\eta)=LHS.
    $$
\end{enumerate}
\end{proof}

\subsection{Time-varying incremental responses}

\begin{assumption}[Sequential selection on observables]\label{assumption:selection_planning}
Assume
\begin{enumerate}
    \item no interference: if $D_{1:T}=d_{1:T}$ then $Y=Y^{(d_{1:T})}$;
    \item conditional exchangeability: $\{Y^{(d_{1:T})}\} \indep D_t|D_{1:(t-1)},X_{1:t}$,  for all $ 1\leq t\leq T$;
    \item positivity: if $f\{d_{1:(t-1)},x_{1:t}\}>0$ then $f\{d_t|d_{1:(t-1)},x_{1:t}\}>0$, where $f\{d_{1:(t-1)},x_{1:t}\}$ and $f\{d_t|d_{1:(t-1)},x_{1:t}\}$ are densities.
\end{enumerate}
\end{assumption}

Assumption~\ref{assumption:selection_planning} is a sequential generalization of the standard selection on observables \cite{rosenbaum1983central}. To handle $\theta_0^{DS}$, we use a standard assumption in transfer learning \cite{pearl2014external}.

\begin{assumption}[Distribution shift]\label{assumption:covariate}
Assume
\begin{enumerate}
    \item $\tilde{P}(Y,D_{1:T},X_{1:T})=P(Y|D_{1:T},X_{1:T})\tilde{P}(D_{1:T},X_{1:T})$;
    \item $\tilde{P}(D_{1:T},X_{1:T})$ is absolutely continuous with respect to $P(D_{1:T},X_{1:T})$.
\end{enumerate}
\end{assumption}
The difference between $P$ and $\tilde{P}$ is only in the distribution of the treatments and covariates. Moreover, the support of $P$ contains the support of $\tilde{P}$.  An immediate consequence is that the regression function $\gamma_0(d_{1:T},x_{1:T})=E(Y|D_{1:T}=d_{1:T},X_{1:T}=x_{1:T})$ remains the same across the different populations $P$ and $\tilde{P}$.

\begin{proof}[Proof of Lemma~\ref{lemma:id_planning}]
We extend the standard argument from $\theta_0^{GF}$ to $\theta_0^{GF,\nabla}$. For clarity, we focus on the case with $T=2$ and we adopt the nonseparable model notation
$
Y^{(d_1,d_2)}=Y(d_1,d_2,\eta)
$,
where $\eta$ is unobserved heterogeneity. As before, our innovation is to restructure the argument so that taking the derivative with respect to $d_2$ does not lead to additional factors, as it otherwise would by chain rule. The argument for $\theta_0^{DS}$ is identical. We proceed in steps.

\begin{enumerate}
    \item Regression. 
    By definition,
    \begin{align*}
         \gamma_0(d_1,d_2,x_1,x_2)
         &=E(Y|D_1=d_1,D_2=d_2,X_1=x_1,X_2=x_2) \\
         &=\int Y(d_1,d_2,\eta)\mathrm{d}P(\eta|d_1,d_2,x_1,x_2).
    \end{align*}
    By the assumed independences, 
    $
    P(\eta|d_1,d_2,x_1,x_2)=P(\eta|d_1,x_1,x_2).
    $
    In summary,
    \begin{align*}
     \gamma_0(d_1,d_2,x_1,x_2)
     &=\int Y(d_1,d_2,\eta)\mathrm{d}P(\eta|d_1,x_1,x_2),\\
    \nabla_{d_2}\gamma_0(d_1,d_2,x_1,x_2)&=\int \nabla_{d_2}Y(d_1,d_2,\eta)\mathrm{d}P(\eta|d_1,x_1,x_2).
     \end{align*}
    \item Target expression. 
    Beginning with the desired expression,
 \begin{align*}
     RHS&=\int \nabla_{d_2}\gamma_0(d_1,d_2,x_1,x_2)\mathrm{d}P(x_1)\mathrm{d}P(x_2|d_1,x_1) \\
        &=\int \nabla_{d_2}Y(d_1,d_2,\eta)\mathrm{d}P(\eta|d_1,x_1,x_2)\mathrm{d}P(x_1)\mathrm{d}P(x_2|d_1,x_1).
     \end{align*}
    Conveniently,
    \begin{align*}
    &\int P(\eta|d_1,x_1,x_2)\mathrm{d}P(x_2|d_1,x_1)= \int P(x_2,\eta|d_1,x_1)\mathrm{d}x_2=P(\eta|x_1); \\
   &P(\eta|x_1)P(x_1)=P(\eta,x_1).
     \end{align*}
    Therefore
    $$
    RHS=\int \nabla_{d_2}Y(d_1,d_2,\eta)\mathrm{d}P(\eta,x_1)=\int \nabla_{d_2}Y(d_1,d_2,\eta)\mathrm{d}P(\eta)=LHS.
    $$
\end{enumerate}
\end{proof}

\subsection{No effect of second dose}

Finally, we provide a formal result to support Remark~\ref{remark:no}

\begin{lemma}[No effect of second dose]
    Suppose Assumption~\ref{assumption:selection_planning} holds and that the second dose has no effect, i.e. 
    $E\{Y^{(d_1,d_2)}|D_1=d_1,X_1=x_1,X_2=x_2\}=E\{Y^{(d_1)}|D_1=d_1,X_1=x_1,X_2=x_2\}$. Then $\gamma_0(d_1,d_2,x_1,x_2)=E(Y|D_1=d_1,X_1=x_1,X_2=x_2)$ and $\theta_0^{GF}(d_1,d_2)=\int E(Y|D=d_1,X=x_1)\mathrm{d}P(x_1)$.
\end{lemma}

\begin{proof}
  By definition of $\gamma_0$, no interference, conditional exchangeability, no effect, and no interference,
\begin{align*}
    \gamma_0(d_1,d_2,x_1,x_2)
    &=E(Y|D_1=d_1,D_2=d_2,X_1=x_1,X_2=x_2) \\
    &=E\{Y^{(d_1,d_2)}|D_1=d_1,D_2=d_2,X_1=x_1,X_2=x_2\} \\
    &=E\{Y^{(d_1,d_2)}|D_1=d_1,X_1=x_1,X_2=x_2\} \\
     &=E\{Y^{(d_1)}|D_1=d_1,X_1=x_1,X_2=x_2\} \\
     &=E(Y|D_1=d_1,X_1=x_1,X_2=x_2).
\end{align*}
Appealing to Lemma~\ref{lemma:id_planning}, this result, and the law of iterated expectations,
\begin{align*}
    \theta_0(d_1,d_2)
    &=\int \gamma_0(d_1,d_2,x_1,x_2)\mathrm{d}P(x_2|d_1,x_1)\mathrm{d}P(x_1) \\
    &=\int E(Y|D_1=d_1,X_1=x_1,X_2=x_2)\mathrm{d}P(x_2|d_1,x_1)\mathrm{d}P(x_1) \\
    &=\int E(Y|D=d_1,X=x_1)\mathrm{d}P(x_1).
\end{align*}
\end{proof}
\section{Algorithm derivation}\label{sec:alg_deriv}

In this supplement, we derive the closed form solutions for (i) mediated responses, (ii) time-varying dose responses, and (iii) counterfactual distributions. We also provide alternative closed form expressions that run faster in statistical software.

\subsection{Mediated response curves}

\begin{proof}[Proof of Theorem~\ref{theorem:representation_mediation}]
In Assumption~\ref{assumption:RKHS}, we impose that the scalar kernels are bounded. This assumption has several implications. First, the feature maps are Bochner integrable \cite[Definition A.5.20]{steinwart2008support}. Bochner integrability permits us to interchange the expectation and inner product. Second, the kernel embeddings exist. Third, the product kernel is also bounded and hence the tensor product RKHS inherits these favorable properties. Therefore
\begin{align*}
    \theta_0^{ME}(d,d')&= \int \gamma_0(d',m,x) \mathrm{d}P(m|d,x)\mathrm{d}P(x)\\
    &=\int \langle \gamma_0, \phi(d')\otimes \phi(m)\otimes \phi(x)\rangle_{\mathcal{H}}  \mathrm{d}P(m|d,x)\mathrm{d}P(x) \\
    &= \int \left\langle \gamma_0, \phi(d')\otimes \int \phi(m)\mathrm{d}P(m|d,x) \otimes \phi(x)\right\rangle_{\mathcal{H}}  \mathrm{d}P(x) \\
    &=\int\left\langle \gamma_0, \phi(d')\otimes \mu_{m}(d,x) \otimes  \phi(x) \right\rangle_{\mathcal{H}}\mathrm{d}P(x)  \\
    &=\left\langle \gamma_0, \phi(d')\otimes \int \{\mu_m(d,x) \otimes \phi(x)\}\mathrm{d}P(x) \right\rangle_{\mathcal{H}}\!.
\end{align*}
By \cite[Lemma 4.34]{steinwart2008support}, the derivative feature map $\nabla_d\phi(d)$ exists, is continuous, and is Bochner integrable since $\kappa_d'<\infty$. Therefore the derivation remains valid for incremental responses.
\end{proof}

\begin{proof}[Derivation of Algorithm~\ref{algorithm:mediation}]
By standard arguments \cite{kimeldorf1971some}
\begin{align*}
\hat{\gamma}(d,m,x)&=
\langle \hat{\gamma}, \phi(d)\otimes \phi(m)\otimes \phi(x) \rangle_{\mathcal{H}}\\
&=Y^{\top}(K_{DD}\odot K_{MM}\odot K_{XX}+n\lambda I)^{-1}(K_{Dd}\odot K_{Mm}\odot K_{Xx}).  
\end{align*}
By \cite[Algorithm 1]{singh2019kernel}, write the conditional kernel embedding as 
$$
    \hat{\mu}_{m}(d,x)=K_{\cdot M}(K_{DD}\odot K_{XX}+n\lambda_1 I)^{-1}(K_{Dd}\odot K_{Xx}).
    $$
    Therefore
    $$
    \frac{1}{n}\sum_{i=1}^n[\hat{\mu}_{m}(d,X_i)\otimes \phi(X_i)]=\frac{1}{n}\sum_{i=1}^n[\{K_{\cdot M}(K_{DD}\odot K_{XX}+n\lambda_1 I)^{-1}(K_{Dd}\odot K_{Xx_i})\}\otimes \phi(X_i)],
    $$
    and
    \begin{align*}
        \hat{\theta}^{ME}(d,d')
        &=\langle \hat{\gamma}, \phi(d')\otimes \frac{1}{n}\sum_{i=1}^n[\hat{\mu}_m(d,X_i)\otimes \phi(X_i)]  \rangle_{\mathcal{H}} \\
        &=\frac{1}{n}\sum_{i=1}^n Y^{\top}(K_{DD}\odot K_{MM}\odot K_{XX}+n\lambda I)^{-1} \\
       &\quad [K_{Dd'}\odot  \{K_{MM}(K_{DD}\odot K_{XX}+n\lambda_1 I)^{-1}(K_{Dd}\odot K_{Xx_i})\}\odot K_{Xx_i}].   
    \end{align*}
    For incremental responses, simply replace $\hat{\gamma}(d,m,x)$ with
\begin{align*}
\nabla_d\hat{\gamma}(d,m,x)&=
\langle \hat{\gamma}, \nabla_d \phi(d)\otimes \phi(m) \otimes \phi(x) \rangle_{\mathcal{H}} \\
&=Y^{\top}(K_{DD}\odot K_{MM} \odot K_{XX}+n\lambda I)^{-1}(\nabla_d K_{Dd}\odot K_{Mm}\odot K_{Xx}).
\end{align*}
\end{proof}

\begin{algo}[Details for Algorithm~\ref{algorithm:mediation_dml}]
For simplicity, we abstract from sample splitting. Denote the kernel matrices by
$
K_{DD}, K_{MM}, K_{XX}\in\mathbb{R}^{n\times n}
$
. Then
\begin{enumerate}
    \item $\hat{\gamma}(d,m,x)=Y^{\top}(K_{DD}\odot K_{MM}\odot K_{XX}+n\lambda I)^{-1}(K_{Dd}\odot K_{Mm}\odot K_{Xx}) $,
    \item $\hat{\pi}(x)=D^{\top}(K_{XX}+n\lambda_2 I)^{-1}K_{Xx}$,
    \item $\hat{\rho}(m,x)=D^{\top}(K_{MM}\odot K_{XX}+n\lambda_3 I)^{-1}(K_{Mm}\odot K_{Xx}) $.
\end{enumerate}
\end{algo}

\subsection{Time-varying dose response curves}

\begin{proof}[Proof of Theorem~\ref{theorem:representation_planning}]
Assumption~\ref{assumption:RKHS} implies Bochner integrability, which permits us to interchange the expectation and inner product. Therefore
\begin{align*}
    \theta_0^{GF}(d_1,d_2)
    &= \int \gamma_0(d_1,d_2,x_1,x_2) \mathrm{d}P(x_2|d_1,x_1)\mathrm{d}P(x_1)\\
    &=\int \langle \gamma_0, \phi(d_1)\otimes\phi(d_2) \otimes \phi(x_1)\otimes \phi(x_2)\rangle_{\mathcal{H}}  \mathrm{d}P(x_2|d_1,x_1)\mathrm{d}P(x_1) \\
    &= \int \left\langle \gamma_0,  \phi(d_1)\otimes\phi(d_2) \otimes \phi(x_1)\otimes \int \phi(x_2) \mathrm{d}P(x_2|d_1,x_1)\right\rangle_{\mathcal{H}}  \mathrm{d}P(x_1)\\
    &=\int\left\langle \gamma_0,  \phi(d_1)\otimes\phi(d_2) \otimes \phi(x_1)\otimes \mu_{x_2}(d_1,x_1) \right\rangle_{\mathcal{H}}\mathrm{d}P(x_1)\\
    &=\left\langle \gamma_0,  \phi(d_1)\otimes\phi(d_2) \otimes \int \{\phi(x_1)\otimes \mu_{x_2}(d_1,x_1)\} \mathrm{d}P(x_1) \right\rangle_{\mathcal{H}}\!.
\end{align*}
The argument for $\theta_0^{DS}$ is identical.
\end{proof}

\begin{proof}[Proof of Theorem~\ref{theorem:representation_planning_T}]
Assumption~\ref{assumption:RKHS} implies Bochner integrability, which permits us to interchange the expectation and inner product. Therefore
\begin{align*}
    \theta_0^{GF}(d_{1:T})
    &= \int \gamma_0(d_{1:T},x_{1:T}) \mathrm{d}P(x_1) \prod_{t=2}^{T} \mathrm{d}P\{x_t|d_{1:(t-1)},x_{1:(t-1)}\}\\
    &=\int \langle \gamma_0, \phi(d_{1:T}) \otimes \phi(x_{1:T})\rangle_{\mathcal{H}_T}  \mathrm{d}P(x_1) \prod_{t=2}^{T} \mathrm{d}P\{x_t|d_{1:(t-1)},x_{1:(t-1)}\} \\
    &= \left\langle \gamma_0,  \phi(d_{1:T})  \otimes \int \phi(x_{1:T}) \mathrm{d}P(x_1) \prod_{t=2}^{T} \mathrm{d}P\{x_t|d_{1:(t-1)},x_{1:(t-1)}\} \right\rangle_{\mathcal{H}_T} \\
    &=\langle \gamma_0,  \phi(d_{1:T}) \otimes \mu_{1:T}\{d_{1:(T-1)}\}  \rangle_{\mathcal{H}_T}.
\end{align*}
\end{proof}

\begin{proof}[Proof of Corollary~\ref{cor:representation_planning_Markov}]
The result is immediate from Theorem~\ref{theorem:representation_planning_T} and Assumption~\ref{assumption:markov}.
\end{proof}

\begin{proof}[Derivation of Algorithm~\ref{algorithm:planning}]
By standard arguments \cite{kimeldorf1971some}
\begin{align*}
&\hat{\gamma}(d_1,d_2,x_1,x_2) \\
&=
\langle \hat{\gamma}, \phi(d_1)\otimes \phi(d_2)\otimes \phi(x_1)\otimes \phi(x_2) \rangle_{\mathcal{H}}\\
&=Y^{\top}(K_{D_1D_1}\odot K_{D_2D_2}\odot K_{X_1X_1}\odot K_{X_2X_2}+n\lambda I)^{-1} (K_{D_1d_1}\odot K_{D_2d_2}\odot  K_{X_1x_1}\odot K_{X_2x_2}).    \end{align*}
By \cite[Algorithm 1]{singh2019kernel}, write the conditional kernel embedding as 
$$
    \hat{\mu}_{x_2}(d_1,x_1)=K_{\cdot X_2}(K_{D_1D_1}\odot K_{X_1X_1}+n\lambda_4 I)^{-1}(K_{D_1d_1}\odot K_{X_1x_1}).
    $$
    Therefore
    \begin{align*}
        &\frac{1}{n}\sum_{i=1}^n[\phi(X_{1i}) \otimes \hat{\mu}_{x_2}(d_1,X_{1i}) ]\\
        &=\frac{1}{n}\sum_{i=1}^n[\phi(X_{1i})\otimes\{K_{\cdot X_2}(K_{D_1D_1}\odot K_{X_1X_1}+n\lambda_4 I)^{-1}(K_{D_1d_1}\odot K_{X_1x_{1i}})\} ]
    \end{align*}
    and
    \begin{align*}
        &\hat{\theta}^{GF}(d_1,d_2)\\
        &=\langle \hat{\gamma}, \phi(d_1)\otimes\phi(d_2)\otimes \frac{1}{n}\sum_{i=1}^n[\phi(X_{1i})\otimes \hat{\mu}_{x_2}(d_1,X_{1i})]  \rangle_{\mathcal{H}} \\
        &=\frac{1}{n}\sum_{i=1}^n Y^{\top}(K_{D_1D_1}\odot K_{D_2D_2}\odot K_{X_1X_1}\odot K_{X_2X_2}+n\lambda I)^{-1} \\
        &\quad [K_{D_1d_1}\odot K_{D_2d_2}\odot  K_{X_1x_{1i}}\odot \{K_{X_2 X_2}(K_{D_1D_1}\odot K_{X_1X_1}+n\lambda_4 I)^{-1} (K_{D_1d_1}\odot K_{X_1x_{1i}})\}].
    \end{align*}
     The argument for $\theta_0^{DS}$ is identical.
\end{proof}

\begin{algo}[Details for Algorithm~\ref{algorithm:planning_dml}]
For simplicity, we abstract from sample splitting. Denote the kernel matrices by
$
K_{D_1D_1},K_{D_2D_2},K_{X_1X_1}, K_{X_2X_2}\in\mathbb{R}^{n\times n}
$
calculated from observations drawn from the population $P$. Then
\begin{enumerate}
    \item $\hat{\gamma}(d_1,d_2,x_1,x_2)=Y^{\top}(K_{D_1D_1}\odot K_{D_2D_2}\odot K_{X_1X_1}\odot K_{X_2X_2}+n\lambda I)^{-1}(K_{D_1d_1}\odot K_{D_2d_2}\odot K_{X_1x_1}\odot K_{X_2x_2}) $,
    \item $\hat{\pi}(x_1)=D_1^{\top}(K_{X_1X_1}+n\lambda_6 I)^{-1}K_{X_1x_1}$,
    \item $\hat{\rho}(d_1,x_1,x_2)=D_2^{\top}(K_{D_1D_1}\odot K_{X_1X_1}\odot K_{X_2X_2}+n\lambda_7 I)^{-1}(K_{D_1d_1}\odot K_{X_1x_1}\odot K_{X_2x_2}) $.
\end{enumerate}
\end{algo}

\subsection{Counterfactual distributions}

\begin{proof}[Proof of Theorem~\ref{theorem:representation_dist}]
The argument is analogous to those in Theorems~\ref{theorem:representation_mediation} and~\ref{theorem:representation_planning}, recognizing that in the distributional setting
\begin{align*}
    \gamma_0(d,m,x)&= E_8^*\{\phi(d)\otimes \phi(m) \otimes \phi(x)\}, \\
     \gamma_0(d_1,d_2,x_1,x_2)&= E_9^*\{\phi(d_1)\otimes \phi(d_2) \otimes \phi(x_1)\otimes \phi(x_2)\}.
\end{align*}
\end{proof}

\begin{proof}[Derivation of Algorithm~\ref{algorithm:dist}]
The argument is analogous to those in Algorithms~\ref{algorithm:mediation} and~\ref{algorithm:planning} appealing to \cite[Algorithm 1]{singh2019kernel}, which gives
\begin{align*}
    \hat{\gamma}(d,m,x)
    &=\hat{E}_8^*[\phi(d)\otimes \phi(m)\otimes  \phi(x)] \\
    &=K_{\cdot Y}(K_{DD}\odot K_{MM}\odot K_{XX}+n\lambda_8 I)^{-1}(K_{Dd}\odot K_{Mm}\odot K_{Xx}),\\
    \hat{\gamma}(d_1,d_2,x_1,x_2)&=\hat{E}_9^*\{\phi(d_1)\otimes \phi(d_2) \otimes \phi(x_1) \otimes \phi(x_2)\}\\
       &=K_{\cdot Y}(K_{D_1D_1}\odot K_{D_2D_2} \odot  K_{X_1X_1}\odot K_{X_2X_2} +n\lambda_9 I)^{-1}\\
       &\quad (K_{D_1d_1}\odot K_{D_2d_2}\odot K_{X_1x_1}\odot K_{X_2x_2}).
\end{align*}
\end{proof}

\subsection{Alternative closed form}

While the closed form expressions above are intuitive, alternative closed form expressions may run faster in statistical software. The expressions for $\hat{\theta}^{ME}$ in Algorithm~\ref{algorithm:mediation} and for $\hat{\theta}^{GF}$ and $\hat{\theta}^{DS}$ in Algorithm~\ref{algorithm:planning} require the summation over $n$ terms, each of which is a vector in $\mathbb{R}^n$. This step is a computational bottleneck since looping is much slower than matrix multiplication in many modern programming languages such as \verb|MATLAB| and \verb|Python|. Here, we present alternative closed form expressions which replace the summations with matrix multiplications. In experiments, these alternative expressions run about $10^2$ times faster. Although both the original and alternative expressions involve $O(n^3)$ computations, the alternative expressions run faster in numerical packages due to the efficient implementation of matrix multiplication.
 
We present a lemma to derive the alternative expressions. Let $\1_n \in \mathbb{R}^n$ be vector of ones.
\begin{lemma}[Matrix and elementwise products] \label{lem:matrix_form}
If $A \in \mathbb{R}^{n\times n}$ and $a,b \in \mathbb{R}^n$ then
$$
A(a \odot b) = \{A \odot (\1_n a^\top)\}b,\quad (Aa) \odot b = \{A \odot (b \1_n^\top)\}a.
$$
\end{lemma}

\begin{proof}
The result is immediate from the definitions of the matrix and elementwise products.
\end{proof}

Equipped with Lemma~\ref{lem:matrix_form}, we provide alternative expressions for $\hat{\theta}^{ME}$, $\hat{\theta}^{GF}$, and $\hat{\theta}^{DS}$.

\begin{algo}[Alternative estimation of mediated response curves]\label{algorithm:mediation2}
Denote the kernel matrices by
$
K_{DD}, K_{MM}, K_{XX}\in\mathbb{R}^{n\times n}
$. Mediated response estimators have closed forms based on
\begin{align*}
    \hat{\theta}^{ME}(d,d')
    &=Y^{\top}(K_{DD}\odot K_{MM}\odot K_{XX}+n\lambda I)^{-1} \\
    &\quad \left[K_{Dd'}\odot  \left\{K_{MM}(K_{DD}\odot K_{XX}+n\lambda_1 I)^{-1} \odot \frac1n K^2_{XX} \right\} K_{Dd}\right]\!, 
\end{align*}
where $(\lambda,\lambda_1)$ are ridge regression penalty parameters.
\end{algo}
\begin{proof}
Denote
\begin{align*}
    R_1 &= Y^{\top}(K_{DD}\odot K_{MM}\odot K_{XX}+n\lambda I)^{-1},\quad R_2 = K_{MM}(K_{DD}\odot K_{XX}+n\lambda_1 I )^{-1}.
\end{align*}
Then by Algorithm~\ref{algorithm:mediation} and Lemma~\ref{lem:matrix_form}, we have
\begin{align*}
\hat{\theta}^{ME}(d,d')
    &= \frac1n \sum_{i=1}^n R_1 \{K_{Dd'}\odot  R_2(K_{Xx_i} \odot K_{Dd}) \odot K_{Xx_i}\}\\
    &= \frac1n \sum_{i=1}^n R_1 [K_{Dd'}\odot  \{(R_2 \odot \1_n K^\top_{Xx_i}) K_{Dd}\} \odot K_{Xx_i}]\\
    &= \frac1n \sum_{i=1}^n R_1 \{K_{Dd'}\odot  (R_2 \odot \1_n K^\top_{Xx_i} \odot K_{Xx_i} \1^\top_n) K_{Dd}\}\\
    &= R_1 \left\{K_{Dd'}\odot  \left(R_2 \odot \frac1n \sum_{i=1}^n K_{Xx_i} K^\top_{Xx_i}\right) K_{Dd}\right\} \\
    &= R_1 \left\{K_{Dd'}\odot  \left(R_2 \odot \frac1n K_{XX}^2\right) K_{Dd}\right\}\!.
\end{align*}
Note that we use the identity
$
    \1_n K^\top_{Xx_i} \odot K_{Xx_i} \1^\top_n =  K_{Xx_i} K_{Xx_i}^\top.
$
\end{proof}

\begin{algo}[Alternative estimation of time-varying response curves]\label{algorithm:planning2}
Denote the kernel matrices by
$
K_{D_1D_1}$, $K_{D_2D_2}$, $K_{X_1X_1}$, $K_{X_2X_2}\in\mathbb{R}^{n\times n}
$
calculated from observations drawn from $P$.  Denote the kernel matrices
$
K_{\tilde{D}_1\tilde{D}_1}$, $K_{\tilde{D}_2\tilde{D}_2}$, $K_{\tilde{X}_1\tilde{X}_1}$, $K_{\tilde{X}_2\tilde{X}_2}\in\mathbb{R}^{n\times n}
$
calculated from observations drawn from $\tilde{P}$.  Time-varying response curve estimators have the closed form solutions
\begin{enumerate}
    \item $\hat{\theta}^{GF}= Y^{\top}(K_{D_1D_1}\odot K_{D_2D_2}\odot K_{X_1X_1}\odot K_{X_2X_2}+n\lambda I)^{-1}$ \\ $
        [K_{D_1d_1}\odot K_{D_2d_2}\odot \{K_{X_2 X_2}(K_{D_1D_1}\odot K_{X_1X_1}+n\lambda_4 I)^{-1} \odot \frac1n K^2_{X_1X_1}\}K_{D_1d_1}]  $,
    \item $\hat{\theta}^{DS}= Y^{\top}(K_{D_1D_1}\odot K_{D_2D_2}\odot K_{X_1X_1}\odot K_{X_2X_2}+n\lambda I)^{-1}$ \\$
        [K_{D_1d_1}\odot K_{D_2d_2}\odot  \{K_{X_2 \tilde{X}_2}(K_{\tilde{D}_1\tilde{D}_1}\odot K_{\tilde{X}_1\tilde{X}_1}+n\lambda_5 I)^{-1} \odot \frac1n (K_{X_1 \tilde X_1}K_{\tilde X_1 \tilde X_1})\}K_{\tilde{D}_1d_1}] $,
\end{enumerate}
where $(\lambda,\lambda_4,\lambda_5)$ are ridge regression penalty parameters.
\end{algo}

\begin{proof}
The argument is analogous to the derivation of Algorithm~\ref{algorithm:mediation2}.
\end{proof}
\section{Nonparametric consistency proofs}\label{sec:consistency_proof}

In this supplement, we (i) present an equivalent definition of smoothness and specialize the smoothness condition in various settings; (ii) present technical lemmas for regression, unconditional kernel embeddings, and conditional kernel embeddings; (iii) appeal to these lemmas to prove uniform consistency for mediation analysis and time-varying dose response curves as well as convergence in distribution for counterfactual distributions.

\subsection{Representations of smoothness}

\textbf{Alternative representations.}
Consider the notation of Section~3. 

\begin{lemma}[Smoothness; Remark 2 of \cite{caponnetto2007optimal}]
If the input measure and Mercer measure are the same then there are equivalent formalisms for the smoothness conditions in Assumptions~\ref{assumption:smooth_gamma} and~\ref{assumption:smooth_op}.
\begin{enumerate}
    \item Smoothness in Assumption~\ref{assumption:smooth_gamma} holds if and only if the regression $\gamma_0$ is a particularly smooth element of $\mathcal{H}$. Formally, define the covariance operator $T$ for $\mathcal{H}$.
    We assume there exists $ g\in \mathcal{H}$ such that $\gamma_0=T^{\frac{c-1}{2}}g$, $c\in(1,2]$, and $\|g\|^2_{\mathcal{H}}\leq\zeta$.
    \item Smoothness in Assumption~\ref{assumption:smooth_op} holds if and only if the conditional expectation operator $E_{\ell}$ is a particularly smooth element of $\mathcal{L}_2(\mathcal{H}_{\mathcal{A}_{\ell}},\mathcal{H}_{\mathcal{B}_{\ell}})$. Formally, define the covariance operator $T_{\ell}=E\{\phi(B_{\ell})\otimes \phi(B_{\ell})\}$ for $\mathcal{L}_2(\mathcal{H}_{\mathcal{A}_{\ell}},\mathcal{H}_{\mathcal{B}_{\ell}})$.
    We assume there exists $ G_{\ell}\in \mathcal{L}_2(\mathcal{H}_{\mathcal{A}_{\ell}},\mathcal{H}_{\mathcal{B}_{\ell}})$ such that $E_{\ell}=(T_{\ell})^{\frac{c_{\ell}-1}{2}}\circ G_{\ell}$, $c_{\ell}\in(1,2]$, and $\|G_{\ell}\|^2_{\mathcal{L}_2(\mathcal{H}_{\mathcal{A}_{\ell}},\mathcal{H}_{\mathcal{B}_{\ell}})}\leq\zeta_{\ell}$.
\end{enumerate}
\end{lemma}

\begin{remark}[Assumption~\ref{assumption:smooth_gamma}]
The covariance operator $T$ for the RKHS $\mathcal{H}$ depends on the setting;
\begin{enumerate}
    \item for mediation analysis: $T=E[\{\phi(D)\otimes \phi(M)\otimes \phi(X)\}\otimes \{\phi(D)\otimes \phi(M)\otimes \phi(X)\}]$;
    \item for time-varying analysis: $T=E[\{\phi(D_1)\otimes \phi(D_2) \otimes \phi(X_1)\otimes \phi(X_2)\}\otimes \{\phi(D_1)\otimes \phi(D_2) \otimes \phi(X_1)\otimes \phi(X_2)\}]$.
\end{enumerate}
\end{remark}

\begin{remark}[Assumption~\ref{assumption:smooth_op}]
\cite{singh2019kernel} show that $T_{\ell}$ and its powers are well defined under Assumption~\ref{assumption:RKHS}. 
The spaces $\mathcal{A}_{\ell}$ and $\mathcal{B}_{\ell}$ depend on the setting.
\begin{enumerate}
    \item For mediation analysis,
    \begin{enumerate}
        \item in the kernel embedding $\mu_m(d,x)$: $\mathcal{A}_1=\mathcal{M}$ and $\mathcal{B}_1=\mathcal{D} \times \mathcal{X}$;
        \item in the regression operator $E_8$: $\mathcal{A}_8=\mathcal{Y}$ and $\mathcal{B}_8=\mathcal{D} \times \mathcal{M} \times \mathcal{X}$.
    \end{enumerate}
    \item For time-varying analysis,
    \begin{enumerate}
        \item in the kernel embedding $\mu_{x_2}(d_1,x_1)$: $\mathcal{A}_4=\mathcal{X}$ and $\mathcal{B}_4=\mathcal{D} \times \mathcal{X}$;
        \item in the kernel embedding $\nu_{x_2}(d_1,x_1)$: $\mathcal{A}_5=\mathcal{X}$ and $\mathcal{B}_5=\mathcal{D} \times \mathcal{X}$;
        \item in the regression operator $E_9$: $\mathcal{A}_9=\mathcal{Y}$ and $\mathcal{B}_9=\mathcal{D} \times \mathcal{D} \times \mathcal{X} \times \mathcal{X}$.
    \end{enumerate}
\end{enumerate}
\end{remark}

\textbf{Explicit specializations.}
\begin{assumption}[Smoothness of the kernel embedding $\mu_m(d,x)$]\label{assumption:smooth_ME}
Assume the following.
\begin{enumerate}
\item The conditional expectation operator $E_1$ is well specified as a Hilbert--Schmidt operator between RKHSs, i.e. $E_1\in \mathcal{L}_2(\mathcal{H}_{\mathcal{M}},\mathcal{H}_{\mathcal{D}}\otimes \mathcal{H}_{\mathcal{X}})$, where
    $
    E_1:\mathcal{H}_{\mathcal{M}} \rightarrow \mathcal{H}_{\mathcal{D}}\otimes \mathcal{H}_{\mathcal{X}},\; f(\cdot)\mapsto E\{f(M)|D=\cdot,X=\cdot\}
    $.
    \item The conditional expectation operator is a particularly smooth element of $\mathcal{L}_2(\mathcal{H}_{\mathcal{M}},\mathcal{H}_{\mathcal{D}}\otimes \mathcal{H}_{\mathcal{X}})$. Formally, define the covariance operator $T_1=E[\{\phi(D)\otimes \phi(X)\} \otimes \{\phi(D)\otimes \phi(X)\}]$ for $\mathcal{L}_2(\mathcal{H}_{\mathcal{M}},\mathcal{H}_{\mathcal{D}}\otimes \mathcal{H}_{\mathcal{X}})$.
    We assume there exists $ G_1\in \mathcal{L}_2(\mathcal{H}_{\mathcal{M}},\mathcal{H}_{\mathcal{D}}\otimes \mathcal{H}_{\mathcal{X}})$ such that $E_1=(T_1)^{\frac{c_1-1}{2}}\circ G_1$, $c_1\in(1,2]$, and $\|G_1\|^2_{\mathcal{L}_2(\mathcal{H}_{\mathcal{M}},\mathcal{H}_{\mathcal{D}}\otimes \mathcal{H}_{\mathcal{X}})}\leq\zeta_1$.
    \end{enumerate}
\end{assumption}

\begin{assumption}[Smoothness of the kernel embedding $\mu_{x_2}(d_1,x_1)$]\label{assumption:smooth_SATE}
Assume the following.
\begin{enumerate}
\item The conditional expectation operator $E_4$ is well specified as a Hilbert--Schmidt operator between RKHSs, i.e. $E_4\in \mathcal{L}_2(\mathcal{H}_{\mathcal{X}},\mathcal{H}_{\mathcal{D}}\otimes \mathcal{H}_{\mathcal{X}})$, where
    $
    E_4:\mathcal{H}_{\mathcal{X}} \rightarrow \mathcal{H}_{\mathcal{D}}\otimes \mathcal{H}_{\mathcal{X}},\; f(\cdot)\mapsto E\{f(X_2)|D_1=\cdot,X_1=\cdot\}.
    $
    \item The conditional expectation operator is a particularly smooth element of $\mathcal{L}_2(\mathcal{H}_{\mathcal{X}},\mathcal{H}_{\mathcal{D}}\otimes \mathcal{H}_{\mathcal{X}})$. Formally, define the covariance operator $T_4=E[\{\phi(D_1)\otimes \phi(X_1)\} \otimes \{\phi(D_1)\otimes \phi(X_1)\}]$ for $\mathcal{L}_2(\mathcal{H}_{\mathcal{X}},\mathcal{H}_{\mathcal{D}}\otimes \mathcal{H}_{\mathcal{X}})$.
    We assume there exists $ G_4\in \mathcal{L}_2(\mathcal{H}_{\mathcal{X}},\mathcal{H}_{\mathcal{D}}\otimes \mathcal{H}_{\mathcal{X}})$ such that $E_4=(T_4)^{\frac{c_4-1}{2}}\circ G_4$, $c_4\in(1,2]$, and $\|G_4\|^2_{\mathcal{L}_2(\mathcal{H}_{\mathcal{X}},\mathcal{H}_{\mathcal{D}}\otimes \mathcal{H}_{\mathcal{X}})}\leq\zeta_4$. 
    \end{enumerate}
\end{assumption}

\begin{assumption}[Smoothness of the kernel embedding $\nu_{x_2}(d_1,x_1)$]\label{assumption:smooth_SDS}
Assume the following.
\begin{enumerate}
\item The conditional expectation operator $E_5$ is well specified as a Hilbert--Schmidt operator between RKHSs, i.e. $E_5\in \mathcal{L}_2(\mathcal{H}_{\mathcal{X}},\mathcal{H}_{\mathcal{D}}\otimes \mathcal{H}_{\mathcal{X}})$, where
    $
    E_5:\mathcal{H}_{\mathcal{X}} \rightarrow \mathcal{H}_{\mathcal{D}}\otimes \mathcal{H}_{\mathcal{X}},\; f(\cdot)\mapsto E_{\tilde{P}}\{f(X_2)|D_1=\cdot,X_1=\cdot\}.
    $
    \item The conditional expectation operator is a particularly smooth element of $\mathcal{L}_2(\mathcal{H}_{\mathcal{X}},\mathcal{H}_{\mathcal{D}}\otimes \mathcal{H}_{\mathcal{X}})$. Formally, define the covariance operator $T_5=E_{\tilde{P}}[\{\phi(D_1)\otimes \phi(X_1)\} \otimes \{\phi(D_1)\otimes \phi(X_1)\}]$ for $\mathcal{L}_2(\mathcal{H}_{\mathcal{X}},\mathcal{H}_{\mathcal{D}}\otimes \mathcal{H}_{\mathcal{X}})$.
    We assume there exists $ G_5\in \mathcal{L}_2(\mathcal{H}_{\mathcal{X}},\mathcal{H}_{\mathcal{D}}\otimes \mathcal{H}_{\mathcal{X}})$ such that $E_5=(T_5)^{\frac{c_5-1}{2}}\circ G_5$, $c_5\in(1,2]$, and $\|G_5\|^2_{\mathcal{L}_2(\mathcal{H}_{\mathcal{X}},\mathcal{H}_{\mathcal{D}}\otimes \mathcal{H}_{\mathcal{X}})}\leq\zeta_5$.
    \end{enumerate}
\end{assumption}

\begin{assumption}[Smoothness of the regression operator $E_8$]\label{assumption:smooth_ME_op}
Assume the following.
\begin{enumerate}
\item The conditional expectation operator $E_8$ is well specified as a Hilbert--Schmidt operator between RKHSs, i.e. $E_8\in \mathcal{L}_2(\mathcal{H}_{\mathcal{Y}},\mathcal{H}_{\mathcal{D}}\otimes\mathcal{H}_{\mathcal{M}} \otimes \mathcal{H}_{\mathcal{X}})$, where
    $
    E_8:\mathcal{H}_{\mathcal{Y}} \rightarrow \mathcal{H}_{\mathcal{D}}\otimes \mathcal{H}_{\mathcal{M}}\otimes \mathcal{H}_{\mathcal{X}},\; f(\cdot)\mapsto E\{f(Y)|D=\cdot,M=\cdot,X=\cdot\}.
    $
    \item The conditional expectation operator is a particularly smooth element of $\mathcal{L}_2(\mathcal{H}_{\mathcal{Y}},\mathcal{H}_{\mathcal{D}}\otimes \mathcal{H}_{\mathcal{M}}\otimes \mathcal{H}_{\mathcal{X}})$. Formally, define the covariance operator $T_8=E[\{\phi(D)\otimes \phi(M)\otimes \phi(X)\} \otimes  \{\phi(D)\otimes \phi(M)\otimes \phi(X)\}]$ for $\mathcal{L}_2(\mathcal{H}_{\mathcal{Y}},\mathcal{H}_{\mathcal{D}}\otimes \mathcal{H}_{\mathcal{M}}\otimes \mathcal{H}_{\mathcal{X}})$.
    We assume there exists $ G_8\in \mathcal{L}_2(\mathcal{H}_{\mathcal{Y}},\mathcal{H}_{\mathcal{D}}\otimes \mathcal{H}_{\mathcal{M}}\otimes \mathcal{H}_{\mathcal{X}})$ such that $E_8=(T_8)^{\frac{c_8-1}{2}}\circ G_8$, $c_8\in(1,2]$, and $\|G_8\|^2_{\mathcal{L}_2(\mathcal{H}_{\mathcal{Y}},\mathcal{H}_{\mathcal{D}}\otimes \mathcal{H}_{\mathcal{M}}\otimes \mathcal{H}_{\mathcal{X}})}\leq\zeta_8$. 
    \end{enumerate}
\end{assumption}

\begin{assumption}[Smoothness of the regression operator $E_9$]\label{assumption:smooth_SDS_op}
Assume the following.
\begin{enumerate}
\item The conditional expectation operator $E_9$ is well specified as a Hilbert--Schmidt operator between RKHSs, i.e. $E_9\in \mathcal{L}_2(\mathcal{H}_{\mathcal{Y}},\mathcal{H}_{\mathcal{D}}\otimes\mathcal{H}_{\mathcal{D}} \otimes \mathcal{H}_{\mathcal{X}}\otimes \mathcal{H}_{\mathcal{X}})$, where
    $
    E_9:\mathcal{H}_{\mathcal{Y}} \rightarrow \mathcal{H}_{\mathcal{D}}\otimes \mathcal{H}_{\mathcal{D}}\otimes \mathcal{H}_{\mathcal{X}}\otimes \mathcal{H}_{\mathcal{X}},\; f(\cdot)\mapsto E\{f(Y)|D_1=\cdot,D_2=\cdot,X_1=\cdot,X_2=\cdot\}.
    $
    \item The conditional expectation operator is a particularly smooth element of $\mathcal{L}_2(\mathcal{H}_{\mathcal{Y}},\mathcal{H}_{\mathcal{D}}\otimes \mathcal{H}_{\mathcal{D}}\otimes \mathcal{H}_{\mathcal{X}}\otimes \mathcal{H}_{\mathcal{X}})$. Formally, define the covariance operator $T_9=E[\{\phi(D_1)\otimes \phi(D_2)\otimes \phi(X_1)\otimes \phi(X_2)\} \otimes  \{\phi(D_1)\otimes \phi(D_2)\otimes \phi(X_1)\otimes \phi(X_2)\}]$ for $\mathcal{L}_2(\mathcal{H}_{\mathcal{Y}},\mathcal{H}_{\mathcal{D}}\otimes \mathcal{H}_{\mathcal{D}}\otimes \mathcal{H}_{\mathcal{X}}\otimes \mathcal{H}_{\mathcal{X}})$.
    We assume there exists $ G_9\in \mathcal{L}_2(\mathcal{H}_{\mathcal{Y}},\mathcal{H}_{\mathcal{D}}\otimes \mathcal{H}_{\mathcal{D}}\otimes \mathcal{H}_{\mathcal{X}}\otimes \mathcal{H}_{\mathcal{X}})$ such that $E_9=(T_9)^{\frac{c_9-1}{2}}\circ G_9$, $c_9\in(1,2]$, and $\|G_9\|^2_{\mathcal{L}_2(\mathcal{H}_{\mathcal{Y}},\mathcal{H}_{\mathcal{D}}\otimes \mathcal{H}_{\mathcal{D}}\otimes \mathcal{H}_{\mathcal{X}}\otimes \mathcal{H}_{\mathcal{X}})}\leq\zeta_9$.
    \end{enumerate}
\end{assumption}

\textbf{Interpreting smoothness for tensor products.}
Another way to interpret the smoothness assumption for a tensor product RKHS follows from a manipulation of the product kernel. For simplicity, consider $k(w,w')=k_1(w_1,w_1')k_2(w_2,w_2')$ where $k_1$ and $k_2$ are exponentiated quadratic kernels over $\mathcal{W}\subset\mathbb{R}$. Define the vector of differences $v=w-w'$. Then $$
k(w,w')
=\exp\left(-\frac{1}{2}\frac{v_1^2}{\iota_1^2}\right)
\exp\left(-\frac{1}{2}\frac{v_2^2}{\iota_2^2}\right)
=\exp\left\{-\frac{1}{2}v^{\top} \begin{pmatrix}\iota_1^{-2} & 0 \\ 0 & \iota_2^{-2}\end{pmatrix}  v  \right\}\!.
$$
In summary, the product of exponentiated quadratic kernels over scalars is an exponentiated quadratic kernel over vectors. Therefore a tensor product of exponentiated quadratic RKHSs $\mathcal{H}_{1}$ and $\mathcal{H}_{2}$ begets an exponentiated quadratic RKHS $\mathcal{H}=\mathcal{H}_{1} \otimes \mathcal{H}_{2}$, for which the smoothness and spectral decay conditions admit their usual interpretation. The same is true whenever the products of kernels beget a recognizable kernel.

\subsection{Lemmas}

\textbf{Regression.}
For expositional purposes, we summarize well known results for the kernel ridge regression estimator $\hat{\gamma}$ of $\gamma_0(w)=E(Y|W=w)$. As in Section~3, we denote the concatenation of the regressors by $W$. For mediation analysis, $W=(D,M,X)$; for time-varying analysis, $W=(D_1,D_2,X_1,X_2)$. Consider the following notation:
\begin{align*}
    \gamma_0&=\argmin_{\gamma\in\mathcal{H}}\mathcal{E}(\gamma),\quad \mathcal{E}(\gamma)=E[\{Y-\gamma(W)\}^2]; \\
    \hat{\gamma}&=\argmin_{\gamma\in\mathcal{H}}\hat{\mathcal{E}}(\gamma),\quad \hat{\mathcal{E}}(\gamma)=\frac{1}{n}\sum_{i=1}^n\{Y_i-\gamma(W_i)\}^2+\lambda\|\gamma\|^2_{\mathcal{H}}.
\end{align*}
We quote a result that verifies the conditions of \cite[Theorem 1.ii]{fischer2017sobolev}.
\begin{lemma}[Regression rate; Proposition S3 of  \cite{singh2020kernel}]\label{theorem:regression}
Suppose Assumptions~\ref{assumption:RKHS},~\ref{assumption:original}, and \ref{assumption:smooth_gamma} hold. Set $\lambda=n^{-1/(c+1/b)}$. Then with probability $1-\delta$, for $n$ sufficiently large,
$$
\|\hat{\gamma}-\gamma_0\|_{\mathcal{H}}\leq r_{\gamma}(n,\delta,b,c)=C\ln(4/\delta) \cdot n^{-\frac{1}{2}\frac{c-1}{c+1/b}},
$$
where $C$ is a constant independent of $n$ and $\delta$.
\end{lemma}

\textbf{Unconditional kernel embedding.} 
For expositional purposes, we summarize well known results for the unconditional kernel embedding estimator $\hat{\mu}_w$ of $\mu_w=E\{\phi(W)\}$. We let $W$ be a generic random variable which we instantiate differently for different causal parameters. We quote a result that appeals to Bennett inequality. 

\begin{lemma}[Kernel embedding rate; Proposition S4 of \cite{singh2020kernel}]\label{theorem:mean}
Suppose Assumptions~\ref{assumption:RKHS} and~\ref{assumption:original} hold. Then with probability $1-\delta$, 
$$
\|\hat{\mu}_w-\mu_w\|_{\mathcal{H}_{\mathcal{W}}}\leq r_{\mu}(n,\delta)=\frac{4\kappa_w \ln(2/\delta)}{\sqrt{n}}.
$$
\end{lemma}
\cite[Theorem 15]{altun2006unifying} originally prove this rate by McDiarmid inequality. See \cite[Theorem 2]{smola2007hilbert} for an argument via Rademacher complexity. See \cite[Proposition A.1]{tolstikhin2017minimax} for an improved constant and the proof that the rate is minimax optimal.

\begin{remark}[Kernel embedding rate]\label{remark:2}
In various applications, $\kappa_w$ varies.
\begin{enumerate}
    \item Mediation analysis: with probability $1-\delta$, for all $d\in\mathcal{D}$
\begin{align*}
    &\left\|\frac{1}{n}\sum_{i=1}^n\{\mu_m(d,X_i)\otimes \phi(X_i)\}-\int \{\mu_m(d,x)\otimes \phi(x)\} \mathrm{d}P(x)\right\|_{\mathcal{H}_{\mathcal{M}}\otimes \mathcal{H}_{\mathcal{X}}}\\
    &\leq r^{ME}_{\mu}(n,\delta)=\frac{4\kappa_m\kappa_x \ln(2/\delta)}{\sqrt{n}}.
\end{align*}
    \item Time-varying analysis:
    \begin{enumerate}
        \item for $\theta_0^{GF}$, with probability $1-\delta$, for all $d_1\in\mathcal{D}$
\begin{align*}
&\left\|\frac{1}{n}\sum_{i=1}^n\{\phi(X_{1i}) \otimes \mu_{x_2}(d_1,X_{1i})\}-\int \{\phi(x_1)\otimes \mu_{x_2}(d_1,x_1)\} \mathrm{d}P(x_1)\right\|_{\mathcal{H}_{\mathcal{X}}\otimes \mathcal{H}_{\mathcal{X}}} \\
&\leq r^{GF}_{\mu}(n,\delta)=\frac{4\kappa_x^2 \ln(2/\delta)}{\sqrt{n}};\end{align*}
        \item for $\theta_0^{DS}$, with probability $1-\delta$, for all $d_1 \in\mathcal{D}$
        \begin{align*}
&\left\|\frac{1}{n}\sum_{i=1}^n\{\phi(\tilde{X}_{1i}) \otimes \mu_{x_2}(d_1,\tilde{X}_{1i})\}-\int \{\phi(x_1)\otimes \mu_{x_2}(d_1,x_1)\} \mathrm{d}\tilde{P}(x_1)\right\|_{\mathcal{H}_{\mathcal{X}}\otimes \mathcal{H}_{\mathcal{X}}} \\
&\leq r^{DS}_{\nu}(\tilde{n},\delta)=\frac{4\kappa^2_x \ln(2/\delta)}{\sqrt{\tilde{n}}}. \end{align*}
    \end{enumerate}
\end{enumerate}
\end{remark}

\textbf{Conditional expectation operator and conditional kernel embedding.} 
As in Sections~4 and~5 as well as Supplement~\ref{sec:dist}, we consider the abstract operator $E_{\ell}\in \mathcal{L}_2(\mathcal{H}_{\mathcal{A}_{\ell}},\mathcal{H}_{\mathcal{B}_{\ell}})$, where $\mathcal{A}_{\ell}$ and $\mathcal{B}_{\ell}$ are spaces that can be instantiated for different causal parameters.
Consider the definitions
\begin{align*}
    E_{\ell}&=\argmin_{E\in\mathcal{L}_2(\mathcal{H}_{\mathcal{A}_{\ell}},\mathcal{H}_{\mathcal{B}_{\ell}})}\mathcal{E}(E),\quad \mathcal{E}(E)=E[\{\phi(A_{\ell})-E^*\phi(B_{\ell})\}^2]; \\
    \hat{E}_{\ell}&=\argmin_{E\in\mathcal{L}_2(\mathcal{H}_{\mathcal{A}_{\ell}},\mathcal{H}_{\mathcal{B}_{\ell}})}\hat{\mathcal{E}}(E),\quad \hat{\mathcal{E}}(E)=\frac{1}{n}\sum_{i=1}^n[\phi(A_{\ell i})-E^*\phi(B_{\ell i})]^2+\lambda_{\ell}\|E\|^2_{\mathcal{L}_2(\mathcal{H}_{\mathcal{A}_{\ell}},\mathcal{H}_{\mathcal{B}_{\ell}})}.
\end{align*}

\begin{lemma}[Conditional kernel embedding rate; Proposition S5 of \cite{singh2020kernel}]\label{theorem:conditional}
Suppose Assumptions~\ref{assumption:RKHS},~\ref{assumption:original}, and~\ref{assumption:smooth_op} hold. Set $\lambda_{\ell}=n^{-1/(c_{\ell}+1/b_{\ell})}$. Then with probability $1-\delta$, for $n$ sufficiently large,
$$
\|\hat{E}_{\ell}-E_{\ell}\|_{\mathcal{L}_2}\leq r_E(\delta,n,b_{\ell},c_{\ell})=C\ln(4/\delta)\cdot n^{-\frac{1}{2}\frac{c_{\ell}-1}{c_{\ell}+1/b_{\ell}}}.
$$
Moreover, for all $b\in\mathcal{B}_{\ell}$
$$
  \|\hat{\mu}_a(b)-\mu_a(b)\|_{\mathcal{H}_{\mathcal{A}_{\ell}}}\leq r_{\mu}(\delta,n,b_{\ell},c_{\ell})=\kappa_{b}\cdot
  r_E(\delta,n,b_{\ell},c_{\ell}).
    $$
\end{lemma}

\begin{remark}\label{remark:3}
In various applications, $\kappa_a$ and $\kappa_b$ vary.
\begin{enumerate}
    \item For mediation analysis: $\kappa_a=\kappa_m$, $\kappa_b=\kappa_d\kappa_x$.
    \item For time-varying analysis: $\kappa_a=\kappa_x$, $\kappa_b=\kappa_d\kappa_x$.
    \item Within counterfactual distributions,
    \begin{enumerate}
        \item for mediation analysis: $\kappa_a=\kappa_y$, $\kappa_b=\kappa_d\kappa_m\kappa_x$;
        \item for time-varying analysis: $\kappa_a=\kappa_y$, $\kappa_b=\kappa_d^2\kappa_x^2$.
    \end{enumerate}
\end{enumerate}
\end{remark}

\subsection{Main results}

Appealing to Lemmas~\ref{theorem:regression},~\ref{theorem:mean}, and~\ref{theorem:conditional} we now prove consistency for (i) mediated responses, (ii) time-varying dose responses, and (iii) counterfactual distributions.

\textbf{Mediated responses.}
To lighten notation, define
$$
\hat{\mu}_{m,x}(d)-\mu_{m,x}(d)=\frac{1}{n}\sum_{i=1}^n \{\hat{\mu}_m(d,X_i) \otimes \phi(X_i)\}-\int \{\mu_m(d,x) \otimes \phi(x)\}\mathrm{d}P(x).
$$

\begin{proposition}[Sequential kernel embedding rate]\label{prop:delta_m}
Suppose Assumptions~\ref{assumption:RKHS}, \ref{assumption:original}, and \ref{assumption:smooth_ME} hold. Then with probability $1-2\delta$, 
$$\left\|\Delta_m\right\|_{\mathcal{H}_{\mathcal{M}}\otimes\mathcal{H}_{\mathcal{X}}}
    \leq \kappa_x \cdot r^{ME}_{\mu}(n,\delta,b_1,c_1)+r^{ME}_{\mu}(n,\delta).
$$
\end{proposition}

\begin{proof}
  By the triangle inequality, 
  \begin{align*}
  \left\|\Delta_m\right\|_{\mathcal{H}_{\mathcal{M}}\otimes\mathcal{H}_{\mathcal{X}}}
  &\leq \left\|\frac{1}{n}\sum_{i=1}^n \{\hat{\mu}_m(d,X_i) \otimes \phi(X_i)\}-\{\mu_m(d,X_i)\otimes \phi(X_i)\}\right\|_{\mathcal{H}_{\mathcal{M}}\otimes\mathcal{H}_{\mathcal{X}}}\\
  &\quad + \left\|\frac{1}{n}\sum_{i=1}^n \{\mu_m(d,X_i)\otimes \phi(X_i)\}-\int \{\mu_m(d,x) \otimes \phi(x)\}\mathrm{d}P(x)\right\|_{\mathcal{H}_{\mathcal{M}}\otimes\mathcal{H}_{\mathcal{X}}}\!.
  \end{align*}
Focusing on the former term, by Lemma~\ref{theorem:conditional}
  \begin{align*}
      &\left\|\frac{1}{n}\sum_{i=1}^n \{\hat{\mu}_m(d,X_i) \otimes \phi(X_i)\}-\{\mu_m(d,X_i)\otimes \phi(X_i)\}\right\|_{\mathcal{H}_{\mathcal{M}}\otimes\mathcal{H}_{\mathcal{X}}} \\
      &=\left\|\frac{1}{n}\sum_{i=1}^n \{\hat{\mu}_m(d,X_i)-\mu_m(d,X_i)\} \otimes \phi(X_i)\right\|_{\mathcal{H}_{\mathcal{M}}\otimes\mathcal{H}_{\mathcal{X}}} \\
      &\leq \kappa_x \cdot \sup_{x\in\mathcal{X}}\left\| \hat{\mu}_m(d,x)-\mu_m(d,x)\right\|_{\mathcal{H}_{\mathcal{M}}} \\
      &\leq \kappa_x \cdot r^{ME}_{\mu}(n,\delta,b_1,c_1).
  \end{align*}
  Focusing on the latter term, by Lemma~\ref{theorem:mean}
  \begin{align*}
      \left\|\frac{1}{n}\sum_{i=1}^n \{\mu_m(d,X_i)\otimes \phi(X_i)\}-\int \{\mu_m(d,x) \otimes \phi(x)\}\mathrm{d}P(x)\right\|_{\mathcal{H}_{\mathcal{M}}\otimes\mathcal{H}_{\mathcal{X}}}\leq r^{ME}_{\mu}(n,\delta).
  \end{align*}
\end{proof}

\begin{proof}[Proof of Theorem~\ref{theorem:consistency_mediation}]
Observe that
   \begin{align*}
        &\hat{\theta}^{ME}(d,d')-\theta_0^{ME}(d,d')\\
        &=\langle \hat{\gamma}, \phi(d')\otimes \frac{1}{n}\sum_{i=1}^n \{\hat{\mu}_m(d,X_i) \otimes \phi(X_i)\} \rangle_{\mathcal{H}} - \langle \gamma_0, \phi(d')\otimes \int \{\mu_m(d,x) \otimes \phi(x)\}\mathrm{d}P(x) \rangle_{\mathcal{H}} \\
        &=\langle \hat{\gamma}, \phi(d')\otimes \Delta_m \rangle_{\mathcal{H}} 
        +\langle(\hat{\gamma}-\gamma_0), \phi(d')\otimes \int \{\mu_m(d,x) \otimes \phi(x)\}\mathrm{d}P(x) \rangle_{\mathcal{H}} \\
        &=\langle (\hat{\gamma}-\gamma_0), \phi(d')\otimes \Delta_m \rangle_{\mathcal{H}} \\
        &\quad + \langle \gamma_0, \phi(d')\otimes \Delta_m \rangle_{\mathcal{H}}+\langle(\hat{\gamma}-\gamma_0), \phi(d')\otimes \int \{\mu_m(d,x) \otimes \phi(x)\}\mathrm{d}P(x) \rangle_{\mathcal{H}}.
    \end{align*}
    Therefore by Lemmas~\ref{theorem:regression},~\ref{theorem:mean}, and~\ref{theorem:conditional} as well as~Proposition~\ref{prop:delta_m}, with probability $1-3\delta$
    \begin{align*}
        &|\hat{\theta}^{ME}(d,d')-\theta^{ME}_0(d,d')|\\
        &\leq \|\hat{\gamma}-\gamma_0\|_{\mathcal{H}}\|\phi(d')\|_{\mathcal{H}_{\mathcal{D}}} \left\|\Delta_m\right\|_{\mathcal{H}_{\mathcal{M}}\otimes\mathcal{H}_{\mathcal{X}}}
       \\
       &\quad+
       \|\gamma_0\|_{\mathcal{H}}\|\phi(d')\|_{\mathcal{H}_{\mathcal{D}}}\left\|\Delta_m\right\|_{\mathcal{H}_{\mathcal{M}}\otimes\mathcal{H}_{\mathcal{X}}}
       \\
       &\quad+
       \|\hat{\gamma}-\gamma_0\|_{\mathcal{H}}\|\phi(d')\|_{\mathcal{H}_{\mathcal{D}}} \left\|\int \{\mu_m(d,x) \otimes \phi(x)\}\mathrm{d}P(x)\right\|_{\mathcal{H}_{\mathcal{M}}\otimes\mathcal{H}_{\mathcal{X}}}
      \\
      &\leq \kappa_d \cdot r_{\gamma}(n,\delta,b,c) \cdot \{\kappa_x \cdot r^{ME}_{\mu}(n,\delta,b_1,c_1)+r^{ME}_{\mu}(n,\delta)\}\\
      &\quad +\kappa_d\cdot\|\gamma_0\|_{\mathcal{H}} \cdot \{\kappa_x \cdot r^{ME}_{\mu}(n,\delta,b_1,c_1)+r^{ME}_{\mu}(n,\delta)\}
      \\
       &\quad+\kappa_m\kappa_d\kappa_x \cdot r_{\gamma}(n,\delta,b,c)
      \\
      &=O\left(n^{-\frac{1}{2}\frac{c-1}{c+1/b}}+n^{-\frac{1}{2}\frac{c_1-1}{c_1+1/b_1}}\right)\!.
    \end{align*}
Likewise for the incremental responses.
\end{proof}

\textbf{Time-varying responses.}
To lighten notation, define
\begin{align*}
    \Delta_p&=\hat{\mu}_{x_1,x_2}(d_1)-\mu_{x_1,x_2}(d_1)=\frac{1}{n}\sum_{i=1}^n \{\phi(X_{1i})\otimes \hat{\mu}_{x_2}(d_1,X_{1i}) \}-\int \{\phi(x_1)\otimes \mu_{x_2}(d_1,x_1) \}\mathrm{d}P(x_1), \\
    \Delta_q&=\hat{\nu}_{x_1,x_2}(d_1)-\nu_{x_1,x_2}(d_1)=\frac{1}{n}\sum_{i=1}^n \{\phi(X_{1i})\otimes \hat{\nu}_{x_2}(d_1,X_{1i}) \}-\int \{\phi(x_1)\otimes \nu_{x_2}(d_1,x_1) \}\mathrm{d}\tilde{P}(x_1).
\end{align*}

\begin{proposition}[Sequential kernel embedding rates]\label{prop:delta_p}
Suppose Assumptions~\ref{assumption:RKHS} and~\ref{assumption:original} hold.
\begin{enumerate}
    \item If in addition Assumption~\ref{assumption:smooth_SATE} holds then with probability $1-2\delta$
    \begin{align*}
        &\left\|\Delta_p\right\|_{\mathcal{H}_{\mathcal{X}}\otimes\mathcal{H}_{\mathcal{X}}}\leq \kappa_x \cdot r^{GF}_{\mu}(n,\delta,b_4,c_4)+r^{GF}_{\mu}(n,\delta).
\end{align*}
    \item If in addition Assumption~\ref{assumption:smooth_SDS} holds then with probability $1-2\delta$
    \begin{align*}
        &\left\|\Delta_q\right\|_{\mathcal{H}_{\mathcal{X}}\otimes\mathcal{H}_{\mathcal{X}}}\leq \kappa_x \cdot r^{DS}_{\nu}(\tilde{n},\delta,c_5)+r^{DS}_{\nu}(\tilde{n},\delta).
\end{align*}
\end{enumerate}
\end{proposition}

\begin{proof}
We prove the result for $\theta_0^{GF}$. The argument for $\theta_0^{DS}$ is identical.
  By the triangle inequality, 
\begin{align*}
    \left\|\Delta_p\right\|_{\mathcal{H}_{\mathcal{X}}\otimes\mathcal{H}_{\mathcal{X}}}
    &\leq \left\|\frac{1}{n}\sum_{i=1}^n \{\phi(X_{1i})\otimes \hat{\mu}_{x_2}(d_1,X_{1i}) \}-\{\phi(X_{1i})\otimes \mu_{x_2}(d_1,X_{1i}) \}\right\|_{\mathcal{H}_{\mathcal{X}}\otimes\mathcal{H}_{\mathcal{X}}} \\
    &\quad +  \left\|\frac{1}{n}\sum_{i=1}^n\{\phi(X_{1i}) \otimes \mu_{x_2}(d_1,X_{1i})\}-\int \{\phi(x_1)\otimes \mu_{x_2}(d_1,x_1)\} \mathrm{d}P(x_1)\right\|_{\mathcal{H}_{\mathcal{X}}\otimes \mathcal{H}_{\mathcal{X}}}\!.
\end{align*}
  Focusing on the former term, by Lemma~\ref{theorem:conditional}
  \begin{align*}
      &\left\|\frac{1}{n}\sum_{i=1}^n \{\phi(X_{1i})\otimes \hat{\mu}_{x_2}(d_1,X_{1i}) \}-\{\phi(X_{1i})\otimes \mu_{x_2}(d_1,X_{1i}) \}\right\|_{\mathcal{H}_{\mathcal{X}}\otimes\mathcal{H}_{\mathcal{X}}} \\
      &=\left\|\frac{1}{n}\sum_{i=1}^n  \phi(X_{1i}) \otimes \{\hat{\mu}_{x_2}(d_1,X_{1i})-\mu_{x_2}(d_1,X_{1i})\} \right\|_{\mathcal{H}_{\mathcal{X}}\otimes\mathcal{H}_{\mathcal{X}}} \\
      &\leq \kappa_x \cdot \sup_{x_1\in\mathcal{X}}\left\| \hat{\mu}_{x_2}(d_1,x_1)-\mu_{x_2}(d_1,x_1)\right\|_{\mathcal{H}_{\mathcal{X}}} \\
      &\leq \kappa_x \cdot r^{GF}_{\mu}(n,\delta,b_4,c_4).
  \end{align*}
  Focusing on the latter term, by Lemma~\ref{theorem:mean}
  \begin{align*}
      \left\|\frac{1}{n}\sum_{i=1}^n\{\phi(X_{1i}) \otimes \mu_{x_2}(d_1,X_{1i})\}-\int \{\phi(x_1)\otimes \mu_{x_2}(d_1,x_1)\} \mathrm{d}P(x_1)\right\|_{\mathcal{H}_{\mathcal{X}}\otimes \mathcal{H}_{\mathcal{X}}}\leq r^{GF}_{\mu}(n,\delta).
  \end{align*}
\end{proof}

\begin{proof}[Proof of Theorem~\ref{theorem:consistency_planning}]
We consider each time-varying response parameter.
\begin{enumerate}
    \item For $\theta_0^{GF}$, observe that
    \begin{align*}
        &\hat{\theta}^{GF}(d_1,d_2)-\theta_0^{GF}(d_1,d_2)\\
        &=\langle \hat{\gamma}, \phi(d_1)\otimes\phi(d_2)\otimes \frac{1}{n}\sum_{i=1}^n\{\phi(X_{1i})\otimes \hat{\mu}_{x_2}(d_1,X_{1i})\}  \rangle_{\mathcal{H}}  \\
        &\quad - \langle \gamma_0,  \phi(d_1)\otimes\phi(d_2) \otimes \int \phi(x_1)\otimes \mu_{x_2}(d_1,x_1) \mathrm{d}P(x_1) \rangle_{\mathcal{H}} \\
        &=\langle \hat{\gamma}, \phi(d_1)\otimes\phi(d_2)\otimes 
        \Delta_p \rangle_{\mathcal{H}} \\
        &\quad +\langle (\hat{\gamma}-\gamma_0),  \phi(d_1)\otimes\phi(d_2) \otimes \int \phi(x_1)\otimes \mu_{x_2}(d_1,x_1) \mathrm{d}P(x_1) \rangle_{\mathcal{H}}  \\
        &=\langle (\hat{\gamma}-\gamma_0), \phi(d_1)\otimes\phi(d_2)\otimes
        \Delta_p \rangle_{\mathcal{H}} 
        \\
       &\quad+\langle \gamma_0, \phi(d_1)\otimes\phi(d_2)\otimes
        \Delta_p \rangle_{\mathcal{H}} \\
        &\quad +\langle (\hat{\gamma}-\gamma_0),  \phi(d_1)\otimes\phi(d_2) \otimes \int \phi(x_1)\otimes \mu_{x_2}(d_1,x_1) \mathrm{d}P(x_1) \rangle_{\mathcal{H}}.
    \end{align*}
    Therefore by Lemmas~\ref{theorem:regression},~\ref{theorem:mean}, and~\ref{theorem:conditional} as well as~Proposition~\ref{prop:delta_p}, with probability $1-3\delta$
    \begin{align*}
        &|\hat{\theta}^{GF}(d_1,d_2)-\theta^{GF}_0(d_1,d_2)|\\
        &\leq \|\hat{\gamma}-\gamma_0\|_{\mathcal{H}}\|\phi(d_1)\|_{\mathcal{H}_{\mathcal{D}}}\|\phi(d_2)\|_{\mathcal{H}_{\mathcal{D}}} 
        \left\|\Delta_p\right\|_{\mathcal{H}_{\mathcal{X}}\otimes\mathcal{H}_{\mathcal{X}}}
       \\
       &\quad +
       \|\gamma_0\|_{\mathcal{H}}\|\phi(d_1)\|_{\mathcal{H}_{\mathcal{D}}}\|\phi(d_2)\|_{\mathcal{H}_{\mathcal{D}}}
       \left\|\Delta_p\right\|_{\mathcal{H}_{\mathcal{X}}\otimes\mathcal{H}_{\mathcal{X}}}
       \\
       &\quad+
       \|\hat{\gamma}-\gamma_0\|_{\mathcal{H}}\|\phi(d_1)\|_{\mathcal{H}_{\mathcal{D}}}\|\phi(d_2)\|_{\mathcal{H}_{\mathcal{D}}}\times \left\|\int \{\phi(x_1)\otimes \mu_{x_2}(d_1,x_1) \}\mathrm{d}P(x_1)\right\|_{\mathcal{H}_{\mathcal{X}}\otimes\mathcal{H}_{\mathcal{X}}}
      \\
      &\leq \kappa^2_d \cdot r_{\gamma}(n,\delta,b,c) \cdot \{\kappa_x \cdot r^{GF}_{\mu}(n,\delta,b_4,c_4)+r^{GF}_{\mu}(n,\delta)\}\\
      &\quad +\kappa^2_d\cdot\|\gamma_0\|_{\mathcal{H}} \cdot \{\kappa_x \cdot r^{GF}_{\mu}(n,\delta,b_4,c_4)+r^{GF}_{\mu}(n,\delta)\}\\
       &\quad+\kappa^2_d\kappa^2_x \cdot r_{\gamma}(n,\delta,b,c)
      \\
      &=O\left(n^{-\frac{1}{2}\frac{c-1}{c+1/b}}+n^{-\frac{1}{2}\frac{c_4-1}{c_4+1/b_4}}\right).
    \end{align*}
    
    \item For $\theta_0^{DS}$, by the same argument
  \begin{align*}
        &|\hat{\theta}^{DS}(d_1,d_2)-\theta^{DS}_0(d_1,d_2)|\\
      &\leq \kappa^2_d \cdot r_{\gamma}(n,\delta,b,c) \cdot \{\kappa_x \cdot r^{DS}_{\nu}(\tilde{n},\delta,c_5)+r^{DS}_{\nu}(\tilde{n},\delta)\}\\
      &\quad +\kappa^2_d\cdot\|\gamma_0\|_{\mathcal{H}} \cdot \{\kappa_x \cdot r^{DS}_{\nu}(\tilde{n},\delta,c_5)+r^{DS}_{\nu}(\tilde{n},\delta)\}\\
      &\quad +\kappa^2_d\kappa^2_x \cdot r_{\gamma}(n,\delta,b,c)
      \\
      &=O\left(n^{-\frac{1}{2}\frac{c-1}{c+1/b}}+\tilde{n}^{-\frac{1}{2}\frac{c_5-1}{c_5+1/b_5}}\right)\!.
    \end{align*}
\end{enumerate}
   Likewise for the incremental responses.
\end{proof}

\textbf{Counterfactual distributions.}
\begin{proof}[Proof of Theorem~\ref{theorem:consistency_dist}]
   The argument is analogous to Theorems~\ref{theorem:consistency_mediation} and~\ref{theorem:consistency_planning}, replacing $\|\gamma_0\|_{\mathcal{H}}$ with $\|E_8\|_{\mathcal{L}_2}$ or $\|E_9\|_{\mathcal{L}_2}$  and replacing $r_{\gamma}(n,\delta,b,c)$ with $r_E(n,\delta,b_8,c_8)$ or $r_E(n,\delta,b_9,c_9)$.
\end{proof}
\section{Semiparametric inference proofs}\label{sec:inference_proof}

In this supplement, we (i) present technical lemmas for regression; (ii) appeal to these lemmas to prove $n^{-1/2}$ consistency, finite sample Gaussian approximation, and semiparametric efficiency of mediated and time-varying treatment effects.

\subsection{Lemmas}

Recall the various nonparametric objects required for inference. For mediated effects,
\begin{align*}
    \gamma_0(d,m,x)&=E(Y|D=d,M=m,X=x),\\
    \pi_0(d;x)&=\text{\normalfont pr}(D=d|X=x),\\
    \rho_0(d;m,x)&=\text{\normalfont pr}(D=d|M=m,X=x),\\
   \omega_0(d,d';x) &=\int \gamma_0(d',m,x)\mathrm{d}P(m|d,x).
\end{align*}
For time-varying treatment effects,
\begin{align*}
\gamma_0(d_1,d_2,x_1,x_2)&=E(Y|D_1=d_1,D_2=d_2,X_1=x_1,X_2=x_2),\\
\pi_0(d_1;x_1)&=\text{\normalfont pr}(D_1=d_1|X_1=x_1),\\
\rho_0(d_2;d_1,x_1,x_2)&=\text{\normalfont pr}(D_2=d_2|D_1=d_1,X_1=x_1,X_2=x_2),\\
\omega_0(d_1,d_2;x_1)&=\int \gamma_0(d_1,d_2,x_1,x_2) \mathrm{d}P(x_2|d_1,x_1).
\end{align*}
We begin by summarizing rates for these various quantities.

\textbf{Response curve rate.}
We provide a uniform rate for $\omega_0$, using nonparametric techniques developed in Supplement~\ref{sec:consistency_proof}.

\begin{proposition}[Uniform $\omega$ rate]\label{prop:unif_omega}
Suppose Assumptions~\ref{assumption:RKHS},~\ref{assumption:original}, and~\ref{assumption:smooth_gamma} hold.
\begin{enumerate}
    \item If in addition Assumption~\ref{assumption:smooth_ME} holds then for mediated efffects, with probability $1-2\delta$,
\begin{align*}
\|\hat{\omega}-\omega_0\|_{\infty}
&\leq r^{ME}_{\omega}(n,\delta,b,c,b_1,c_1)\\
&=\kappa_d\kappa_x \cdot r_{\gamma}(n,\delta,b,c) \cdot r^{ME}_{\mu}(n,\delta,b_1,c_1)+\kappa_d\kappa_x \cdot \|\gamma_0\|_{\mathcal{H}}\cdot  r^{ME}_{\mu}(n,\delta,b_1,c_1)\\
&\quad +\kappa_d\kappa_m\kappa_x\cdot  r_{\gamma}(n,\delta,b,c)\\
&=O\left(n^{-\frac{1}{2}\frac{c-1}{c+1/b}}+n^{-\frac{1}{2}\frac{c_1-1}{c_1+1/b_1}}\right)\!.
\end{align*}
    \item If in addition Assumption~\ref{assumption:smooth_SATE} holds then for time-varying treatment effects, with probability $1-2\delta$
\begin{align*}
\|\hat{\omega}-\omega_0\|_{\infty}
&\leq r^{GF}_{\omega}(n,\delta,b,c,b_4,c_4)\\
&=\kappa^2_d\kappa_x \cdot r_{\gamma}(n,\delta,b,c) \cdot r^{GF}_{\mu}(n,\delta,b_4,c_4)+\kappa^2_d\kappa_x \cdot \|\gamma_0\|_{\mathcal{H}}\cdot  r^{GF}_{\mu}(n,\delta,b_4,c_4)\\
&\quad +\kappa^2_d\kappa^2_x\cdot  r_{\gamma}(n,\delta,b,c)\\
&=O\left(n^{-\frac{1}{2}\frac{c-1}{c+1/b}}+n^{-\frac{1}{2}\frac{c_4-1}{c_4+1/b_4}}\right)\!.
\end{align*}
\end{enumerate}
\end{proposition}

\begin{proof}
We prove each result.
\begin{enumerate}
    \item For the mediated effect, fix $(d,d',x)$. Then
    \begin{align*}
        &\hat{\omega}(d,d';x)-\omega_0(d,d';x)\\
        &=\langle\hat{\gamma},\phi(d')\otimes \hat{\mu}_m(d,x)\otimes \phi(x)\rangle_{\mathcal{H}}-\langle \gamma_0,\phi(d')\otimes \mu_m(d,x)\otimes \phi(x) \rangle_{\mathcal{H}}  \\
        &=\langle\hat{\gamma},\phi(d')\otimes \{\hat{\mu}_m(d,x)-\mu_m(d,x)\}\otimes \phi(x)\rangle_{\mathcal{H}}+\langle (\hat{\gamma}-\gamma_0),\phi(d')\otimes \mu_m(d,x)\otimes \phi(x) \rangle_{\mathcal{H}} \\
        &=\langle(\hat{\gamma}-\gamma_0),\phi(d')\otimes \{\hat{\mu}_m(d,x)-\mu_m(d,x)\}\otimes \phi(x)\rangle_{\mathcal{H}}\\
        &\quad +\langle \gamma_0,\phi(d')\otimes \{\hat{\mu}_m(d,x)-\mu_m(d,x)\}\otimes \phi(x)\rangle_{\mathcal{H}}\\
        &\quad +\langle (\hat{\gamma}-\gamma_0),\phi(d')\otimes \mu_m(d,x)\otimes \phi(x) \rangle_{\mathcal{H}}.
    \end{align*}
    Therefore by Lemmas~\ref{theorem:regression} and~\ref{theorem:conditional}, with probability $1-2\delta$
    \begin{align*}
        &|\hat{\omega}(d,d';x)-\omega_0(d,d';x)| \\
        &\leq \|\hat{\gamma}-\gamma_0\|_{\mathcal{H}}\|\phi(d')\|_{\mathcal{H}_{\mathcal{D}}} \|\hat{\mu}_m(d,x)-\mu_m(d,x)\|_{\mathcal{H}_{\mathcal{M}}} \|\phi(x)\|_{\mathcal{H}_{\mathcal{X}}} \\
        &\quad + \|\gamma_0\|_{\mathcal{H}} \|\phi(d')\|_{\mathcal{H}_{\mathcal{D}}} \|\hat{\mu}_m(d,x)-\mu_m(d,x)\|_{\mathcal{H}_{\mathcal{M}}} \|\phi(x)\|_{\mathcal{H}_{\mathcal{X}}} \\
        &\quad + \|\hat{\gamma}-\gamma_0\|_{\mathcal{H}}\|\phi(d')\|_{\mathcal{H}_{\mathcal{D}}} \|\mu_m(d,x)\|_{\mathcal{H}_{\mathcal{M}}} \|\phi(x)\|_{\mathcal{H}_{\mathcal{X}}} \\
        &\leq \kappa_d\kappa_x \cdot r_{\gamma}(n,\delta,b,c) \cdot r^{ME}_{\mu}(n,\delta,b_1,c_1)+\kappa_d\kappa_x \cdot \|\gamma_0\|_{\mathcal{H}}\cdot  r^{ME}_{\mu}(n,\delta,b_1,c_1)\\
        &\quad +\kappa_d\kappa_m\kappa_x\cdot  r_{\gamma}(n,\delta,b,c) \\
        &=O\left(n^{-\frac{1}{2}\frac{c-1}{c+1/b}}+n^{-\frac{1}{2}\frac{c_1-1}{c_1+1/b_1}}\right)\!.
    \end{align*}
    \item For the time-varying treatment effect, fix $(d_1,d_2,x_1)$. Then
    \begin{align*}
        &\hat{\omega}(d_1,d_2;x_1)-\omega_0(d,d';x_1)\\
        &=\langle\hat{\gamma},\phi(d_1)\otimes \phi(d_2) \otimes \phi(x_1)\otimes \hat{\mu}_{x_2}(d_1,x_1)\rangle_{\mathcal{H}}\\
        &\quad -\langle \gamma_0,\phi(d_1)\otimes \phi(d_2) \otimes \phi(x_1)\otimes \mu_{x_2}(d_1,x_1) \rangle_{\mathcal{H}}  \\
        &=\langle\hat{\gamma},\phi(d_1)\otimes \phi(d_2) \otimes \phi(x_1)\otimes \{\hat{\mu}_{x_2}(d_1,x_1)-\mu_{x_2}(d_1,x_1)\}\rangle_{\mathcal{H}}
        \\
        &\quad +\langle (\hat{\gamma}-\gamma_0),\phi(d_1)\otimes \phi(d_2) \otimes \phi(x_1)\otimes \mu_{x_2}(d_1,x_1) \rangle_{\mathcal{H}} \\
        &=\langle(\hat{\gamma}-\gamma_0),\phi(d_1)\otimes \phi(d_2) \otimes \phi(x_1)\otimes \{\hat{\mu}_{x_2}(d_1,x_1)-\mu_{x_2}(d_1,x_1)\}\rangle_{\mathcal{H}}\\
        &\quad +\langle \gamma_0,\phi(d_1)\otimes \phi(d_2) \otimes \phi(x_1)\otimes \{\hat{\mu}_{x_2}(d_1,x_1)-\mu_{x_2}(d_1,x_1)\}\rangle_{\mathcal{H}}\\
        &\quad +\langle (\hat{\gamma}-\gamma_0),\phi(d_1)\otimes \phi(d_2) \otimes \phi(x_1)\otimes \mu_{x_2}(d_1,x_1) \rangle_{\mathcal{H}}.
    \end{align*}
    Therefore by Lemmas~\ref{theorem:regression} and~\ref{theorem:conditional}, with probability $1-2\delta$
    \begin{align*}
        &|\hat{\omega}(d,d';x)-\omega_0(d,d';x)| \\
        &\leq \|\hat{\gamma}-\gamma_0\|_{\mathcal{H}}\|\phi(d_1)\|_{\mathcal{H}_{\mathcal{D}}}\|\phi(d_2)\|_{\mathcal{H}_{\mathcal{D}}} \|\phi(x_1)\|_{\mathcal{H}_{\mathcal{X}}} \|\hat{\mu}_{x_2}(d_1,x_1)-\mu_{x_1}(d_1,x_1)\|_{\mathcal{H}_{\mathcal{X}}} \\
        &\quad + \|\gamma_0\|_{\mathcal{H}} \|\phi(d_1)\|_{\mathcal{H}_{\mathcal{D}}}\|\phi(d_2)\|_{\mathcal{H}_{\mathcal{D}}}\|\phi(x_1)\|_{\mathcal{H}_{\mathcal{X}}} \|\hat{\mu}_{x_2}(d_1,x_1)-\mu_{x_1}(d_1,x_1)\|_{\mathcal{H}_{\mathcal{X}}}  \\
        &\quad + \|\hat{\gamma}-\gamma_0\|_{\mathcal{H}}\|\phi(d_1)\|_{\mathcal{H}_{\mathcal{D}}}\|\phi(d_2)\|_{\mathcal{H}_{\mathcal{D}}} \|\phi(x_1)\|_{\mathcal{H}_{\mathcal{X}}} \|\mu_{x_2}(d_1,x_1)\|_{\mathcal{H}_{\mathcal{X}}} \\
        &\leq \kappa^2_d\kappa_x \cdot r_{\gamma}(n,\delta,b,c) \cdot r^{GF}_{\mu}(n,\delta,b_4,c_4)+\kappa^2_d\kappa_x \cdot \|\gamma_0\|_{\mathcal{H}}\cdot  r^{GF}_{\mu}(n,\delta,b_4,c_4)\\
        &\quad +\kappa^2_d\kappa^2_x\cdot  r_{\gamma}(n,\delta,b,c) \\
        &=O\left(n^{-\frac{1}{2}\frac{c-1}{c+1/b}}+n^{-\frac{1}{2}\frac{c_4-1}{c_4+1/b_4}}\right)\!.
    \end{align*}
\end{enumerate}
\end{proof}

\textbf{Mean square rate.}
Observe that $(\gamma_0,\pi_0,\rho_0)$ can be estimated by nonparametric regressions. We write an abstract result for kernel ridge regression then specialize it for these various nonparametric regressions. For the abstract result, write
\begin{align*}
    \gamma_0&=\argmin_{\gamma\in\mathcal{H}}\mathcal{E}(\gamma),\quad \mathcal{E}(\gamma)=E[\{Y-\gamma(W)\}^2], \\
    \hat{\gamma}&=\argmin_{\gamma\in\mathcal{H}}\hat{\mathcal{E}}(\gamma),\quad \hat{\mathcal{E}}(\gamma)=\frac{1}{n}\sum_{i=1}^n\{Y_i-\gamma(W_i)\}^2+\lambda\|\gamma\|^2_{\mathcal{H}}.
\end{align*}
Define the mean square learning rate $\mathcal{R}(\hat{\gamma}_{\ell})$ of $\hat{\gamma}_{\ell}$ trained on observations indexed by $I^c_{\ell}$ as
$$
    \mathcal{R}(\hat{\gamma}_{\ell})=E[\{\hat{\gamma}_{\ell}(W)-\gamma_0(W)\}^2\mid I^c_{\ell}].
$$

\begin{lemma}[Regression mean square rate]\label{theorem:mse}
Suppose Assumptions~\ref{assumption:RKHS},~\ref{assumption:original}, and \ref{assumption:smooth_gamma} hold. Set $\lambda=n^{-1/(c+1/b)}$. Then with probability $1-\delta$, for $n$ sufficiently large,
$$
\{\mathcal{R}(\hat{\gamma}_{\ell})\}^{1/2}\leq s_{\gamma}(n,\delta,b,c)=C\ln(4/\delta) \cdot n^{-\frac{1}{2}\frac{c}{c+1/b}}.
$$
where $C$ is a constant independent of $n$ and $\delta$.
\end{lemma}

\begin{proof}
The proof is identical to the proof of \cite[Proposition S3]{singh2020kernel}, changing the Hilbert scale from one to zero. In particular, the proof verifies the conditions of \cite[Theorem 1.ii]{fischer2017sobolev}.
\end{proof}

\begin{remark}\label{remark:mse}
Note that in various applications, $(b,c)$ vary.
\begin{enumerate}
    \item For mediated effects:
    \begin{enumerate}
        \item $\{\mathcal{R}(\hat{\gamma}_{\ell})\}^{1/2}\leq s^{ME}_{\gamma}(n,b,c)$,
        \item $\{\mathcal{R}(\hat{\pi}_{\ell})\}^{1/2}\leq s^{ME}_{\pi}(n,b_2,c_2)$,
        \item $\{\mathcal{R}(\hat{\rho}_{\ell})\}^{1/2}\leq s^{ME}_{\rho}(n,b_3,c_3)$.
    \end{enumerate}
    \item For time-varying treatment effects:
        \begin{enumerate}
         \item $\{\mathcal{R}(\hat{\gamma}_{\ell})\}^{1/2}\leq s^{GF}_{\gamma}(n,b,c)$,
        \item $\{\mathcal{R}(\hat{\pi}_{\ell})\}^{1/2}\leq s^{GF}_{\pi}(n,b_6,c_6)$,
        \item $\{\mathcal{R}(\hat{\rho}_{\ell})\}^{1/2}\leq s^{GF}_{\rho}(n,b_7,c_7)$.
    \end{enumerate}
\end{enumerate}
\end{remark}

\subsection{Gaussian approximation}

We quote an abstract result for semiparametric inference in longitudinal settings. The result concerns a causal parameter $
\theta_0=E[\psi(\delta_0,\nu_0,\alpha_0,\eta_0;W)]
$  whose multiply robust moment function is of the form
$$
\psi(\delta,\nu,\alpha,\eta;W)=\nu(W)+\alpha(W)\{Y-\delta(W)\}+\eta(W)\{\delta(W)-\nu(W)\}.
$$
To lighten notation, we write
$
\psi_0(W)=\psi(\delta_0,\nu_0,\alpha_0,\eta_0;W)
$
and define the oracle moments
$$
0=E\{\psi_0(W)-\theta_0\},\; \sigma^2=E[\{\psi_0(W)-\theta_0\}^2],\; \xi^3=E\{|\psi_0(W)-\theta_0|^3\},\; \chi^4=E[\{\psi_0(W)-\theta_0\}^4].
$$

\begin{lemma}[Semiparametric inference; Corollary 6.1 of \cite{singh2021finite}]\label{lemma:inference}
Suppose the following conditions hold:
\begin{enumerate}
\item Neyman orthogonal moment function $\psi$;
    \item bounded residual variances: $E[\{Y-\delta_0(W)\}^2 \mid W ]\leq \bar{\sigma}_1^2$ and
     $E[\{\delta_0(W)-\nu_0(W_1)\}^2 \mid W_1 ]\leq \bar{\sigma}_2^2$, where $W_1\subset W$ is the argument of $\nu_0$;
    \item bounded balancing weights: $\|\alpha_0\|_{\infty}\leq\bar{\alpha}$ and
     $\|\eta_0\|_{\infty}\leq\bar{\eta}$;
    \item censored balancing weight estimators: $\|\hat{\alpha}_{\ell}\|_{\infty}\leq\bar{\alpha}'$ and
      $\|\hat{\eta}_{\ell}\|_{\infty}\leq\bar{\eta}'$.
\end{enumerate}
Next assume the regularity condition on moments
$
\left\{\left(\xi/\sigma\right)^3+\chi^2\right\}n^{-1/2}\rightarrow0.
$
Finally assume the learning rate conditions
\begin{enumerate}
    \item $\left(1+\bar{\eta}/\sigma+\bar{\eta}'/\sigma\right)\{\mathcal{R}(\hat{\nu}_{\ell})\}^{1/2}\rightarrow 0$;
    \item $\left(\bar{\alpha}/\sigma+\bar{\alpha}'+\bar{\eta}/\sigma+\bar{\eta}'\right)\{\mathcal{R}(\hat{\delta}_{\ell})\}^{1/2}\rightarrow 0$;
       \item $(\bar{\eta}'+\bar{\sigma}_1)\{\mathcal{R}(\hat{\alpha}_{\ell})\}^{1/2}\rightarrow 0$;
    \item $\bar{\sigma}_2\{\mathcal{R}(\hat{\eta}_{\ell})\}^{1/2}\rightarrow 0$;
    \item $[\{n \mathcal{R}(\hat{\nu}_{\ell}) \mathcal{R}(\hat{\eta}_{\ell})\}^{1/2}]/\sigma \rightarrow0$;
     \item $[\{n \mathcal{R}(\hat{\delta}_{\ell}) \mathcal{R}(\hat{\alpha}_{\ell})\}^{1/2}]/\sigma \rightarrow0$;
      \item $[\{n \mathcal{R}(\hat{\delta}_{\ell}) \mathcal{R}(\hat{\eta}_{\ell})\}^{1/2}]/\sigma \rightarrow0$.
\end{enumerate}
Then
$$
\hat{\theta}=\theta_0=o_p(1),\quad \frac{\sqrt{n}}{\sigma}(\hat{\theta}-\theta_0)\leadsto \mathcal{N}(0,1),\quad \text{\normalfont pr} \left\{\theta_0 \in  \left(\hat{\theta}\pm \varsigma_a\hat{\sigma} n^{-1/2} \right)\right\}\rightarrow 1-a.
$$
Moreover, the finite sample rate of Gaussian approximation can expressed in terms of the learning rates above.
\end{lemma}

We will match symbols to appeal to this result

\subsection{Main results}

Appealing to Proposition~\ref{prop:unif_omega}, Lemma~\ref{theorem:mse}, and Lemma~\ref{lemma:inference}, we now prove inference for (i) mediated effects and (ii) time-varying treatment effects.

\textbf{Mediated effects.}
Recall
\begin{align*}
    \gamma_0(d,m,x)&=E(Y|D=d,M=m,X=x),\\
    \omega_0(d,d';x) &=\int \gamma_0(d',m,x)\mathrm{d}P(m|d,x), \\
    \pi_0(d;x)&=\text{\normalfont pr}(D=d|X=x),\\
    \rho_0(d;m,x)&=\text{\normalfont pr}(D=d|M=m,X=x).\\
\end{align*}
Fix $(d,d')$. Matching symbols with the abstract Gaussian approximation, 
\begin{align*}
\delta_0(W)&=\gamma_0(d',M,X), \\ 
    \nu_0(W)&=\int \gamma_0(d',m,X)\mathrm{d}P(m|d,X)=\omega_0(d,d';X), \\
    \alpha_0(W)&=\frac{1_{D=d'}}{\rho_0(d';M,X) } \frac{\rho_0(d;M,X)}{\pi_0(d;X)},  \\
    \eta_0(W)&=\frac{1_{D=d}}{\pi_0(d;X)}. 
\end{align*}

\begin{proof}[Proof of Theorem~\ref{theorem:inference_mediation}]
We proceed in steps, verifying the conditions of Lemma~\ref{lemma:inference}.
\begin{enumerate}
    \item Neyman orthogonality. By \cite[Proposition A.2]{singh2021finite}, it suffices to show the following four equalities:
\begin{align*}
&E[\nu(W)\{1-\eta_0(W)\}]=0, \\
 &E[\delta(W)\{\eta_0(W)-\alpha_0(W)\}]=0, \\
    &E[\alpha(W)\{Y-\delta_0(W)\}]=0,\\
   &E[\eta(W)\{\delta_0(W)-\nu_0(W)\}]=0.
\end{align*}
We verify each one.
\begin{enumerate}
    \item First, write
    \begin{align*}
        E[\nu(W)\{1-\eta_0(W)\}]
        &=E\left[\omega_0(d,d';X)\left\{1-\frac{1_{D=d}}{\pi_0(d;X)} \right\}\right] \\
        &=E\left[\omega_0(d,d';X)\left\{1-\frac{E(1_{D=d}|X)}{\pi_0(d;X)} \right\}\right] \\
        &=E\left[\omega_0(d,d';X)\left\{1-\frac{\pi_0(d;X)}{\pi_0(d;X)} \right\}\right]  \\
        &=0.
    \end{align*}
    \item Second, write
    \begin{align*}
        &E[\delta(W)\{\eta_0(W)-\alpha_0(W)\}]\\
        &=E\left[\gamma(d',M,X)\left\{\frac{1_{D=d}}{\pi_0(d;X)} -\frac{1_{D=d'}}{\rho_0(d';M,X) } \frac{\rho_0(d;M,X)}{\pi_0(d;X)}\right\}\right] \\
        &=E\left[\gamma(d',M,X)\left\{\frac{E(1_{D=d}|M,X)}{\pi_0(d;X)} -\frac{E(1_{D=d'}|M,X)}{\rho_0(d';M,X) } \frac{\rho_0(d;M,X)}{\pi_0(d;X)}\right\}\right]\\
        &=E\left[\gamma(d',M,X)\left\{\frac{\rho_0(d;M,X)}{\pi_0(d;X)} -\frac{\rho_0(d';M,X)}{\rho_0(d';M,X) } \frac{\rho_0(d;M,X)}{\pi_0(d;X)}\right\}\right] \\
        &=0.
    \end{align*}
    \item Third, write
    \begin{align*}
        &E[\alpha(W)\{Y-\delta_0(W)\}] \\
        &=E\left[\frac{1_{D=d'}}{\rho(d';M,X) } \frac{\rho(d;M,X)}{\pi(d;X)}\left\{Y-\gamma_0(d',M,X)\right\}\right] \\
          &=E\left[\frac{1_{D=d'}}{\rho(d';M,X) } \frac{\rho(d;M,X)}{\pi(d;X)}\left\{\gamma_0(D,M,X)-\gamma_0(d',M,X)\right\}\right] \\
        &=E\left(\frac{1}{\rho(d';M,X) } \frac{\rho(d;M,X)}{\pi(d;X)}E[1_{D=d'}\left\{\gamma_0(D,M,X)-\gamma_0(d',M,X)\right\}|M,X]\right)\!.
    \end{align*}
    Focusing on the inner expectation,
    \begin{align*}
        &E[1_{D=d'}\left\{\gamma_0(D,M,X)-\gamma_0(d',M,X)\right\}|M,X]\\
        &= E[1_{D=d'}\left\{\gamma_0(D,M,X)-\gamma_0(d',M,X)\right\}|D=d',M,X]\text{\normalfont pr}(D=d'|M,X) \\
        &=E[\left\{\gamma_0(d',M,X)-\gamma_0(d',M,X)\right\}|D=d',M,X]\text{\normalfont pr}(D=d'|M,X) \\
        &=0.
    \end{align*}
    \item Finally, write
    \begin{align*}
        &E[\eta(W)\{\delta_0(W)-\nu_0(W)\}]\\
        &=E\left[\frac{1_{D=d}}{\pi(d;X)} \{\gamma_0(d',M,X)-\omega_0(d,d';X)\}\right] \\
        &=E\left(\frac{1}{\pi(d;X)} E[1_{D=d}\{\gamma_0(d',M,X)-\omega_0(d,d';X)\}|X]\right)\!.
    \end{align*}
    Focusing on the inner expectation,
    \begin{align*}
        &E[1_{D=d}\{\gamma_0(d',M,X)-\omega_0(d,d';X)\}|X]\\
        &=E[1_{D=d}\{\gamma_0(d',M,X)-\omega_0(d,d';X)\}|D=d,X]\text{\normalfont pr}(d|X) \\
        &=E[\{\gamma_0(d',M,X)-\omega_0(d,d';X)\}|D=d,X]\text{\normalfont pr}(d|X)\\
        &=0.
    \end{align*}
\end{enumerate}
    \item Residual variances are bounded since $Y$ is bounded.
    \item Balancing weights are bounded from the assumption of bounded propensity scores.
    \item Balancing weight estimators are censored from the assumption of censored propensity score estimators.
    \item The regularity on moments holds by hypothesis.
    \item Individual rate conditions. After bounding various quantities by constants, the rate conditions simplify as follows.
    \begin{enumerate}
        \item For $\{\mathcal{R}(\hat{\nu}_{\ell})\}^{1/2}\rightarrow 0$, by Proposition~\ref{prop:unif_omega}, $\{\mathcal{R}(\hat{\nu}_{\ell})\}^{1/2} \leq r^{ME}_{\omega}(n,b,c,b_1,c_1)$.
    \item For $\{\mathcal{R}(\hat{\delta}_{\ell})\}^{1/2}\rightarrow 0$, write
    \begin{align*}
        \mathcal{R}(\delta)
        &=E[\{\delta(W)-\delta_0(W)\}^2] \\
        &=E[\{\gamma(d',M,X)-\gamma_0(d',M,X)\}^2] \\
        &=E\left[\frac{1_{D=d'}}{\rho_0(d';M,X)}\{\gamma(D,M,X)-\gamma_0(D,M,X)\}^2\right] \\
        &\leq C E[\{\gamma(D,M,X)-\gamma_0(D,M,X)\}^2] \\
        &=C \mathcal{R}(\gamma).
    \end{align*}
    Using this result and Lemma~\ref{theorem:mse}, $\{\mathcal{R}(\hat{\delta}_{\ell})\}^{1/2} \leq C \{\mathcal{R}(\hat{\gamma}_{\ell})\}^{1/2} \leq s^{ME}_{\gamma}(n,b,c)$.
       \item For $\{\mathcal{R}(\hat{\alpha}_{\ell})\}^{1/2}\rightarrow 0$,
       write
       \begin{align*}
           &\mathcal{R}(\alpha) \\
           &=E[\{\alpha(W)-\alpha_0(W)\}^2] \\
           &=E\left[\left\{\frac{1_{D=d'}}{\rho(d';M,X) } \frac{\rho(d;M,X)}{\pi(d;X)} 
           -\frac{1_{D=d'}}{\rho_0(d';M,X) } \frac{\rho_0(d;M,X)}{\pi_0(d;X)} \right\}^2\right] \\
           &\leq C E\left[\left\{ 
           \rho(d;M,X)\rho_0(d';M,X){\pi_0(d;X)} 
           -\rho_0(d;M,X)\rho(d';M,X)\pi(d;X)
           \right\}^2\right]\!.
       \end{align*}
       Focusing on the RHS, and using natural abbreviations
       \begin{align*}
            \rho\rho_0'\pi_0
           -\rho_0\rho'\pi 
           &=\rho\rho_0'\pi_0
           \pm \rho_0\rho_0'\pi_0
           \pm \rho_0\rho'\pi_0
           -\rho_0\rho'\pi\\
           &=(\rho-\rho_0)\rho_0'\pi_0
           +\rho_0(\rho_0'-\rho')\pi_0 
           +\rho_0\rho'(\pi_0-\pi).
       \end{align*}
       Therefore by the triangle inequality,
 \begin{align*}
     \mathcal{R}(\alpha)&\leq C E[\{\rho(d;M,X)-\rho_0(d;M,X)\}^2] + C E[\{\rho(d';M,X)-\rho_0(d';M,X)\}^2]\\
     &\quad + C E[\{\pi(d;X)-\pi_0(d;X)\}^2] \\
     &= C\{\mathcal{R}(\rho)+\mathcal{R}(\pi)\}.
 \end{align*}
       Using this result and Lemma~\ref{theorem:mse},
       $\{\mathcal{R}(\hat{\alpha}_{\ell})\}^{1/2} \leq C\{ \mathcal{R}(\hat{\pi}_{\ell})+ \mathcal{R}(\hat{\rho}_{\ell})\}^{1/2} \leq s^{ME}_{\pi}(n,b_2,c_2)+s^{ME}_{\rho}(n,b_3,c_3)$.
    \item For $\{\mathcal{R}(\hat{\eta}_{\ell})\}^{1/2}\rightarrow 0$, write 
    \begin{align*}
         \mathcal{R}(\eta)
           &=E[\{\eta(W)-\eta_0(W)\}^2] \\
           &=E\left[\left\{\frac{1_{D=d}}{\pi(d;X)}-\frac{1_{D=d}}{\pi_0(d;X)}\right\}^2\right] \\
           &\leq C E\left[\left\{\pi_0(d;X)-\pi(d;X)\right\}^2\right] \\
           &= C \mathcal{R}(\pi).
    \end{align*}
       Using this result and Lemma~\ref{theorem:mse},
    $\{\mathcal{R}(\hat{\eta}_{\ell})\}^{1/2} \leq C \{\mathcal{R}(\hat{\pi}_{\ell})\}^{1/2} \leq s^{ME}_{\pi}(n,b_2,c_2) $.
    \end{enumerate}
    \item Product rate conditions. After bounding various quantities by constants, the rate conditions simplify as follows.
    \begin{enumerate}
        \item For $\{n \mathcal{R}(\hat{\nu}_{\ell}) \mathcal{R}(\hat{\eta}_{\ell})\}^{1/2} \rightarrow0$, as in the argument for the individual rate conditions,
        \begin{align*}
            \{n \mathcal{R}(\hat{\nu}_{\ell}) \mathcal{R}(\hat{\eta}_{\ell})\}^{1/2} &\leq n^{1/2}  r^{ME}_{\omega}(n,b,c,b_1,c_1) s^{ME}_{\pi}(n,b_2,c_2)\\
            &=n^{1/2}  \left(n^{-\frac{1}{2}\frac{c-1}{c+1/b}}+n^{-\frac{1}{2}\frac{c_1-1}{c_1+1/b_1}}\right)  n^{-\frac{1}{2}\frac{c_2}{c_2+1/b_2}}.
        \end{align*}
        A sufficient condition is that
        $$
        1-\frac{c-1}{c+1/b}-\frac{c_2}{c_2+1/b_2}>0,\quad 1-\frac{c_1-1}{c_1+1/b}-\frac{c_2}{c_2+1/b_2}>0,
        $$
        i.e.
        $$
        \min \left(\frac{c-1}{c+1/b},\frac{c_1-1}{c_1+1/b_1}\right)+\frac{c_2}{c_2+1/b_2}>1.
        $$
     \item For $\{n \mathcal{R}(\hat{\delta}_{\ell}) \mathcal{R}(\hat{\alpha}_{\ell})\}^{1/2} \rightarrow0$, as in the argument for the individual rate conditions,
     \begin{align*}
          \{n \mathcal{R}(\hat{\delta}_{\ell}) \mathcal{R}(\hat{\alpha}_{\ell})\}^{1/2}  &\leq n^{1/2}  s^{ME}_{\gamma}(n,b,c) \{s^{ME}_{\pi}(n,b_2,c_2)+s^{ME}_{\rho}(n,b_3,c_3)\} \\
          &=n^{1/2}n^{-\frac{1}{2}\frac{c}{c+1/b}} \left(n^{-\frac{1}{2}\frac{c_2}{c_2+1/b_2}}+n^{-\frac{1}{2}\frac{c_3}{c_3+1/b_3}}\right)\!. 
     \end{align*}
      A sufficient condition is that
        $$
        1-\frac{c}{c+1/b}-\frac{c_2}{c_2+1/b_2}>0,\quad 1-\frac{c}{c+1/b}-\frac{c_3}{c_3+1/b_3}>0,
        $$
        i.e.
        $$
        \min \left(\frac{c_2}{c_2+1/b_2},\frac{c_3}{c_3+1/b_3}\right)+\frac{c}{c+1/b}>1.
        $$
               Since we have assumed bounded kernels, $(b,b_2,b_3)\geq 1$. By correct specification, $(c,c_2,c_3)\geq 1$. Finally, we have already assumed $c>1$. These conditions imply the desired inequality.
      \item For $\{n \mathcal{R}(\hat{\delta}_{\ell}) \mathcal{R}(\hat{\eta}_{\ell})\}^{1/2} \rightarrow0$, as in the argument for the individual rate conditions,
      $$
      \{n \mathcal{R}(\hat{\delta}_{\ell}) \mathcal{R}(\hat{\eta}_{\ell})\}^{1/2} \leq  n^{1/2} s^{ME}_{\gamma}(n,b,c) s^{ME}_{\pi}(n,b_2,c_2),
      $$
      which is dominated by the previous product rate condition.
    \end{enumerate}
    \end{enumerate}
\end{proof}

\textbf{Time-varying treatment effects.}
Recall
\begin{align*}
\gamma_0(d_1,d_2,x_1,x_2)&=E(Y|D_1=d_1,D_2=d_2,X_1=x_1,X_2=x_2),\\
\omega_0(d_1,d_2;x_1)&=\int \gamma_0(d_1,d_2,x_1,x_2)\mathrm{d}P(x_2|d_1,x_1), \\
\pi_0(d_1;x_1)&=\text{\normalfont pr}(D_1=d_1|X_1=x_1),\\
\rho_0(d_2;d_1,x_1,x_2)&=\text{\normalfont pr}(D_2=d_2|D_1=d_1,X_1=x_1,X_2=x_2)
 \mathrm{d}P(x_2|d_1,x_1).
\end{align*}
Fix $(d_1,d_2)$. Matching symbols with the abstract Gaussian approximation, 
\begin{align*}
\delta_0(W)&=\gamma_0(d_1,d_2,X_1,X_2), \\ 
    \nu_0(W)&=\int \gamma_0(d_1,d_2,X_1,x_2)\mathrm{d}P(x_2|d_1,X_1)=\omega_0(d_1,d_2;X_1), \\
    \alpha_0(W)&=\frac{1_{D_1=d_1}}{\pi_0(d_1;X_1)} \frac{1_{D_2=d_2}}{\rho_0(d_2;d_1,X_1,X_2)},  \\
    \eta_0(W)&=\frac{1_{D_1=d_1}}{\pi_0(d_1;X_1)} 
\end{align*}

\begin{proof}[Proof of Theorem~\ref{theorem:inference_planning}]
We proceed in steps, verifying the conditions of Lemma~\ref{lemma:inference}.
\begin{enumerate}
    \item Neyman orthogonality follows from \cite[Proposition 4.1]{singh2021finite}.
    \item Residual variances are bounded since $Y$ is bounded.
    \item Balancing weights are bounded from the assumption of bounded propensity scores.
    \item Balancing weight estimators are censored from the assumption of censored propensity score estimators.
    \item The regularity on moments holds by hypothesis.
    \item Individual rate conditions. After bounding various quantities by constants, the rate conditions simplify as follows.
    \begin{enumerate}
        \item For $\{\mathcal{R}(\hat{\nu}_{\ell})\}^{1/2}\rightarrow 0$, by Proposition~\ref{prop:unif_omega}, $\{\mathcal{R}(\hat{\nu}_{\ell})\}^{1/2} \leq r^{GF}_{\omega}(n,b,c,b_4,c_4)$.
    \item For $\{\mathcal{R}(\hat{\delta}_{\ell})\}^{1/2}\rightarrow 0$, by \cite[Proposition 6.2]{singh2021finite} and Lemma~\ref{theorem:mse}, $\{\mathcal{R}(\hat{\delta}_{\ell})\}^{1/2} \leq C \{\mathcal{R}(\hat{\gamma}_{\ell})\}^{1/2} \leq s^{GF}_{\gamma}(n,b,c)$.
       \item For $\{\mathcal{R}(\hat{\alpha}_{\ell})\}^{1/2}\rightarrow 0$, by \cite[Proposition 6.2]{singh2021finite} and Lemma~\ref{theorem:mse},
       $\{\mathcal{R}(\hat{\alpha}_{\ell})\}^{1/2} \leq C\{ \mathcal{R}(\hat{\pi}_{\ell})+ \mathcal{R}(\hat{\rho}_{\ell})\}^{1/2} \leq s^{GF}_{\pi}(n,b_6,c_6)+s^{GF}_{\rho}(n,b_7,c_7)$.
    \item For $\{\mathcal{R}(\hat{\eta}_{\ell})\}^{1/2}\rightarrow 0$, by \cite[Proposition 6.2]{singh2021finite} and Lemma~\ref{theorem:mse},
    $\{\mathcal{R}(\hat{\eta}_{\ell})\}^{1/2} \leq C \{\mathcal{R}(\hat{\pi}_{\ell})\}^{1/2} \leq s^{GF}_{\pi}(n,b_6,c_6) $.
    \end{enumerate}
    \item Product rate conditions. After bounding various quantities by constants, the rate conditions simplify as follows.
    \begin{enumerate}
        \item For $\{n \mathcal{R}(\hat{\nu}_{\ell}) \mathcal{R}(\hat{\eta}_{\ell})\}^{1/2} \rightarrow0$, as in the argument for the individual rate conditions,
        \begin{align*}
            \{n \mathcal{R}(\hat{\nu}_{\ell}) \mathcal{R}(\hat{\eta}_{\ell})\}^{1/2} &\leq n^{1/2}  r^{GF}_{\omega}(n,b,c,b_4,c_4) s^{GF}_{\pi}(n,b_6,c_6)\\
            &=n^{1/2}  \left(n^{-\frac{1}{2}\frac{c-1}{c+1/b}}+n^{-\frac{1}{2}\frac{c_4-1}{c_4+1/b_4}}\right)  n^{-\frac{1}{2}\frac{c_6}{c_6+1/b_6}}.
        \end{align*}
        A sufficient condition is that
        $$
        1-\frac{c-1}{c+1/b}-\frac{c_6}{c_6+1/b_6}>0,\quad 1-\frac{c_4-1}{c_4+1/b}-\frac{c_6}{c_6+1/b_6}>0,
        $$
        i.e.
        $$
        \min \left(\frac{c-1}{c+1/b},\frac{c_4-1}{c_4+1/b_4}\right)+\frac{c_6}{c_6+1/b_6}>1.
        $$
     \item For $\{n \mathcal{R}(\hat{\delta}_{\ell}) \mathcal{R}(\hat{\alpha}_{\ell})\}^{1/2} \rightarrow0$, as in the argument for the individual rate conditions,
     \begin{align*}
          \{n \mathcal{R}(\hat{\delta}_{\ell}) \mathcal{R}(\hat{\alpha}_{\ell})\}^{1/2}  &\leq n^{1/2}  s^{GF}_{\gamma}(n,b,c) \{s^{GF}_{\pi}(n,b_6,c_6)+s^{GF}_{\rho}(n,b_7,c_7)\} \\
          &=n^{1/2}n^{-\frac{1}{2}\frac{c}{c+1/b}} \left(n^{-\frac{1}{2}\frac{c_6}{c_6+1/b_6}}+n^{-\frac{1}{2}\frac{c_7}{c_7+1/b_7}}\right)\!. 
     \end{align*}
      A sufficient condition is that
        $$
        1-\frac{c}{c+1/b}-\frac{c_6}{c_6+1/b_6}>0,\quad 1-\frac{c}{c+1/b}-\frac{c_7}{c_7+1/b_7}>0,
        $$
        i.e.
        $$
        \min \left(\frac{c_6}{c_6+1/b_6},\frac{c_7}{c_7+1/b_7}\right)+\frac{c}{c+1/b}>1.
        $$
          Since we have assumed bounded kernels, $(b,b_6,b_7)\geq 1$. By correct specification, $(c,c_6,c_7)\geq 1$. Finally, we have already assumed $c>1$. These conditions imply the desired inequality.
      \item For $\{n \mathcal{R}(\hat{\delta}_{\ell}) \mathcal{R}(\hat{\eta}_{\ell})\}^{1/2} \rightarrow0$, as in the argument for the individual rate conditions,
      $$
      \{n \mathcal{R}(\hat{\delta}_{\ell}) \mathcal{R}(\hat{\eta}_{\ell})\}^{1/2} \leq  n^{1/2} s^{GF}_{\gamma}(n,b,c) s^{GF}_{\pi}(n,b_6,c_6),
      $$
      which is dominated by the previous product rate condition.
    \end{enumerate}
    \end{enumerate}
\end{proof}
\section{Tuning}\label{sec:tuning}




The same two kinds of hyperparameters that arise in kernel ridge regressions arise in our estimators: ridge regression penalties and kernel hyperparameters. In this supplement, we describe practical tuning procedures. To simplify the discussion, we focus on the regression of $Y$ on $W$.

\subsection{Ridge penalty}

It is convenient to tune $\lambda$ by leave-one-out cross validation (LOOCV) or generalized cross validation (GCV), since the validation losses have closed form solutions. The latter is asymptotically optimal in the $\mathbb{L}^2$ sense \cite{craven1978smoothing,li1986asymptotic}. In our main results, we require theoretical regularization that is optimal for the $\mathbb{L}^2$ norm, which is also optimal in RKHS norm \cite{fischer2017sobolev}. As such, GCV should lead to the theoretically required regularization for both nonparametric consistency and semiparametric inference. In practice, LOOCV and GCV lead to nearly identical tuning.

\begin{algo}[Ridge penalty tuning by LOOCV; Algorithm S4 of \cite{singh2020kernel}]
Construct the matrices
$$
H_{\lambda}=I-K_{WW}(K_{WW}+n\lambda I)^{-1}\in\mathbb{R}^{n\times n},\quad \tilde{H}_{\lambda}=diag(H_{\lambda})\in\mathbb{R}^{n\times n},
$$
where $\tilde{H}_{\lambda}$ has the same diagonal entries as $H_{\lambda}$ and off diagonal entries of zero. Then set
$$
\lambda^*=\argmin_{\lambda \in\Lambda} \frac{1}{n}\|\tilde{H}_{\lambda}^{-1}H_{\lambda} Y\|_2^2,\quad \Lambda\subset\mathbb{R}.
$$
\end{algo}

\begin{algo}[Ridge penalty tuning by GCV; Algorithm S5 of \cite{singh2020kernel}]
Set
$$
\lambda^*=\argmin_{\lambda \in\Lambda} \frac{1}{n}\|\{\text{\normalfont tr}(H_{\lambda})\}^{-1}\cdot H_{\lambda} Y\|_2^2,\quad \Lambda\subset\mathbb{R}.
$$
\end{algo}

\subsection{Kernel}

The exponentiated quadratic kernel is widely used among machine learning practitioners:
$$
k(w,w')=\exp\left\{-\frac{1}{2}\frac{(w-w')^2_{\mathcal{W}}}{\iota^2}\right\}\!.
$$
Importantly, it satisfies the required properties; it is continuous, bounded, and characteristic.
%
%
Its hyperparameter is called the lengthscale $\iota$. A convenient heuristic is to set the lengthscale equal to the median interpoint distance of $(W_i)$, where the interpoint distance between the observations $i$ and $j$ is $\|W_i-W_j\|_{\mathcal{W}}$. When the input $W$ is multidimensional, we use the kernel obtained as the product of scalar kernels for each input dimension. For example, if $\mathcal{W}\subset \mathbb{R}^d$ then
$$
k(w,w')=\prod_{j=1}^d \exp\left\{-\frac{1}{2}\frac{(w_j-w_j')^2}{\iota_j^2}\right\}\!.
$$
Each lengthscale $\iota_j$ is set according to the median interpoint distance for that input dimension.

\section{Simulation details}\label{sec:simulation_details}

In this supplement, we provide additional details for various simulations: mediated responses, mediated effects, time-varying responses, and time-varying effects.

\subsection{Nonparametric mediated response}

A single observation is a tuple $(Y,D,M,X)$ for the outcome, treatment, mediator, and covariates where $Y,D,M,X\in\mathbb{R}$. A single observation is generated is as follows. Draw unobserved noise as $u,v,w \overset{i.i.d.}{\sim}\mathcal{U}(-2,2)$. Draw the covariate as $X\sim \mathcal{U}(-1.5,1.5)$. Then set
\begin{align*}
    D&=0.3X+w,\\
    M&=0.3D+0.3X+v, \\
    Y&=0.3D+0.3M+0.5 DM +0.3X+0.25 D^3+u.
\end{align*}
\cite{huber2018direct} also present a simpler version of this design.

We implement our nonparametric estimator  $\hat{\theta}^{ME}(d,d')$ (\texttt{RKHS}, white) described in Section~4, with the tuning procedure described in Supplement~\ref{sec:tuning}. Specifically, we use ridge penalties determined by leave-one-out cross validation, and the product exponentiated quadratic kernel with lengthscales set by the median heuristic. We implement \cite{huber2018direct} (\texttt{IPW}, checkered gray) using the default settings of the command \texttt{medweightcont} in the \texttt{R} package \texttt{causalweight}.

\subsection{Semiparametric mediated effect}

A single observation is a tuple $(Y,D,M,X)$ for the outcome, treatment, mediator, and covariates where $Y,M,X\in\mathbb{R}$ and $D\in\{0,1\}$. A single observation is generated is as follows. Draw unobserved noise as $u,v,w \overset{i.i.d.}{\sim}\mathcal{U}(-2,2)$. Draw the covariate as $X\sim \mathcal{U}(-1.5,1.5)$. Then set
\begin{align*}
    D&=\1(0.3X+w>0),\\
    M&=0.3D+0.3X+v, \\
    Y&=0.3D+0.3M+0.5 DM +0.3X+0.25 D^3+u.
\end{align*}
Note that
\begin{align*}
    E\{M^{(d)}\}&=0.3d, \quad
    E[Y^{\{d',M^{(d)}\}}]=0.55d'+0.09d+0.15d'd.
\end{align*}
Hence 
$
\{\theta_0^{ME}(0,0),\theta_0^{ME}(1,0),\theta_0^{ME}(0,1),\theta_0^{ME}(1,1)\}=(0,0.09,0.55,0.79).
$

We implement our semiparametric estimator  $\hat{\theta}^{ME}(d,d')$ (\texttt{RKHS}) described in Supplement~\ref{sec:semi}, with $L=5$ folds and the tuning procedure described in Supplement~\ref{sec:tuning}. Specifically, we use ridge penalties determined by leave-one-out cross validation. For continuous variables, we use the product exponentiated quadratic kernel with lengthscales set by the median heuristic. For discrete variables, we use the indicator kernel.

\subsection{Nonparametric time-varying response}

This design extends the time-fixed dose response design of \cite{colangelo2020double} to a time-varying setting. Set the number of periods to be $T=2$. A single observation is a tuple $(Y,D_{1:2},X_{1:2})$ for the outcome, treatments, and covariates, where $Y,D_t\in\mathbb{R}$ and $X_t\in\mathbb{R}^{p}$. A single observation is generated is as follows. Draw unobserved noise as $\epsilon_t \overset{i.i.d.}{\sim}\mathcal{N}(0,1)$. Define the vector $\beta \in\mathbb{R}^{p}$ by $\beta_{j}= j^{-2}$. Define the matrix $\Sigma\in\mathbb{R}^{p\times p}$ such that $\Sigma_{ii}=1$ and $\Sigma_{ij}=\frac{1}{2}\cdot 1\{|i-j|=1\}$ for $i\neq j$. Then draw $W_t\overset{i.i.d.}{\sim}\mathcal{N}(0,\Sigma)$ and set
\begin{align*}
X_1&=W_1, \\
    D_1&=\Lambda(X_1^{\top}\beta )+0.75\nu_1,\\ 
    X_2&=0.5 (1-D_1) X_1+0.5  W_2, \\
    D_2&=\Lambda(0.5X_1^{\top}\beta + X_2^{\top}\beta - 0.2D_1)+0.75\nu_2,\\ 
    Y&=0.5\{1.2D_1+1.2X_1^{\top}\beta+D_1^2+D_1 (X_1)_1\}+\{1.2D_2+1.2X_2^{\top}\beta+D_2^2+D_2(X_2)_1\}+\epsilon,
\end{align*}
where we use the truncated logistic link function
$$
\Lambda(t)=(0.9-0.1)\frac{\exp(t)}{1+\exp(t)}+0.1.
$$
It follows that $\theta_0^{GF}(d_1,d_2)=0.6d_1+0.5d_1^2+1.2d_2+d_2^2$.

Clinical intuition guides simulation design. At $t=1$, the patient experiences the correlated symptoms $X_1$. Based on the symptom index $X_1^\top \beta$ as well as unobserved considerations $\nu_1$, the doctor administers the first dose $D_1$. At $t=2$, the patient experiences new symptoms $X_2$ which are the average of attenuated original symptoms $(1-D_1)X_1$ and new developments $W_2$. The first dose $D_1$ has an expectation between zero and one, so a higher dose tends to attenuate the original symptoms $X_1$ more. The doctor administers a second dose $D_2$ based on the initial symptom index $X_1^\top \beta$, the new symptom index $X_2^\top \beta$, the initial dose $D_1$, and unobserved considerations $\nu_2$. The patient's health outcome $Y$ depends on the history of symptoms $X_{1:2}$ and treatments $D_{1:2}$ as well as chance $\epsilon$. Specifically, the health outcome depends on the treatment levels quadratically, the symptom indices linearly, and an interaction between treatment level and the primary symptom linearly. More weight is given to the recent history.

\begin{figure}[ht]
\begin{centering}
     \begin{subfigure}[b]{0.3\textwidth}
         \centering
         \includegraphics[width=\textwidth]{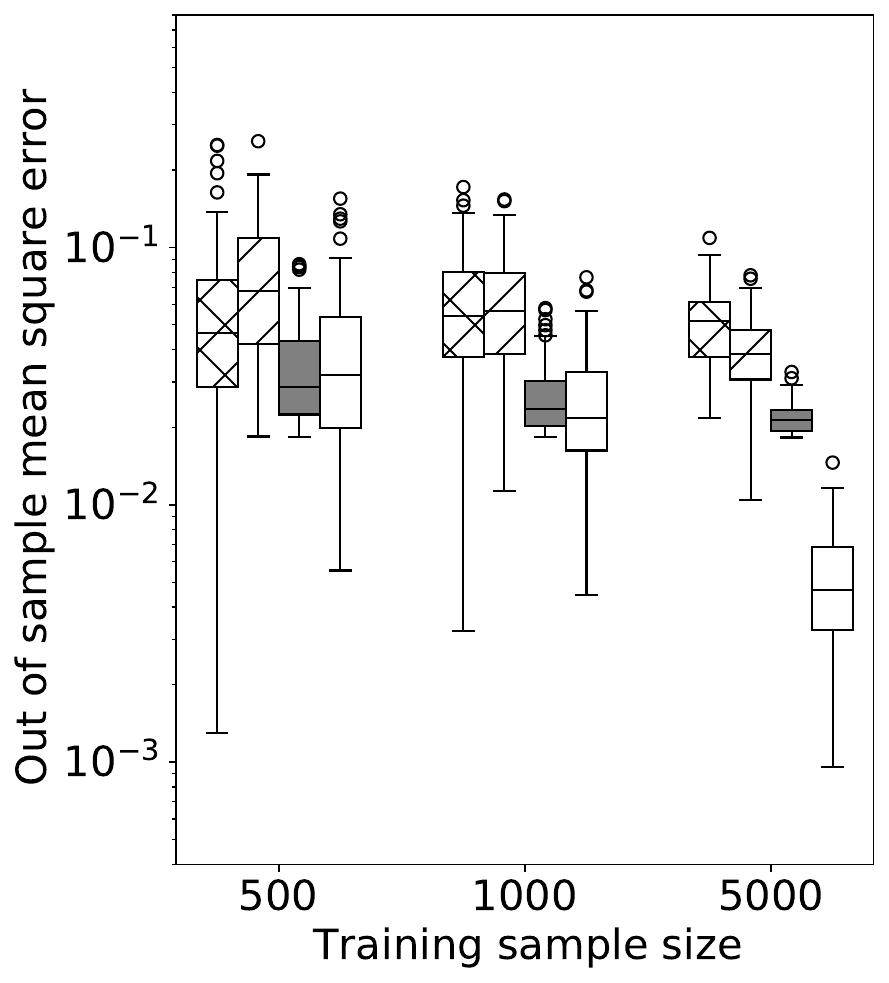}
         \caption{$X_t\in\mathbb{R}$.}
     \end{subfigure}
     \hfill
     \begin{subfigure}[b]{0.3\textwidth}
         \centering
         \includegraphics[width=\textwidth]{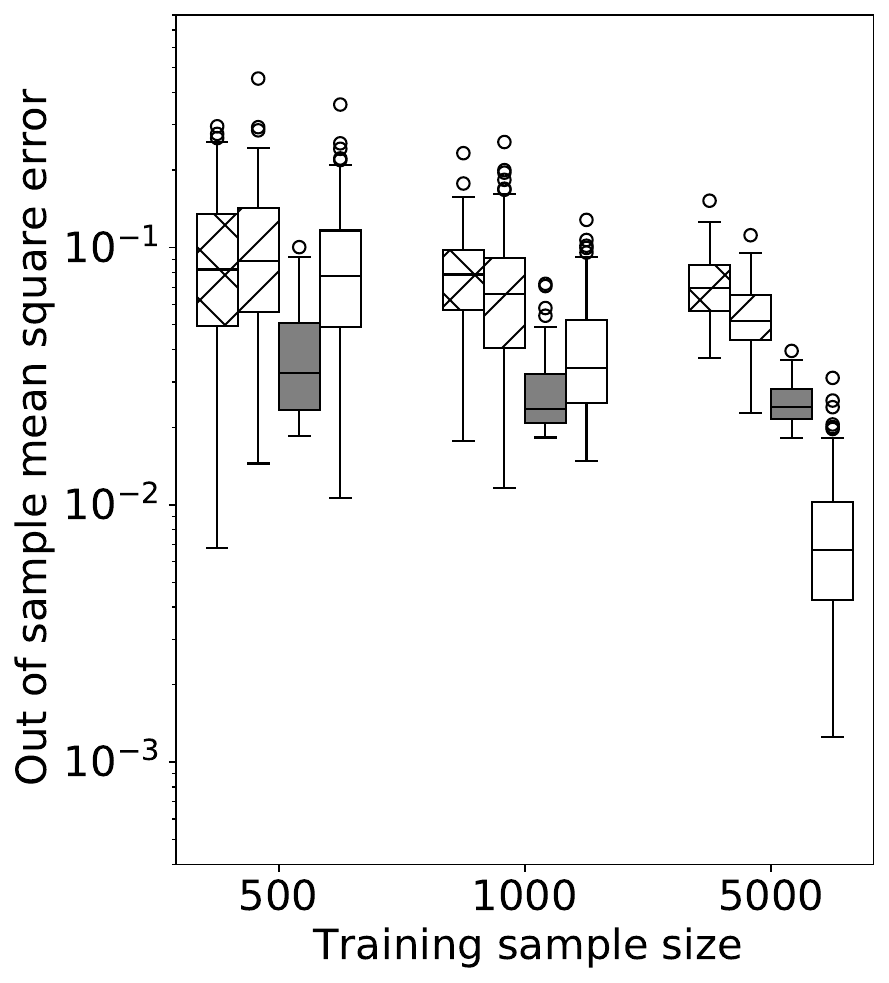}
         \caption{$X_t\in\mathbb{R}^{10}$.}
     \end{subfigure}
     \hfill 
     \begin{subfigure}[b]{0.3\textwidth}
         \centering
         \includegraphics[width=\textwidth]{img/SATE_synthetic_backdoor_dim_100_rep100_bma.pdf}
         \caption{$X_t\in\mathbb{R}^{100}$.}
     \end{subfigure}
\par
\caption{\label{fig:dynamic_sim}
Nonparametric time-varying dose response simulations. We implement four estimators. 
         From left to right, these are 
         \cite{singh2020kernel} \{\texttt{RKHS(ATE)}, checkered white\}, 
         \cite{singh2020kernel} \{\texttt{RKHS(CATE)}, lined white\},
         \cite{lewis2020double} (\texttt{SNMM}, gray), 
         and our own \{\texttt{RKHS(GF)}, white\}.
}
\end{centering}
\end{figure}

We implement our nonparametric estimator $\hat{\theta}^{GF}(d_1,d_2)$ \{\texttt{RKHS(GF)}, white\} described in Section~5, with the tuning procedure described in Supplement~\ref{sec:tuning}. Specifically, we use ridge penalties determined by leave-one-out cross validation, and the product exponentiated quadratic kernel with lengthscales set by the median heuristic. We use the same tuning procedures when implementing the time-fixed dose response \{\texttt{RKHS(ATE)}, checkered white\} and heterogeneous response \{\texttt{RKHS(CATE)}, lined white\} estimators of \cite{singh2020kernel}. Finally, we implement the estimator of \cite{lewis2020double} (\texttt{SNMM}, gray) using the flexible \texttt{heterogeneous} setting in \texttt{Python} code shared by the authors, though their algorithm is designed for linear Markov models without effect modification.

Figure~\ref{fig:dynamic_sim} presents results for different choices of $dim(X_t)\in \{1,10,100\}$ for low, moderate, and high dimensional settings. Across choices of $dim(X_t)$, \texttt{RKHS(GF)} significantly outperforms the alternatives at the sample size 5000. In particular, \texttt{RKHS(GF)} outperforms \texttt{SNMM} with a p value less than $10^{-3}$ by the Wilcoxon rank sum test for $dim(X_t)\in\{1,10,100\}$.

\subsection{Semiparametric time-varying treatment effect}

A single observation is a tuple $(Y,D_{1:2},X_{1:2})$ for the outcome, treatments, and covariates, where $Y\in\mathbb{R}$, $D_t\in\{0,1\}$, and $X_t\in\mathbb{R}^{p}$. A single observation is generated is as follows. Draw unobserved noise as $\epsilon_t \overset{i.i.d.}{\sim}\mathcal{N}(0,1)$. Define the vector $\beta \in\mathbb{R}^{p}$ as the matrix $\Sigma\in\mathbb{R}^{p\times p}$  as before. Then draw $W_t\overset{i.i.d.}{\sim}\mathcal{N}(0,\Sigma)$ and set
\begin{align*}
X_1&=W_1, \\
    D_1&~\sim Bernoulli\{\Lambda(X_1^{\top}\beta)\}, \\ 
    X_2&=0.5 (1-D_1) X_1+0.5  W_2, \\
    D_2&\sim Bernoulli\{\Lambda(0.5X_1^{\top}\beta+ X_2^{\top}\beta-0.2D_1)\}, \\ 
    Y&=0.5\{1.2D_1+1.2X_1^{\top}\beta+D_1^2+D_1 (X_1)_1\}+\{1.2D_2+1.2X_2^{\top}\beta+D_2^2+D_2(X_2)_1\}\\
    &\quad +0.5D_1D_2+\epsilon,
\end{align*}
where we use the truncated logistic link function $\Lambda(t)$ as before.
Note that
$$
E\{Y^{(d_1,d_2)}\}=0.5(1.2d_1+d_1)+(1.2d_2+d_2)+0.5d_1d_2=1.1d_1+2.2d_2+0.5d_1d_2.
$$
Hence
$
\{\theta_0^{GF}(0,0),\theta_0^{GF}(1,0),\theta_0^{GF}(0,1),\theta_0^{GF}(1,1)\}=(0,1.1,2.2,3.8).
$

We implement our semiparametric estimator $\hat{\theta}^{GF}(d_1,d_2)$ \{\texttt{RKHS(GF)}\} described in Supplement~\ref{sec:semi}, with $L=5$ folds and the tuning procedure described in Supplement~\ref{sec:tuning}. Specifically, we use ridge penalties determined by leave-one-out cross validation. For continuous variables, we use the product exponentiated quadratic kernel with lengthscales set by the median heuristic. For discrete variables, we use the indicator kernel.

\section{Application details}\label{sec:application_details}

We provide implementation details and robustness checks for the program evaluation.

\subsection{Mediation analysis}

We implement our nonparametric estimator $\hat{\theta}^{ME}(d,d')$ described in Section~4. We use the tuning procedure described in Supplement~\ref{sec:tuning}. Specifically, we use ridge penalties determined by leave-one-out cross validation. For continuous variables, we use the product exponentiated quadratic kernel with lengthscales set by the median heuristic. For discrete variables, we use the indicator kernel. We use the covariates $X\in\mathbb{R}^{40}$ of \cite{colangelo2020double}.

In the main text, we focus on the $n=2,913$ observations for which $D \geq 40$ and $M>0$, i.e. individuals who completed at least one week of training and who found employment. This choice follows \cite{colangelo2020double,singh2020kernel}. We now verify that our results are robust to the choice of sample. Specifically, we consider the $n=3,906$ observations with $D\geq 40$.

For each sample, we visualize class hours $D$ with a histogram in Figure~\ref{fig:hist}. The class hour distribution in the sample with $D\geq 40$ is similar to the class hour distribution in the sample with $D\geq 40$ and $M>0$ that we use in the main text.

\begin{figure}[ht]
\begin{centering}
     \begin{subfigure}[b]{0.45\textwidth}
         \centering
         \includegraphics[width=\textwidth]{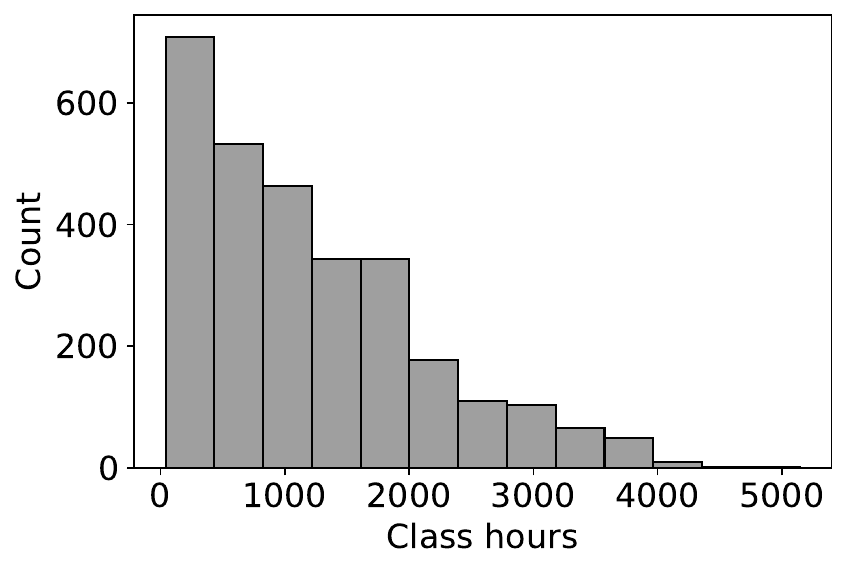}
         \caption{$D\geq 40$ and $M>0$.}
     \end{subfigure}
     \hfill
     \begin{subfigure}[b]{0.45\textwidth}
         \centering
         \includegraphics[width=\textwidth]{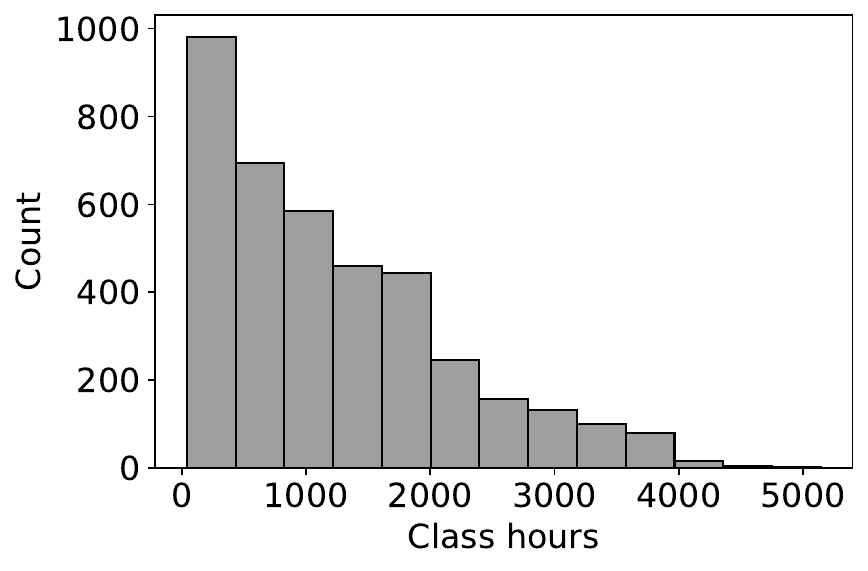}
         \caption{$D\geq 40$.}
     \end{subfigure}
\par
\caption{\label{fig:hist}
Class hours for different samples.}
\end{centering}
\end{figure}

Next, we estimate total, direct, and indirect responses for the new sample choice. Figure~\ref{fig:JC2_d} 
visualizes results. For the sample with $D\geq 40$, the mediated responses are essentially identical to the mediated responses for the sample with $D\geq 40$ and $M>0$ presented in the main text. These results confirm the robustness of the results we present in the main text.

\begin{figure}[ht]
\begin{centering}
     \begin{subfigure}[b]{0.3\textwidth}
         \centering
         \includegraphics[width=\textwidth]{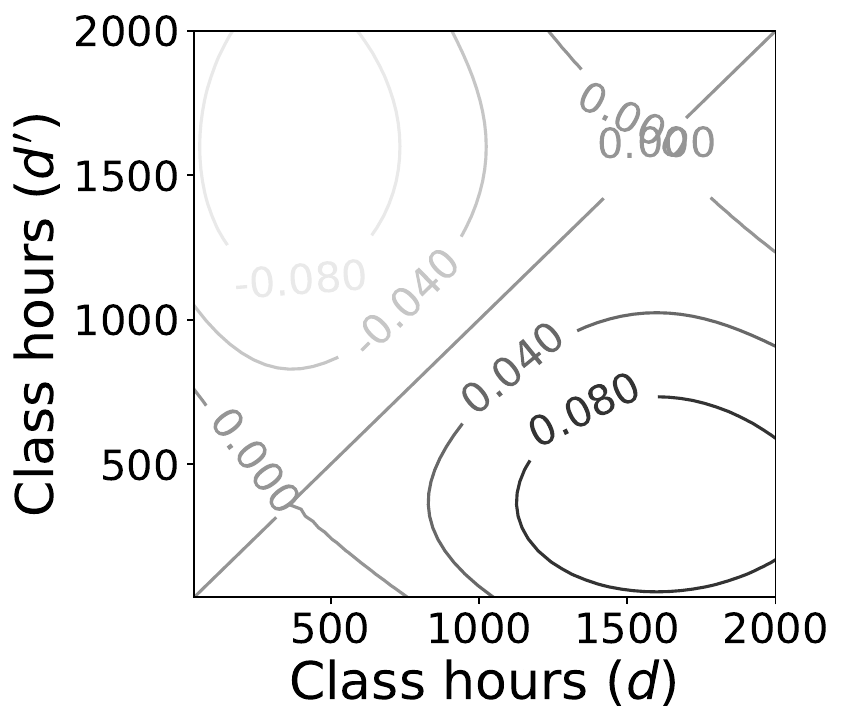}
         \caption{Total response.}
     \end{subfigure}
     \hfill
     \begin{subfigure}[b]{0.3\textwidth}
         \centering
         \includegraphics[width=\textwidth]{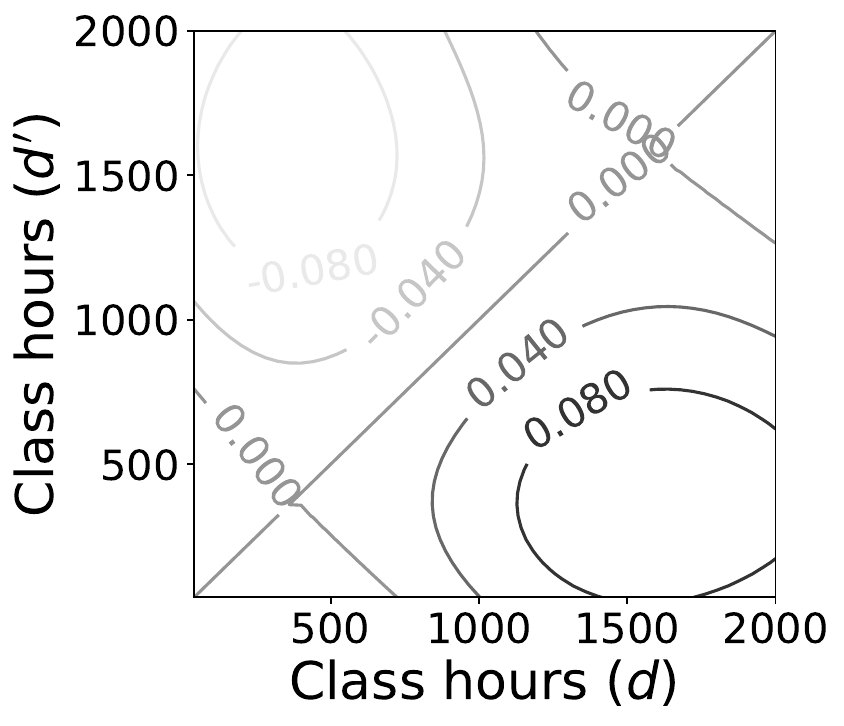}
         \caption{Direct response.}
     \end{subfigure}
     \hfill 
     \begin{subfigure}[b]{0.3\textwidth}
         \centering
         \includegraphics[width=\textwidth]{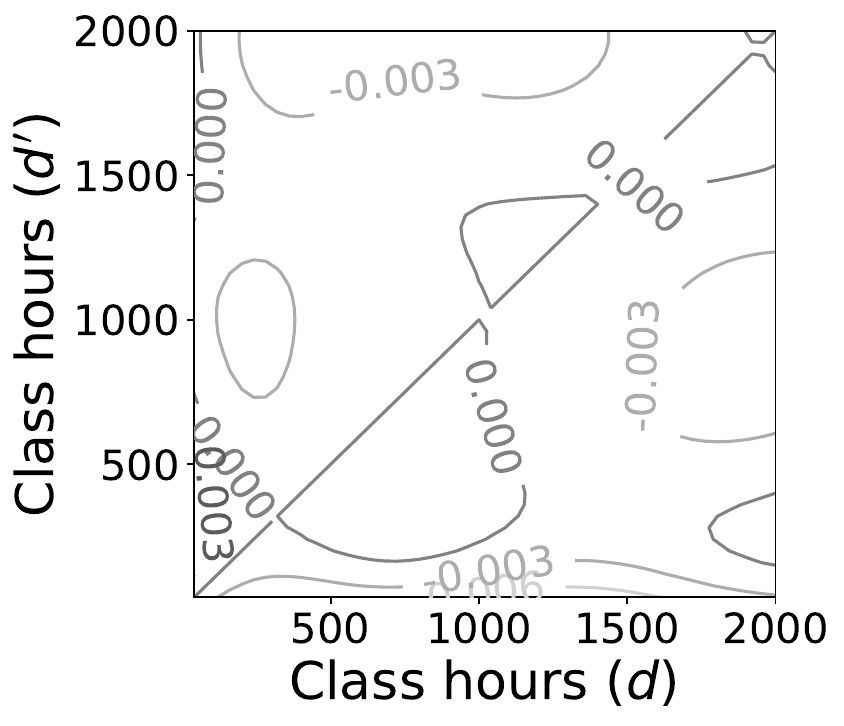}
         \caption{Indirect response.}
     \end{subfigure}
\par
\caption{\label{fig:JC2_d}
Effect of job training on arrests: $D\geq 40$. We implement our estimators for total, direct, and indirect response curves (\texttt{RKHS}, solid).
}
\end{centering}
\end{figure}

\subsection{Time-varying responses}

We implement our nonparametric estimators $\hat{\theta}^{GF}(d_1,d_2)$ and $\theta_0^{GF,\nabla}(d_1,d_2)$ described in Section~5 as well as our semiparametric estimator $\hat{\theta}^{GF}(d_1,d_2)$ described in Supplement~\ref{sec:semi}. We use the tuning procedure described in Supplement~\ref{sec:tuning}. Specifically, we use ridge penalties determined by leave-one-out cross validation. For continuous variables, we use the product exponentiated quadratic kernel with lengthscales set by the median heuristic. For discrete variables, we use the indicator kernel. For the semiparametric case, we truncate extreme propensity scores to lie in the interval $[0.05,0.95]$ consistent with Assumption~\ref{assumption:bounded_propensity_planning}.

We extract the covariates $X_1\in\mathbb{R}^{65}$ from the baseline survey and the covariates $X_2\in\mathbb{R}^{30}$ from the one year follow up interview raw files provided by \cite[Section III.A]{schochet2008does}. The covariates $X_1$ are a superset of $X\in\mathbb{R}^{40}$ chosen by \cite{colangelo2020double}, and they include age, gender, ethnicity, language competency, education, marital status, household size, household income, previous receipt of social aid, family background, health, and health related behavior. We enlarge the set of covariates to include additional variables that may vary from one year to the next, e.g. extra variables concerning health and health related behavior. When a variable never varies over time, e.g. race, we include it in $X_1$ but not $X_2$ so as not to clash with the characteristic property in our RKHS approach. The follow up interview was also less extensive than the baseline interview. These two reasons explain why $X_1\in\mathbb{R}^{65}$ yet $X_2\in\mathbb{R}^{30}$.

In the main text, we focus on the $n=3,141$ observations for which $D_1+D_2 \geq 40$ and $Y>0$, i.e. individuals who completed at least one week of training and who found employment. This choice generalizes our previous rule. We now verify that our results are robust to the choice of sample. Specifically, we consider the $n=3,710$ observations with $D_1+D_2 \geq 40$.

For each sample, we visualize the total class hours $D_1+D_2$ with a histogram in Figure~\ref{fig:hist2}. The class hour distribution in the sample with $D_1+D_2 \geq 40$ is similar to the class hour distribution in the sample with $D_1+D_2\geq 40$ and $Y>0$ that we use in the main text.

\begin{figure}[ht]
\begin{centering}
     \begin{subfigure}[b]{0.45\textwidth}
         \centering
         \includegraphics[width=\textwidth]{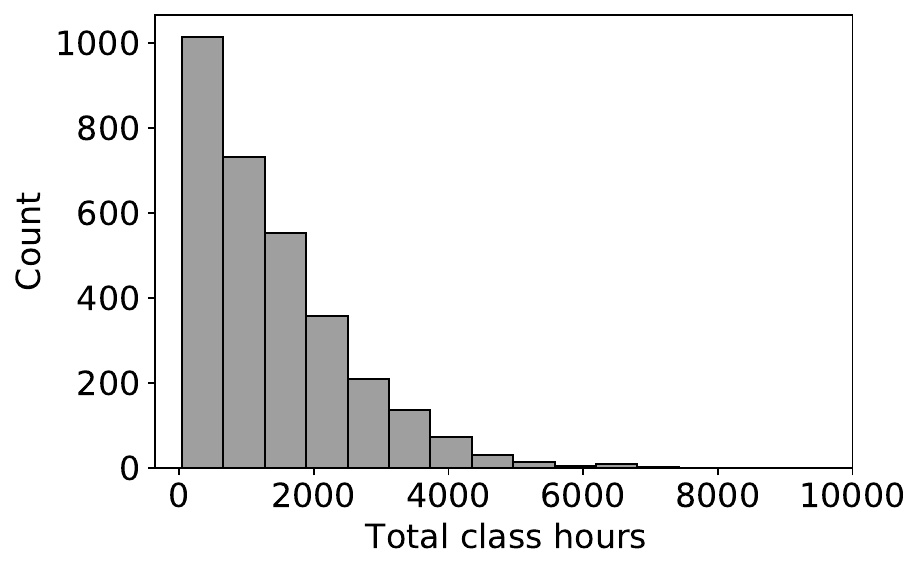}
         \caption{$D_1+D_2\geq 40$ and $Y>0$.}
     \end{subfigure}
     \hfill
     \begin{subfigure}[b]{0.45\textwidth}
         \centering
         \includegraphics[width=\textwidth]{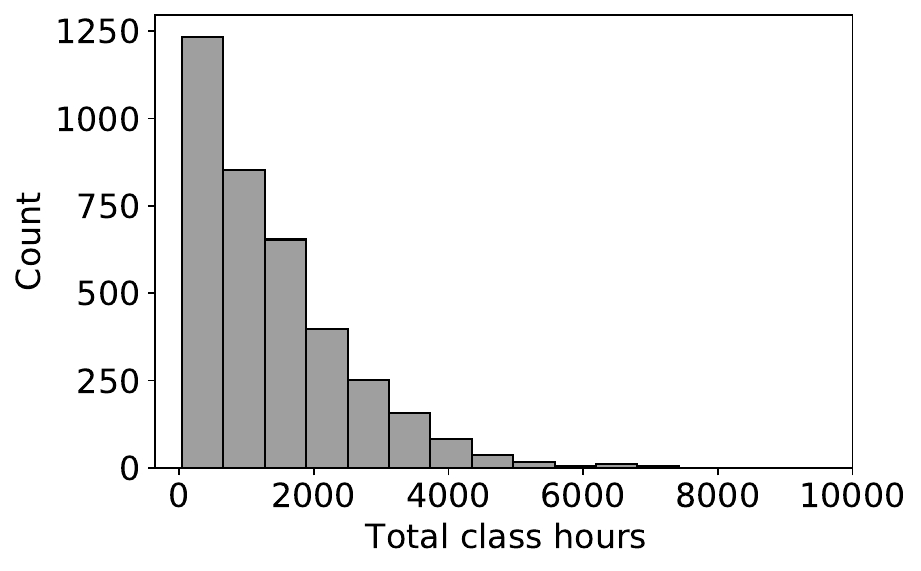}
         \caption{$D_1+D_2\geq 40$.}
     \end{subfigure}
\par
\caption{\label{fig:hist2}
Class hours for different samples}
\end{centering}
\end{figure}

Next, we estimate time-varying responses for the new sample choice. Figure~\ref{fig:JC3_d} 
visualizes results. For the sample with $D_1+D_2\geq 40$, the time-varying responses are similar to the time-varying responses of the sample with $D_1+D_2\geq 40$ and $Y>0$ presented in the main text, albeit with a lower plateau and overall lower levels of counterfactual employment. These results confirm the robustness of the results we present in the main text.

\begin{figure}[ht]
\begin{centering}
     \begin{subfigure}[b]{0.45\textwidth}
         \centering
         \includegraphics[width=\textwidth]{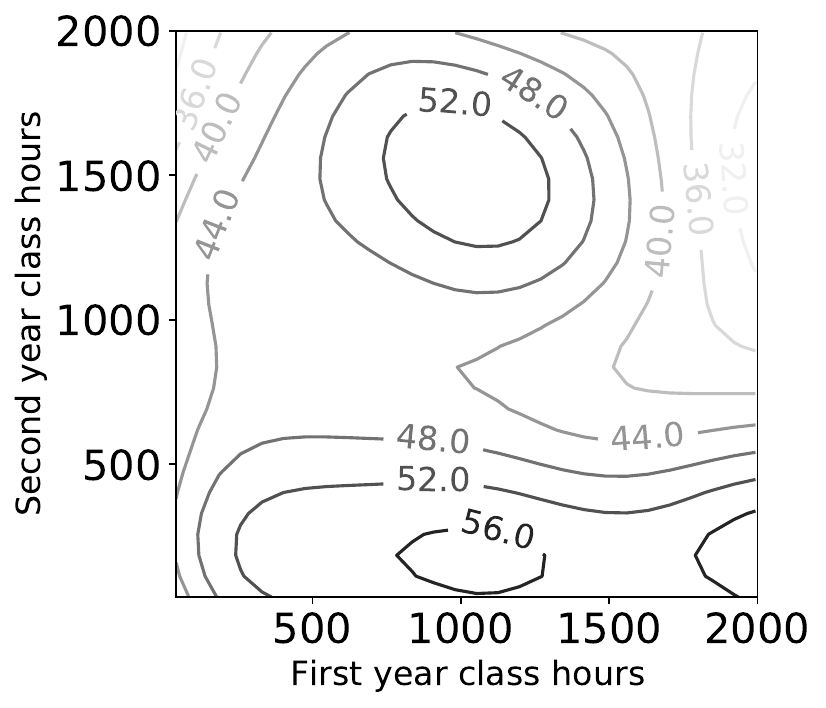}
         \caption{Time-varying dose response.}
     \end{subfigure}
     \hfill
     \begin{subfigure}[b]{0.45\textwidth}
         \centering
         \includegraphics[width=\textwidth]{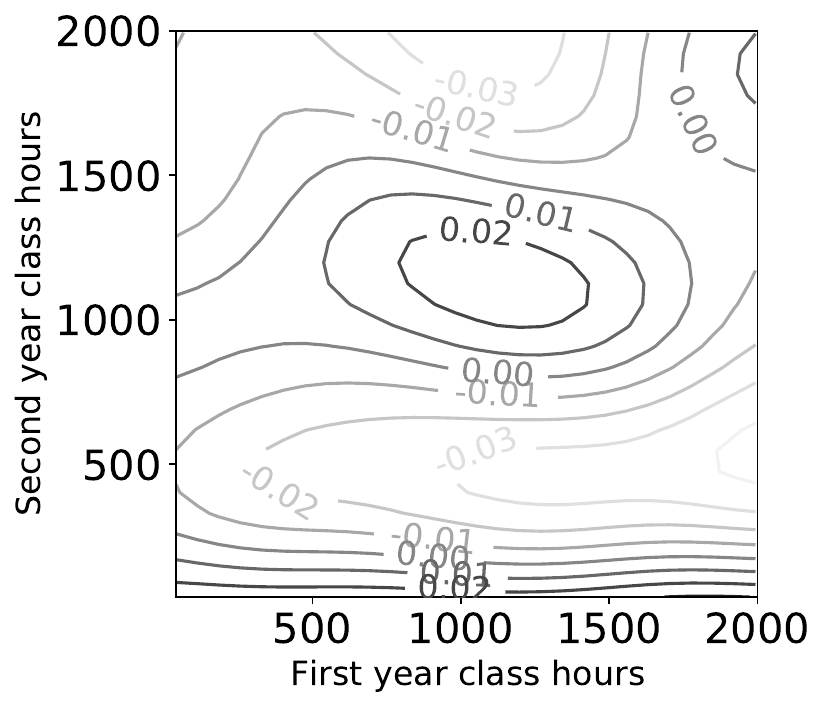}
         \caption{Time-varying incremental response.}
     \end{subfigure}
\par
\caption{
Effect of job training on employment: $D_1+D_2\geq 40$. We implement our estimators for time-varying dose and incremental response curves \{\texttt{RKHS(GF)}\!, solid\}.}
\label{fig:JC3_d}
\end{centering}
\end{figure}

\end{document}